\documentclass[11pt]{article}
\usepackage{amsfonts}
\usepackage{amsmath}
\usepackage{amssymb}
\usepackage{graphicx}
\usepackage{geometry, tabularx}
\providecommand{\U}[1]{\protect\rule{.1in}{.1in}}

\newtheorem{theorem}{Theorem}

\newtheorem{Assumption}[theorem]{\bf Assumption}

\newtheorem{exa}{Example}

\newtheorem{lemma}{Lemma}

\newtheorem{proposition}{Proposition}
\newtheorem{remark}{Remark}

\usepackage{graphicx,color}
\usepackage[bf,SL,BF]{subfigure}

\newenvironment{proof}[1][Proof]{\noindent\textbf{#1.} }{\ \rule{0.5em}{0.5em}}
\normalsize
\numberwithin{equation}{section}
\allowdisplaybreaks

\begin{document}

\title{Optimal loss-carry-forward taxation for L\'{e}vy risk processes stopped \\at general draw-down time}

\author{Wenyuan Wang\thanks{School of Mathematical Sciences, Xiamen University, Fujian, 361005, P.R. China. Email: wwywang@xmu.edu.cn},\ \ \  Zhimin Zhang\thanks{Corresponding Author.\ \ College of Mathematics and Statistics, Chongqing University, Chongqing, 401331, P.R. China. Email: zmzhang@cqu.edu.cn}}

\maketitle
\vspace{-.25in}
\begin{abstract}
Motivated by Kyprianou and Zhou (2009), Wang and Hu (2012), Avram et al. (2017), Li et al. (2017) and Wang and Zhou (2018), we consider in this paper the problem of maximizing the expected accumulated discounted tax payments of an insurance company, whose reserve process (before taxes are deducted) evolves as a spectrally negative L\'{e}vy process with the usual exclusion of negative subordinator or deterministic drift. Tax payments are collected according to the very general loss-carry-forward tax system introduced in Kyprianou and Zhou (2009). To achieve a balance between taxation optimization and solvency, we consider an interesting modified objective function by considering the expected accumulated discounted tax payments of the company until the general draw-down time, instead of until the classical ruin time. The optimal tax return function together with the optimal tax strategy is derived, and some numerical examples are also provided.
\end{abstract}


\section{Introduction}

Albrecher and Hipp (2007) firstly introduced the so-called loss-carry-forward tax into the classical
compound Poisson risk model. In their model, taxation is imposed at a constant proportional rate $\gamma\in(0,1)$
whenever the surplus process is in its running supremum (and, hence, in a profitable situation). The authors established a remarkably simple relationship between the ruin probabilities of the
surplus processes with and without tax; obtained the solution for the expected accumulated discounted tax payments; and characterized the surplus threshold $M\in(0,\infty)$ for starting taxation such that the expected accumulated
discounted tax payments is maximized. Here, we say that $M$ is a surplus threshold for starting taxation, if tax is collected by the tax authority only when the surplus has exceeded $M$ and is in a profitable situation at the same time.

During the past decade, there has been much progress in the study of loss-carry-forward tax when the
underlying surplus processes are the spectrally negative L\'{e}vy processes, the time-homogeneous diffusion processes and the Markove additive processes. In the following, we will summarize the vast loss-carry-forward tax concerned literatures which, classified by the problems addressed in these literatures, consists of four components: (1) Gerber-Shiu function; (2) Distribution of the accumulated discounted tax payments; (3) Maximizing the expected accumulated discounted tax payments by delaying starting taxation until the surplus exceeds a critical threshold level; and (4) Finding the optimal tax strategy to maximize the expected accumulated discounted tax payments.

With respect to (1), the study concerning loss-carry-forward tax has undergone an impressive metamorphosis. The earliest work of Albrecher and Hipp (2007) found a simple relationship (or, called, tax identity) between the ruin probabilities under the classical compound Poisson risk models with and without constant tax rate. By linking queueing concepts with risk theory, another simple and insightful proof for the tax identity was provided in Albrecher et al. (2009). The tax identity was then extended to the classical compound Poisson risk processes with constant credit interest rate and surplus-dependent tax rate in Wei (2009); to the spectrally negative L\'{e}vy risk processes
with constant tax rate in Albrecher et al. (2008); to the time-homogeneous diffusion risk processes with surplus-dependent tax rate in Li et al. (2013); and to the Markove additive risk processes with surplus-dependent tax rate in Albrecher et al. (2014).
Full form of Gerber-Shiu functions was derived by Wang et al. (2011) in the classical compound Poisson risk model with a constant tax rate; by Ming et al. (2010) in the classical compound Poisson risk model with a constant tax, credit interest and debit interest rate; by Cheung and Landriault (2012) in the classical compound Poisson risk models with surplus-dependent premium and tax rate; by Wei et al. (2010) in the Markov-modulated risk models with constant tax rate; and by Kyprianou and Zhou (2009) in the L\'{e}vy risk models with surplus-dependent tax rate.

%Solutions for the two-side exit problem can be found in Ming et al. (2017) under the the classical compound Poisson risk process with a constant interest and tax rate; in Cheung and Landriault (2012) under the the classical compound Poisson risk process with surplus-dependent premium and tax rate; in Albrecher et al. (2008a)  under the spectrally negative L\'{e}vy risk processes with constant tax rate, in Kyprianou and Zhou (2009) and Wang and Ming (2018) under the spectrally negative L\'{e}vy risk processes with surplus-dependent tax rate; in Li et al. (2013) under the time-homogeneous diffusion risk processes with surplus-dependent tax rate; and in Albrecher et al. (2014) under the Markove additive risk processes with surplus-dependent tax rate.

One can find the earliest work for (2) and (3) in Albrecher and Hipp (2007) in the classical compound Poisson risk model with a constant tax rate. These results were then generalized by Wang et al. (2010) to the the classical compound Poisson risk processes with constant interest and tax rate; by Cheung and Landriault (2012) to the classical compound Poisson risk model with surplus-dependent premium and tax rate; by Albrecher et al. (2008a) to the spectrally negative L\'{e}vy risk processes with constant tax rate. Results for (2) can also be found in Li et al. (2013) under the time-homogeneous diffusion risk processes with surplus-dependent tax rate. In addition, under the spectrally negative L\'{e}vy risk processes with surplus-dependent tax rate, Kyprianou and Zhou (2009) obtained solution for Problem (2); while Renaud (2009) obtained arbitrary moments of the accumulated discounted tax payments.

For the dual model of the classical compound Poisson risk process, Albrecher et al. (2008b) derived results of (1-3) under the assumption of exponential jump sizes. The periodic tax has also been taken into consideration under the spectrally negative L\'{e}vy risk processes. In Hao and Tang (2009), asymptotic formulas for the ruin probability with periodic tax is discussed under suitable assumptions imposed on the L\'{e}vy measure. While in Zhang et al. (2017), periodic tax is investigated and solutions for (1-3) are given. In addition, capital injection has been included in risk processes with loss-carry-forward tax in Albrecheer and Ivanovs (2014), where power identities are obtained. For works concerning the two-side exit problem for risk models with loss-carry-forward tax, the readers are referred to Cheung and Landriault (2012), Albrecher et al. (2008a), Kyprianou and Zhou (2009), Wang and Ming (2018), Li et al. (2013), Albrecher et al. (2014) and Albrecher et al. (2011).

It is known that the work of De Finetti (1957), where a critical surplus threshold for starting paying dividends was determined to maximize the expected accumulated discounted dividends, aroused great research interests in optimizing the dividend payment strategy. Likewise, the studies on (3) aroused research interests in optimizing the loss-carry-forward tax strategy, which lies in the scope of (4).
Recently, under the spectrally negative L\'{e}vy risk processes, Wang and Hu (2012) formulated (4) in the setup of stochastic control and characterized the optimal
tax strategy which maximizes the expected accumulated discounted tax payments until ruin, via the standard approach of stochastic control theory. The optimal tax return function was also obtained.

Much lately, Avram et al. (2017) considered the variants of (1-3) through replacing the ruin time by
the draw-down time with a linear draw-down function in the spectrally negative L\'{e}vy risk processes with constant tax rate. The subsequent work of Li et al. (2017) proved several results involving the general draw-down time from the running maximum for the spectrally negative L\'{e}vy process. Using a very delicate approximating method, the Laplace transforms for two-side exit problems involving the general draw-down time, the hitting time and creeping time over a maximum related general draw-down level, as well as an associated potential measure are all characterized in terms of scale functions.
Quite recently, Wang and Zhou (2018) considered a general version of de Finetti's
optimal dividend problem in which the ruin time is replaced with the general draw-down time for the spectrally negative L\'{e}vy risk processes. The authors identified a condition under which
the barrier dividend strategy turned out to be optimal among all admissible dividend strategies.
For more results on general draw-down time, we are referred to Pistorius (2007).

Motivated by Kyprianou and Zhou (2009), Wang and Hu (2012), Avram et al. (2017), Li et al. (2017) and Wang and Zhou (2018), the present paper
is also concerned with the problem of optimizing the payment of the loss-carry-forward
tax in the L\'{e}vy setup, under the optimization criterion of maximizing the expected accumulated discounted tax payments until the general draw-down time from the running maximum.  The main goal of this paper is to extend the classical ruin-based tax optimization solution to the general draw-down-based tax optimization solution, with the latter involving a trade-off between taxation optimization and solvency.
Specifically,  we shall search for the optimal tax return function and the optimal tax strategy which maximizes the expected accumulated discounted tax
payments until the general draw-down time, instead of until the classical ruin time as the existing taxation optimization related literatures have done (see, Wang and Hu (2012)).
We mention that this extension makes our optimization problem interesting and practical in the following senses:
\begin{itemize}
	\item[(a)]
Roughly speaking, the general draw-down (see \eqref{definition of draw-down times} for its definition; see the paragraph right below \eqref{definition of draw-down times} for its origin) refers to the first time when a drop in value (say, of a stock price or value of a portfolio) from
a historical peak exceeds a certain surplus-related level. Hence, draw-down is fundamentally useful from the perspective of risk management as it can be applied in measuring and
managing extreme risks. In fact, applications of draw-down can be found in many areas. Draw-down is frequently used by mutual fund managers
and commodity trading advisors as an alternative
measurement for volatility, especially when the downward risks of assets returns are of primary interest (see, Schuhmacher and Eling (2011)).
There are also close ties between draw-down and problems in mathematical finance, in insurance context, and in statistics. To name a few: the draw-down can be useful in pricing and managing of Russian
options or insurance options against a drop in the value of a stock (see, e.g. Asmussen et al. (2014), Avram et al. (2004) and Meilijson (2003)); the equivalence between the joint distributions of ruin-related quantities under the barrier dividend strategy and draw-down related quantities under a dividend-free risk model can be found in Avram et al. (2007) and Loffen (2008); draw-down turned out to be the optimal solution in the cumulative sum statistical procedure in Page (1954).
The general draw-down times also find interesting applications of defining Az\'{e}ma-Yor martingales
to solve the Skorokhod embedding problem (see, Az\'{e}ma and Yor (1979) and Pistorius (2007), etc.). For more detailed review of applications of draw-down, the interested readers are referred to Landriault et al. (2017) and Li (2015).

\item[(b)] By extending the classical ruin to the general draw-down, one can easily adjust the draw-down function so that the surplus levels remain positive at the terminal draw-down time with positive probability. It thus provides an interesting  alternative  of optimal taxation problem that achieves a balance between taxation optimization and solvency (see, Wang and Zhou (2018)).

    To be in more detail, let us look at Definition \eqref{definition of draw-down times} endowed with positive-valued $\xi$, whenever the surplus process hits a new record high $x$, the surplus process is allowed to deviate down from $x$ with deviation magnitude no more than $x-\xi(x)\in(0,x)$, otherwise general draw-down occurs.
    In contrast, the classical ruin occurs when the surplus process deviates down with deviation magnitude no less than $x$, a deviation magnitude large enough to lead to general draw-down.
    Intuitively speaking, if adopted by the insurance company as risk assessment tool, the general draw-down only allows for "moderate" downward deviation magnitude of the surplus process, while the classical ruin allows for "large" downward deviation magnitude of the surplus process. Hence, compared with the classical ruin, the general draw-down turns out to be the more favourable risk assessment tool from the practical view point of "protecting" the insurance company.

    More importantly, when $\xi$ is chosen to be positive-valued, there is a positive probability that the surplus remains positive at general draw-down time. Take the L\'{e}vy risk processes with nontrivial Gaussian part for example, there should be positive probability that the surplus hits the general draw-down level $\xi(y)\in(0,y)$ at general draw-down time, where $y\in(0,\infty)$ denotes the running maximum until the general draw-down time.

    If $\xi$ is additionally assumed such that $\xi$ and $x-\xi(x)$ are both increasing, as its running supremum grows larger, the surplus process is allowed to deviate down more far away from the running supremum, meanwhile the surplus remains at higher level at general draw-down time, leaving the insurance company more flexible disposal choices of risks. This seems to be very close to reality.

    If $\xi$ is chosen to be negative-valued, whenever the surplus process hits a new record high $x$, the company is allowed to continue running its businesses until its surplus process drops down below 0 with a deficit bigger than $|\xi(x)|$, in the latter case the company stops running its businesses and we say "absolute ruin" occurs. We are referred to, for example Avram et al. (2018),  for the definition of "absolute ruin" for the spectrally negative L\'{e}vy risk processes.

    In addition, it is interesting to see that the draw-down provides much more surplus level related information compared to the classical ruin. It is also interesting to see that when our model and tax structure are specified to those in the existing literature,  our results coincide with those existing results (see, Remark \ref{2}).

\item[(c)]
This extension makes our optimization problem much more difficult and complicated.
First, new solutions of the general draw-down-based two-side exit problem and the expected discounted accumulated tax payments (in Section 3), should be obtained in advance so as to push forward the Hamilton-Jacobi-Bellman (HJB) equation and the verification proposition in Section 4.
Second, because the general draw-down function introduces new complexity into the HJB equation, characterizing its solutions becomes more challenging and hence much more endeavors are devoted to guessing and verifying its candidate solutions, which is very critical in addressing optimal control problems since solutions to the HJB equation usually corresponds to the optimal tax return function and the optimal tax strategy.
\end{itemize}

We also mention that, the standard line of stochastic control theory was adopted when addressing our general draw-down-based loss-carry-forward tax optimization problem, just as the dividend optimization literatures (see, Avram et al. (2007), Loeffen (2008) and Wang and Zhou (2018), etc.) have done when addressing their dividend optimization problems. However, the present paper is much different from those dividend optimization literatures in implementing each steps along the standard line of stochastic control theory, due to the nature of the loss-carry-forward tax strategy. For example, we use an alternative martingale approach instead of the It\^{o}'s formula in Proposition \ref{verification pro}.

This paper is structured as follows. In Section 2, after a quick review of some preliminaries of the spectrally negative L\'{e}vy process, we give the mathematical presentation of the problem. The solution of the two-side exit problem and the return function for a given tax strategy are given in Section 3. It is then argued in Section 4 that, if the optimal tax return function is once continuously differentiable, then it satisfies a Hamilton-Jacobi-Bellman (HJB) equation. Conversely, it is proven that a solution to the HJB equation coincides with the optimal tax return function. In Section 5, the solution to the HJB equation is constructed and the optimal strategy is also found, provided that Assumption 1 holds true. In Section 6, for some spectrally negative L\'{e}vy
processes, we construct those appropriate class of draw-down functions which fulfills Assumption 1, and hence optimal tax strategy are known. Finally, some numerical examples are presented in Section 7. \vspace{-.2in}

  %%%%%%%%%%%%%%%%%%%%%%%%%%%%%%%%%%%%%%%%%%%%%%%%%%%%%%%%%%%%%%%%
  \section{Mathematical presentation of the optimization problem}
  %%%%%%%%%%%%%%%%%%%%%%%%%%%%%%%%%%%%%%%%%%%%%%%%%%%%%%%%%%%%%%%%

To give a mathematical formulation of the optimization problem, we start from a one-dimensional L\'{e}vy process $X=\{X(t);t\geq0\}$ defined on $(\Omega, \mathcal{F}, \mathbb{F}=\{\mathcal{F}_t, t\geq 0\}, \mathbb{P})$, a filtered probability space satisfying the usual conditions. Throughout this paper, we assume that $X$ is
 a spectrally negative L\'{e}vy process with the usual exclusion of pure increasing linear drift and the negative of a subordinator. We denote by $\mathbb{P}_x$ the law of $X$ with $X_0=x\geq0$, and let $\mathbb{E}_x$ denote the corresponding expectation operator.
  Let  $\{\bar{X}(t);t\geq0\}$ be the running supremum  process of $X$, where $\bar{X}(t):=\sup\limits_{0\leq s\leq t}X(s)$ for $t\geq0$. We assume that in the case of no control, the surplus process for a company evolves as $\{X(t);t\geq0\}$.

A loss-carry-forward tax  strategy
%$\gamma$
is described by a one-dimensional stochastic process $\{\gamma\left(\bar{X}(t)\right),t\geq0\}$, where $\gamma:[0,\infty)\rightarrow[\gamma_{1},\gamma_{2}]$ is a measurable function with $\gamma_{1}$ and $\gamma_{2}$ being constants such that $0\leq\gamma_{1}\leq\gamma_{2}\leq1$. Since there is a one-to-one correspondence between the function $\gamma$ and the tax strategy
$\{\gamma\left(\bar{X}(t)\right),t\geq0\}$, we will also denote the corresponding loss-carry-forward tax strategy by $\gamma$ for short. For the case of an insurance company, $\gamma\left(\bar{X}(t)\right)$ denotes the fraction of the insurer's income that is paid out as tax at time $t$ if it is in a profitable situation. Here we say that the company is in a profitable situation at time $t$ if we have $X(t)=\bar{X}(t)$. Thus, when applying strategy  $\gamma$, the cumulative tax until time $t$ is
\begin{eqnarray}\label{acc.tax.paym.}
\int_{0}^{t}\gamma\left(\bar{X}(s)\right)\mathrm{d}\bar{X}(s),
\end{eqnarray}
and the controlled surplus process is given by
\begin{eqnarray}
U^{\gamma}(t)=X(t)-\int_{0}^{t}\gamma\left(\bar{X}(s)\right)\mathrm{d}\bar{X}(s).
\end{eqnarray}

The strategy $\gamma$ is said to be admissible if the process $\{\gamma\left(\bar{X}(t)\right),t\geq0\}$ is adapted to the filtration $\{\mathcal{F}_{t};t\geq0\}$ and $\gamma_{1}\leq\gamma\left(\bar{X}(t)\right)\leq\gamma_{2}$. By $\Gamma$ we denote the set of
all measurable functions
% one-dimensional stochastic processes of the form $\{\gamma\left(\bar{X}(t)\right),t\geq0\}$ with
$\gamma:[0,\infty)\rightarrow[\gamma_{1},\gamma_{2}]$, i.e., $\Gamma$ consists of all admissible tax strategies in the sense of equating a tax strategy with its corresponding function $\gamma$.
For a given admissible strategy $\gamma\in\Gamma$, we define the tax return function $f_{\gamma}$ by
\begin{eqnarray}
f_{\gamma}(x)=\mathbb{E}_{x}\left[\int_{0}^{\tau_{\xi}^{\gamma}}\mathrm{e}^{-qt}\gamma\left(\bar{X}(t)\right)\mathrm{d}\bar{X}(t)\right], \qquad x\in[0,\infty),
\end{eqnarray}
where $q>0$ is a discount factor, and
\begin{eqnarray}\label{definition of draw-down times}
\tau^{\gamma}_{\xi}=\mbox{inf}\left\{t\geq0;U^{\gamma}(t)<\xi\left(\bar{U}^{\gamma}(t)\right)\right\}
\end{eqnarray}
with $\bar{U}^{\gamma}(t):=\sup\limits_{0\leq s\leq t}U^{\gamma}(s)$, is the general draw-down time associated with the general draw-down function $\xi$, for the process $U^{\gamma}$. The function $\xi: [0,\infty)\rightarrow(-\infty,+\infty)$ is called a general draw-down function if it is a continuously differentiable function satisfying $\xi(y)<y$ for all $y\geq0$.

It is interesting to see why the stopping time defined by \eqref{definition of draw-down times} is called the \textbf{general} draw-down time.
In the literature, the reflected process of a Markov process $Y$ at its supremum, defined as $\sup_{s\in[0,t]}Y(s)-Y(t)$, is called the draw-down process of $Y$ (see, Landriault et al. (2017) for example).
Hence, it is natural to name the first time when the magnitude of draw-down process exceeds a given threshold $a>0$, that is
$$\ell_{a}^{+}:=\inf\{t\geq0; \sup_{s\in[0,t]}Y(s)-Y(t)>a\}=\inf\{t\geq0; Y(t)<\sup_{s\in[0,t]}Y(s)-a\},$$
as the draw-down time. By comparing the definitions of $\ell_{a}^{+}$ and $\tau^{\gamma}_{\xi}$, we may observe that $\ell_{a}^{+}$ is a special case of the \textbf{general} draw-down time for the process $Y$ in that its \textbf{general} draw-down function is specialized to $\xi_{0}(x)=x-a$. This explains the origin of the name of the \textbf{general} draw-down time.
In addition, one can find that the classical ruin time is recovered if one chooses
$\xi\equiv0$ in (\ref{definition of draw-down times}), while the draw-down provides much more reserve level related information compared to ruin. In this sense,
draw-down shall be an efficient tool of characterizing extreme risks from the risk management's point of view.

The objectives of this paper are to find the optimal tax return function, which is defined by
\begin{eqnarray}\label{optimal function}
f(x)=\sup\limits_{\gamma\in\Gamma}f_{\gamma}(x)=
\sup\limits_{\gamma\in\Gamma}\mathbb{E}_{x}\left[\int_{0}^{\tau^{\gamma}_{\xi}}\mathrm{e}^{-qt}\gamma\left(\bar{X}(t)\right)\mathrm{d}\bar{X}(t)\right], \qquad x\in[0,\infty),
\end{eqnarray}
and to find an optimal tax strategy $\gamma^{\ast}$ that satisfies $f(x)=f_{\gamma^{\ast}}(x)$ for all $x$.

The Laplace exponent of $X$ is defined by
\begin{eqnarray}
\psi(\theta)=\mbox{log}\mathbb{E}_{x}[\mathrm{e}^{\theta (X(1)-x)}],
\end{eqnarray}
which is known to be finite for at least $\theta\in[0,\infty)$ in which case it is strictly convex and infinitely differentiable.
As in Bertoin (1996), the $q$-scale functions $\{W^{(q)};q\geq0\}$ of $X$ are defined as follows. For each $q\geq0$, $W^{(q)}:\,[0,\infty)\rightarrow[0,\infty)$ is the unique strictly increasing and continuous function with Laplace transform
\begin{eqnarray}
\int_{0}^{\infty}\mathrm{e}^{-\theta x}W^{(q)}(x)\mathrm{d}x=\frac{1}{\psi(\theta)-q},\mbox{ for }\theta>\Phi(q),
\end{eqnarray}
where $\Phi(q)$ is the right inverse of $\psi$, i.e. the largest root of the equation (in $\theta$) $\psi(\theta)=q$. For simplicity, we write $W(x)$ for $W^{(0)}(x)$. In addition, it follows from Zhou (2007) that
\begin{eqnarray}\label{lim-W}
\lim\limits_{x\rightarrow\infty}\frac{(W^{(q)})'(x)}{W^{(q)}(x)}=\Phi(q)>0,\,\,\mbox{ for }q>0.
\end{eqnarray}
For any $x\in\mathbb{R}$ and $\vartheta\geq0$, there exists a well known exponential change of measure that one may perform for
spectrally negative L\'{e}vy processes,
\begin{eqnarray}
\left.\frac{\mathbb{P}_{x}^{\vartheta}}{\mathbb{P}_{x}}\right|_{\mathcal{F}_{t}}=\mathrm{e}^{\vartheta(X(t)-x)-\psi(\vartheta)t}.
\end{eqnarray}
Furthermore, under the probability measures $\mathbb{P}_{x}^{\vartheta}$,  $X$ remains within the class of spectrally negative L\'{e}vy
processes. Here and after, we shall refer to the functions $W_{\vartheta}^{(q)}$ and $W_{\vartheta}$ as the functions that play
the role of the $q$-scale functions and $0$-scale function considered under the measure $\mathbb{P}_{x}^{\vartheta}$.

We  also briefly recall concepts in excursion theory for the reflected process $\{\bar{X}(t)-X(t);t\geq0\}$, and we refer to Bertoin (1996) for more details.
For $x\in\mathbf{R}$, the process $\{L(t):= \bar{X}(t)-x, t\geq0\}$ serves as a local time at $0$ for
the Markov process $\{\bar{X}(t)-X(t);t\geq0\}$ under $\mathbb{P}_{x}$.
Let the corresponding inverse local time be defined as
$$L^{-1}(t):=\inf\{s\geq0\mid L(s)>t\}=\sup\{s\geq0\mid L(s)\leq t\}.$$
Let further $L^{-1}(t-)=\lim\limits_{s\uparrow t}L^{-1}(s)$.
The Poisson point process of excursions indexed by this local time is  denoted by $\{(t, \varepsilon_{t}); t\geq0\}$
$$\varepsilon_{t}(s):=X(L^{-1}(t))-X(L^{-1}(t-)+s), \,\,s\in(0,L^{-1}(t)-L^{-1}(t-)],$$
whenever $L^{-1}(t)-L^{-1}(t-)>0$.
For the case of $L^{-1}(t)-L^{-1}(t-)=0 $, define $\varepsilon_{t}=\Upsilon$ with $\Upsilon$ being an additional isolated point.
Accordingly, we denote a generic excursion as $\varepsilon(\cdot)$
(or, $\varepsilon$ for short) belonging to the space $\mathcal{E}$ of canonical excursions.
The intensity measure of the process $\{(t, \varepsilon_{t}); t\geq0\}$ is given by $\mathrm{d}t\times \mathrm{d}n$ where $n$
is a measure on the space of excursions.
%The lifetime of a canonical excursion $\varepsilon$ is denoted by $\zeta$, and its excursion height is
%denoted by $\overline{\varepsilon}=\sup\limits_{t\in[0,\zeta]}\varepsilon(t)$. The first passage time of a canonical excursion $\varepsilon$ will be defined
%by
%\begin{eqnarray}\label{exc.upcro.b}
%\rho_{b}^{+}:=\inf\{t\in[0,\zeta];\varepsilon(t)>b\},
%\end{eqnarray}
%with the convention $\inf\emptyset=\zeta$.
In particular, $\bar{\varepsilon}=\sup\limits_{s\geq0} \varepsilon(s)$ serves as an example of $n$-measurable
functional of the canonical excursion.
Recalling the definition of $L^{-1}(t)$, we can verify that,
\begin{eqnarray}\label{}
L^{-1}(t-)=\lim_{u\uparrow t}L^{-1}(u)=\inf\{s\geq0\mid L(s)\geq t\}=\sup\{s\geq0\mid L(s)<t\},\nonumber
\end{eqnarray}
and
\begin{eqnarray}\label{equivalent set}
s<L^{-1}(t-)\Leftrightarrow L(s)<t.
\end{eqnarray}

Throughout this paper, to avoid the trivial case that $X$ drifts to $-\infty$ (draw-down is certain), we assume $\psi'(0+)\geq0$.
It is also assumed that each scale function has continuous derivative of second order.

  %%%%%%%%%%%%%%%%%%%%%%%%%%%%%%%%%%%%%%%%%%%%%%%%%%%%%%%%%%%%%%%%
  \section{Two-side exit problem and expected accumulated discounted tax}
  %%%%%%%%%%%%%%%%%%%%%%%%%%%%%%%%%%%%%%%%%%%%%%%%%%%%%%%%%%%%%%%%

In this section, solutions of the two-side exit problem and the expected accumulated discounted tax payments are given, in preparation for motivating the Hamilton-Jacobi-Bellman (HJB) equation (see, Proposition \ref{HJB}) and the verification proposition (see, Proposition \ref{verification pro}) in Section 4. The solutions are
explicitly expressed by the scale functions associated with the driving spectrally negative L\'{e}vy process.

Define the first up-crossing time of level $b$ as follows,
\begin{eqnarray}
\tau^{+}_{b}=\inf\{t\geq0;\,U^{\gamma}(t)>b\},
\end{eqnarray}
with the convention that $\inf\emptyset=\infty$. For $z\geq x$ and $\gamma\in\Gamma$, define the following function $\overline{\gamma}_{x}(z)$ of $z\in[x,\infty)$,
\begin{eqnarray}
\overline{\gamma}_{x}(z):=x+\int_{x}^{z}(1-\gamma(y))\mathrm{d}y, \qquad z\geq x\geq0.
\end{eqnarray}
By Lemma 2.1 in Kyprianou and Zhou (2009), we know that the random times $\{t\geq0: U^{\gamma}\left(t\right) = \bar{U}^{\gamma}\left(t\right)\}$ agree precisely with $\{t\geq0: X\left(t\right) = \bar{X}\left(t\right)\}$.
Hence we have $X\left(\tau^{+}_{x+h}\right)=\bar{X}\left(\tau^{+}_{x+h}\right)$, which together with the definitions of $\overline{\gamma}_{x}(z)$ and $U^{\gamma}$ yields, for $h\geq0$,
\begin{eqnarray}
\label{v12.3}
x+h=U^{\gamma}\left(\tau^{+}_{x+h}\right)\hspace{-0.3cm}&=&\hspace{-0.3cm}\bar{X}\left(\tau^{+}_{x+h}\right)-\int_{0}^{\tau^{+}_{x+h}}\gamma\left(\bar{X}(s)\right)\mathrm{d}\bar{X}(s)
\nonumber\\
\hspace{-0.3cm}&=&\hspace{-0.3cm}x+\int_{0}^{\tau^{+}_{x+h}}\left(1-\gamma\left(\bar{X}(s)\right)\right)\mathrm{d}\bar{X}(s)
\nonumber\\
\hspace{-0.3cm}&=&\hspace{-0.3cm}x+\int_{ x}^{\bar{X}\left(\tau^{+}_{x+h}\right)}\left(1-\gamma\left(y\right)\right)\mathrm{d}y
\nonumber\\
\hspace{-0.3cm}
&=&\hspace{-0.3cm}\overline{\gamma}_{x}\left(\bar{X}\left(\tau^{+}_{x+h}\right)\right), \qquad x\in[0,\infty),
\end{eqnarray}
which together with the fact that $\overline{\gamma}_{x}$ is strictly increasing and continuous implies
\begin{eqnarray}
\bar{X}\left(\tau^{+}_{x+h}\right)=\overline{\gamma}_{x}^{-1}(x+h), \qquad x\in[0,\infty),
\end{eqnarray}
where $\overline{\gamma}_{x}^{-1}$ denotes the well-defined inverse function of $\overline{\gamma}_{x}$.

\vspace{0.2cm}

The following Proposition \ref{exit problem} gives the solution to the two-side exit problem in terms of scale functions via an excursion argument. The excursion theory has also been used in some recent existing literatures, see Li et al. (2017), Avram et al. (2017), Kyprianou and Zhou (2009), and so on.

\begin{proposition}
\label{exit problem}
Given any $x \in(0,a]$, any measurable function $\gamma:[0,\infty)\rightarrow[\gamma_{1},\gamma_{2}]$ with $0\leq\gamma_{1}\leq\gamma_{2}<1$, and any continuous function $\xi(\cdot):\mathbb{R}\rightarrow\mathbb{R}$ with $\xi(y)<y$ for all $y\in\mathbb{R}$, we have,
\begin{eqnarray}\label{exit problem eqaution}
\mathbb{E}_{x}\left[\mathrm{e}^{-q \tau_{a}^{+}}\mathbf{1}_{\{\tau_{a}^{+}<\tau^{\gamma}_{\xi}\}}\right]=\exp\{-\int_{x}^{a}\frac{W^{(q)'}(y-\xi\left(y\right))}
{W^{(q)}(y-\xi\left(y\right))}\frac{1}{1-\gamma\left(\overline{\gamma}_{x}^{-1}(y)\right)}\mathrm{d}y\}, \qquad x\in[0,\infty).
\end{eqnarray}
\end{proposition}

\begin{proof} See the Appendix.
\end{proof}

\vspace{0.3cm}
For a given strategy $\gamma$, the solution of its expected accumulated discounted tax payments (tax return function) is given by the following Proposition \ref{expected discounted tax}.

\begin{proposition}\label{expected discounted tax}
Given any $x \in(0,a]$, any measurable function $\gamma:[0,\infty)\rightarrow[\gamma_{1},\gamma_{2}]$ with $0\leq\gamma_{1}\leq\gamma_{2}<1$, and any continuous function $\xi(\cdot):\mathbb{R}\rightarrow\mathbb{R}$ with $\xi(y)<y$ for all $y\in\mathbb{R}$, we have
\begin{eqnarray}\label{expected discounted tax identity}
\hspace{-0.3cm}&&\hspace{-0.3cm}\mathbb{E}_{x}\left[\int_{0}^{\tau_{a}^{+}\wedge \tau^{\gamma}_{\xi}}\mathrm{e}^{-q t}\gamma\left(\bar{X}(t)\right)\mathrm{d}\left(\bar{X}(t)-x\right)\right]
\nonumber\\
\hspace{-0.3cm}&=&\hspace{-0.3cm}
\int_{0}^{\overline{\gamma}_{x}^{-1}(a)-x}
\exp\{-\int_{0}^{y}\frac{W^{(q)'}\left(\overline{\gamma}_{x}(x+t)-\xi\left(\overline{\gamma}_{x}(x+t)\right)\right)}{W^{(q)}\left(\overline{\gamma}_{x}(x+t)-\xi\left(\overline{\gamma}_{x}(x+t)\right)\right)}\mathrm{d}t\}
\gamma\left(y+x\right)\mathrm{d}y, \qquad x\in[0,\infty).
\end{eqnarray}
\end{proposition}

\begin{proof} See the Appendix.
\end{proof}

\begin{remark}\label{rem1}
Letting  $a\rightarrow \infty$ in (\ref{expected discounted tax identity}) we get for $x\in[0,\infty)$
\begin{eqnarray}\label{expected discounted tax identity2}
\hspace{-0.3cm}&&\hspace{-0.3cm}\mathbb{E}_{x}\left[\int_{0}^{\tau^{\gamma}_{\xi}}\mathrm{e}^{-q t}\gamma\left(\bar{X}(t)\right)\mathrm{d}\left(\bar{X}(t)-x\right)\right]
\nonumber\\
\hspace{-0.3cm}&=&\hspace{-0.3cm}
\int_{0}^{\infty}
\exp\{-\int_{0}^{y}\frac{W^{(q)'}\left(\overline{\gamma}_{x}(x+t)-\xi\left(\overline{\gamma}_{x}(x+t)\right)\right)}{W^{(q)}\left(\overline{\gamma}_{x}(x+t)-\xi\left(\overline{\gamma}_{x}(x+t)\right)\right)}\mathrm{d}t\}
\gamma\left(y+x\right)\mathrm{d}y
\nonumber\\
\hspace{-0.3cm}&=&\hspace{-0.3cm}
\int_{x}^{\infty}
\exp\{-\int_{x}^{y}\frac{W^{(q)'}\left(\overline{\gamma}_{x}(s)-\xi\left(\overline{\gamma}_{x}(s)\right)\right)}{W^{(q)}\left(\overline{\gamma}_{x}(s)-\xi\left(\overline{\gamma}_{x}(s)\right)\right)}\mathrm{d}s\}
\gamma\left(y\right)\mathrm{d}y.
\end{eqnarray}
The above equation gives the solution of the expected total discounted tax payments paid until the general draw-down time.

By the way, since $\frac{W^{(q)'}(x)}{W^{(q)}(x)}\geq \Phi(q)>0$ holds true for all $x>0$ and $q>0$ (see equation \eqref{W'/W}), we can get the following upper bound for the tax return function $f_{\gamma}(x)$,
$$f_{\gamma}(x)\leq \gamma_{2}\int_{x}^{\infty}
\mathrm{e}^{-\Phi(q)(y-x)}\mathrm{d}y=\gamma_{2}\left/\Phi(q)\right., \qquad x\in[0,\infty).$$
Due to the arbitrariness of the above $\gamma$, we can further deduce that,
\begin{eqnarray}\label{upper bound for the optimal tax return function}
f(x)=\sup\limits_{\gamma\in\Gamma}f_{\gamma}(x)\leq\gamma_{2}\left/\Phi(q),\right. \qquad x\in[0,\infty).
\end{eqnarray}
which provides an upper bound for the optimal tax return function.
\end{remark}

  %%%%%%%%%%%%%%%%%%%%%%%%%%%%%%%%%%%%%%%%%%%%%%%%%%%%%%%%%%%%%%%%
  \section{HJB equation and verification arguments}
  %%%%%%%%%%%%%%%%%%%%%%%%%%%%%%%%%%%%%%%%%%%%%%%%%%%%%%%%%%%%%%%%

In this section, the Hamilton-Jacobi-Bellman (HJB) equation and the verification proposition are presented. It turns out that, the optimal tax return function (say, the expected accumulated discounted tax payments of the optimal tax strategy) satisfies the Hamilton-Jacobi-Bellman (HJB) equation; while a solution to the Hamilton-Jacobi-Bellman (HJB) equation is not necessarily the optimal tax return function, and hence a verification proposition is in need to guarantee the solution to the HJB equation is indeed the optimal tax return function.

The following Proposition \ref{HJB} gives the Hamilton-Jacobi-Bellman (HJB) equation satisfied by the optimal tax return function defined by (\ref{optimal function}).

\begin{proposition}\label{HJB}
Assume that $f(x)$ defined by (\ref{optimal function})
is once continuously differentiable over $(0,\infty)$. Then $f(x)$ satisfies the following Hamilton-Jacobi-Bellman
equation,
\begin{eqnarray}\label{HJB equation}
\sup\limits_{\gamma\in[\gamma_{1},\gamma_{2}]}\Big[\frac{\gamma}{1-\gamma}-\frac{1}{1-\gamma} \frac{W^{(q)'}(x-\xi(x))}{W^{(q)}(x-\xi(x))}f(x)+f'(x)\Big]=0, \qquad x\in[0,\infty).
\end{eqnarray}
\end{proposition}

\begin{proof} See the Appendix.
\end{proof}

\vspace{0.5cm}
In the following proposition, we are to argue that a solution to the Hamilton-Jacobi-Bellman
equation (\ref{HJB equation}) serves as the optimal tax return function $f(x)$ given by (\ref{optimal function}).

\begin{proposition}\label{verification pro}\emph{(Verification Proposition)}\quad
Assume that $f(x)$ is a once continuously differentiable solution to (\ref{HJB equation}), and let $\gamma^{*}$
be the function that maximizes the left hand side of (\ref{HJB equation}), i.e.,
\begin{eqnarray}\label{determin opt str}
\gamma^{*}(z)=\mathop{\arg\max}_{\gamma\in[\gamma_{1},\gamma_{2}]}\Big[\frac{\gamma}{1-\gamma}-\frac{1}{1-\gamma}
\frac{W^{(q)'}\left(\overline{\gamma}_{x}(z)-\xi\left(\overline{\gamma}_{x}(z)\right)\right)}
{W^{(q)}\left(\overline{\gamma}_{x}(z)-\xi\left(\overline{\gamma}_{x}(z)\right)\right)}f\left(\overline{\gamma}_{x}(z)\right)+f'\left(\overline{\gamma}_{x}(z)\right)\Big],\qquad z\geq x\geq0.
\end{eqnarray}
Then $\gamma^{*}$ serves as an optimal tax strategy.
\end{proposition}

\begin{proof} See the Appendix.
\end{proof}

  %%%%%%%%%%%%%%%%%%%%%%%%%%%%%%%%%%%%%%%%%%%%%%%%%%%%%%%%%%%%%%%%
  \section{Characterizing the optimal return function and strategy}
  %%%%%%%%%%%%%%%%%%%%%%%%%%%%%%%%%%%%%%%%%%%%%%%%%%%%%%%%%%%%%%%%

As usual in stochastic control problems, we guess the qualitative nature of the solution of the optimal return function and make those assumptions.
Based on those assumptions we find the solution of the HJB equation and later, we need to verify that the obtained solution satisfies the assumptions.
To be more specific, we
\begin{itemize}
\item
guess in advance that "$\frac{W^{(q)'}(x-\xi(x))}{W^{(q)}(x-\xi(x))}f(x)\leq1$ for all $x$" which is equivalent to (\ref{condition1}), to arrive at a candidate solution (\ref{represntation of the optimal function1}) to (\ref{HJB equation}). And then we should prove that (\ref{represntation of the optimal function1}) satisfies "$\frac{W^{(q)'}(x-\xi(x))}{W^{(q)}(x-\xi(x))}f(x)\leq1$ for all $x$", to guarantee that this candidate solution (\ref{represntation of the optimal function1}) is indeed a solution to (\ref{HJB equation}), see Proposition \ref{verify solution1} and its proof.

\item
guess in advance that "$\frac{W^{(q)'}(x-\xi(x))}{W^{(q)}(x-\xi(x))}f(x)\geq1$ for all $x$" which is equivalent to (\ref{condition2}), to arrive at a candidate solution (\ref{represntation of the optimal function2}) to (\ref{HJB equation}). And then we should prove that (\ref{represntation of the optimal function2}) satisfies "$\frac{W^{(q)'}(x-\xi(x))}{W^{(q)}(x-\xi(x))}f(x)\geq1$ for all $x$", to guarantee that this candidate solution (\ref{represntation of the optimal function2}) is indeed a solution to (\ref{HJB equation}), see Proposition \ref{verify solution2} and its proof.

\item
guess in advance that "$\frac{W^{(q)'}(x-\xi(x))}{W^{(q)}(x-\xi(x))}f(x)\geq1$ for all $x\in(0,x_{1})$ and $\frac{W^{(q)'}(x-\xi(x))}{W^{(q)}(x-\xi(x))}f(x)\leq1$ for all $x\in[x_{1},\infty)$", to arrive at a candidate solution (\ref{4.13}) to (\ref{HJB equation}). And then we should prove that (\ref{4.13}) satisfies "$\frac{W^{(q)'}(x-\xi(x))}{W^{(q)}(x-\xi(x))}f(x)\geq1$ for all $x\in(0,x_{1})$ and $\frac{W^{(q)'}(x-\xi(x))}{W^{(q)}(x-\xi(x))}f(x)\leq1$ for all $x\in[x_{1},\infty)$", to guarantee that this candidate solution (\ref{4.13}) is indeed a solution to (\ref{HJB equation}), see case (\emph{v}) in Proposition \ref{verify solution3} and its proof.
Or, vice versa for case (\emph{vi}) in Proposition \ref{verify solution3}.
\end{itemize}

It is already argued in Proposition \ref{verification pro} that the solution of (\ref{HJB equation}) serves as the optimal tax return function of the taxation optimization problem (\ref{optimal function}), the optimal tax strategy is then obtained using such a solution through (\ref{determin opt str}).
In order to solve (\ref{HJB equation}), one may eliminate the operator "$\sup\limits_{\gamma\in[\gamma_{1},\gamma_{2}]}$" in the right hand side of (\ref{HJB equation}) beforehand. To this end, we first guess
that the solution $f(x)$ of (\ref{HJB equation}) satisfies $\frac{W^{(q)'}(x-\xi(x))}{W^{(q)}(x-\xi(x))}f(x)\leq1$ for all $x$, in which case the function of $\gamma$, say, $\frac{\gamma}{1-\gamma}-\frac{1}{1-\gamma} \frac{W^{(q)'}(x-\xi(x))}{W^{(q)}(x-\xi(x))}f(x)$ is nondecreasing with respect to $\gamma$ for all given $x$. Then using the fact that $\frac{\gamma}{1-\gamma}-\frac{1}{1-\gamma} \frac{W^{(q)'}(x-\xi(x))}{W^{(q)}(x-\xi(x))}f(x)$ is a nondecreasing function of $\gamma$, we may rewrite (\ref{HJB equation}) as
\begin{eqnarray}\label{4.1}
0=\gamma_{2}-\frac{W^{(q)'}(x-\xi(x))}{W^{(q)}(x-\xi(x))}f(x)+(1-\gamma_{2})f'(x),\mbox{ for all }x\geq0.
\end{eqnarray}
Instead of solving (\ref{4.1}),  we first turn to the homogeneous differential equation of equation (\ref{4.1}), which can be re-written as,
\begin{eqnarray}
\frac{\mathrm{d}}{\mathrm{d}x}\ln f(x)=\frac{f'(x)}{f(x)}=\frac{1}{1-\gamma_{2}}\frac{W^{(q)'}(x-\xi(x))}{W^{(q)}(x-\xi(x))},\qquad x\geq0.\nonumber
\end{eqnarray}
The solution of the above homogeneous differential equation is,
\begin{eqnarray}
f(x)=C\exp\{\frac{1}{1-\gamma_{2}}\int_{0}^{x}\frac{W^{(q)'}(y-\xi(y))}{W^{(q)}(y-\xi(y))}\mathrm{d}y\},\qquad x\geq0,\nonumber
\end{eqnarray}
with $C$ being some constant. By the standard method of variation of constant, we guess that the solution of (\ref{4.1}) is,
\begin{eqnarray}\label{4.2}
f(x)=C(x)\exp\{\frac{1}{1-\gamma_{2}}\int_{0}^{x}\frac{W^{(q)'}(y-\xi(y))}{W^{(q)}(y-\xi(y))}\mathrm{d}y\},\qquad x\geq0,
\end{eqnarray}
with $C(x)$ being some unknown function to be determined later. Plugging (\ref{4.2}) into (\ref{4.1}) we get,
\begin{eqnarray}
0\hspace{-0.3cm}&=&\hspace{-0.3cm}\gamma_{2}-\frac{W^{(q)'}(x-\xi(x))}{W^{(q)}(x-\xi(x))}C(x)\exp\{\frac{1}{1-\gamma_{2}}\int_{0}^{x}\frac{W^{(q)'}(y-\xi(y))}{W^{(q)}(y-\xi(y))}\mathrm{d}y\}
\nonumber\\
\hspace{-0.3cm}&&\hspace{-0.3cm}+(1-\gamma_{2})\exp\{\frac{1}{1-\gamma_{2}}\int_{0}^{x}\frac{W^{(q)'}(y-\xi(y))}{W^{(q)}(y-\xi(y))}\mathrm{d}y\}
\left(C\,'(x)+\frac{C(x)}{1-\gamma_{2}}\frac{W^{(q)'}(x-\xi(x))}{W^{(q)}(x-\xi(x))}\right),\nonumber
\end{eqnarray}
which can be simplified as follows,
\begin{eqnarray}
C\,'(x)\hspace{-0.3cm}&=&\hspace{-0.3cm}-\frac{\gamma_{2}}{1-\gamma_{2}}\exp\{-\frac{1}{1-\gamma_{2}}\int_{0}^{x}\frac{W^{(q)'}(y-\xi(y))}{W^{(q)}(y-\xi(y))}\mathrm{d}y\}.\nonumber
\end{eqnarray}
The solution of the above differential equation is given by
\begin{eqnarray}\label{4.3}
C(x)\hspace{-0.3cm}&=&\hspace{-0.3cm}C_{0}-\frac{\gamma_{2}}{1-\gamma_{2}}\int_{0}^{x}\exp\{-\frac{1}{1-\gamma_{2}}\int_{0}^{y}\frac{W^{(q)'}(u-\xi(u))}{W^{(q)}(u-\xi(u))}\mathrm{d}u\}\mathrm{d}y,
\end{eqnarray}
where $C_{0}$ is some constant. Combining (\ref{4.3}) and (\ref{4.2}) we know that the solution of (\ref{4.1}) is as follows,
\begin{eqnarray}\label{4.4}
f(x)\hspace{-0.3cm}&=&\hspace{-0.3cm}\left(C_{0}-\frac{\gamma_{2}}{1-\gamma_{2}}\int_{0}^{x}\exp\{-\frac{1}{1-\gamma_{2}}\int_{0}^{y}\frac{W^{(q)'}(u-\xi(u))}{W^{(q)}(u-\xi(u))}\mathrm{d}u\}\mathrm{d}y\right)
\nonumber\\
\hspace{-0.3cm}&&\hspace{-0.3cm}\times\exp\{\frac{1}{1-\gamma_{2}}\int_{0}^{x}\frac{W^{(q)'}(y-\xi(y))}{W^{(q)}(y-\xi(y))}\mathrm{d}y\},\qquad x\geq0.
\end{eqnarray}
It remains to determine the unknown constant $C_{0}$. For this purpose, we first deduce that
\begin{eqnarray}\label{lim.for.det.C01}
\lim_{x\rightarrow\infty}\exp\{\frac{1}{1-\gamma_{2}}\int_{0}^{x}\frac{W^{(q)'}(y-\xi(y))}{W^{(q)}(y-\xi(y))}\mathrm{d}y\}=\infty,
\end{eqnarray}
which together with the fact that $f(\infty)<\infty$ implies that the equation (\ref{4.4}) can be rewritten as,
\begin{eqnarray}\label{represntation of the optimal function1}
f(x)\hspace{-0.3cm}&=&\hspace{-0.3cm}\frac{\gamma_{2}}{1-\gamma_{2}}\int_{x}^{\infty}\exp\{-\frac{1}{1-\gamma_{2}}\int_{0}^{y}\frac{W^{(q)'}(u-\xi(u))}{W^{(q)}(u-\xi(u))}\mathrm{d}u\}\mathrm{d}y
\nonumber\\
\hspace{-0.3cm}&&\hspace{-0.3cm}\times\exp\{\frac{1}{1-\gamma_{2}}\int_{0}^{x}\frac{W^{(q)'}(y-\xi(y))}{W^{(q)}(y-\xi(y))}\mathrm{d}y\}\nonumber\\
\hspace{-0.3cm}&=&\hspace{-0.3cm}\frac{\gamma_{2}}{1-\gamma_{2}}\int_{x}^{\infty}\exp\{-\frac{1}{1-\gamma_{2}}\int_{x}^{y}\frac{W^{(q)'}(u-\xi(u))}{W^{(q)}(u-\xi(u))}\mathrm{d}u\}\mathrm{d}y,\qquad x\geq0.
\end{eqnarray}

We are to prove (\ref{lim.for.det.C01}) in the sequel. It is known that the function $W^{(q)}(x)$ is continuously differentiable, positive valued and strictly
increasing over $(0,\infty)$.
%    It is already assumed that $W^{(0)'}(x)$ is continuous on $(0,\infty)$, which, according to Chan et al. (2011), is equivalent to,
%    $$n(\bar{\varepsilon}=x)=0,\,\,\forall\,x\in(0,\infty),$$
%    which due to the equivalence between the two probability measures $\mathbb{P}_{x}$ and $\mathbb{P}_{x}^{\Phi(q)}$, is further equivalent to,
%    $$n_{\Phi(q)}(\bar{\varepsilon}=x)=0,\,\,\forall\,x\in(0,\infty),$$
%    which leads to the continuity of $W_{\Phi(q)}^{(0)'}(x)$.
By the equality $W^{(q)}(x)=\mathrm{e}^{\Phi(q)x}W_{\Phi(q)}^{(0)}(x)$ one can conclude that the function $W^{(0)'}(x)$ must also be a continuous function over $(0,\infty)$.
%   More precisely, according to Chan et al. (2011) we have the following representation for $W^{(0)'}(x)$,
%   \begin{eqnarray}
%   \frac{W^{(0)'}(x)}{W^{(0)}(x)}=n(\bar{\varepsilon}\geq x)=n(\bar{\varepsilon}>x)>0,\,\forall\,x\in(0,\infty),\nonumber
%   \end{eqnarray}
%   which implies that the function $\frac{W^{(0)'}(x)}{W^{(0)}(x)}$ of $x\in(0,\infty)$ is non-increasing. The above inequality together with the continuously differentiability of
%   $W^{(0)}(x)$ on $(0,\infty)$ also imply that
%   $$\frac{W^{(0)'}(0+)}{W^{(0)}(0+)}=\lim\limits_{x\downarrow0}\frac{W^{(0)'}(x)}{W^{(0)}(x)}>0,$$
%   is well defined. Hence we can extend the definition of the function $\frac{W^{(0)'}(x)}{W^{(0)}(x)}$ to $[0,\infty)$. The extended function will still be denoted by
%   $\frac{W^{(0)'}(x)}{W^{(0)}(x)}$ (with $\frac{W^{(0)'}(0)}{W^{(0)}(0)}=\frac{W^{(0)'}(0+)}{W^{(0)}(0+)}$), which is continuous on $[0,\infty)$.
%   Using the identity $W^{(q)}(x)=\mathrm{e}^{\Phi(q)x}W_{\Phi(q)}^{(0)}(x)$, one can deduce that,
Furthermore, we have
\begin{eqnarray}\label{W'/W}
\frac{W^{(q)'}(x)}{W^{(q)}(x)}\hspace{-0.3cm}&=&\hspace{-0.3cm}\frac{\Phi(q)\mathrm{e}^{\Phi(q)x}W_{\Phi(q)}^{(0)}(x)+\mathrm{e}^{\Phi(q)x}W_{\Phi(q)}^{(0)'}(x)}{\mathrm{e}^{\Phi(q)x}W_{\Phi(q)}^{(0)}(x)}=\Phi(q)+
\frac{W_{\Phi(q)}^{(0)'}(x)}{W_{\Phi(q)}^{(0)}(x)}\nonumber\\
\hspace{-0.3cm}&=&\hspace{-0.3cm}\Phi(q)+n_{\Phi(q)}(\bar{\varepsilon}>x)\nonumber\\
\hspace{-0.3cm}&\geq&\hspace{-0.3cm}\Phi(q)>0,\quad \forall\,x\in(0,\infty),\quad \forall\,q>0.
\end{eqnarray}
%Actually we have $\frac{W^{(q)'}(0+)}{W^{(q)}(0+)}$
Hence, recalling $y-\xi(y)>0$ for all $y\in[0,\infty)$,
%we can extend the definition of the function $\frac{W^{(q)'}(x)}{W^{(q)}(x)}$ to $[0,\infty)$. The extended function will still be denoted by
%$\frac{W^{(q)'}(x)}{W^{(q)}(x)}$ (with $\frac{W^{(q)'}(0)}{W^{(q)}(0)}=\frac{W^{(q)'}(0+)}{W^{(q)}(0+)}\geq \Phi(q)>0$), which is continuous on $[0,\infty)$.
%Now, using the fact that
%$$\frac{W^{(q)'}(x)}{W^{(q)}(x)}\geq \Phi(q)>0,\,\forall\,x\in[0,\infty),\qquad \forall\,q>0,$$
one can find that
%the left hand side of equation (\ref{lim.for.det.C01}) results in
\begin{eqnarray}\label{lim.for.det.C0}
\hspace{-0.3cm}&&\hspace{-0.3cm}\lim_{x\rightarrow\infty}\exp\{\frac{1}{1-\gamma_{2}}\int_{0}^{x}\frac{W^{(q)'}(y-\xi(y))}{W^{(q)}(y-\xi(y))}\mathrm{d}y\}
\geq
\lim_{x\rightarrow\infty}\exp\{\frac{1}{1-\gamma_{2}}\Phi(q)x\}
=\infty,\nonumber
\end{eqnarray}
as is required in (\ref{lim.for.det.C01}).

For the candidate optimal solution given by (\ref{represntation of the optimal function1}), we need to justify that $\frac{W^{(q)'}(x-\xi(x))}{W^{(q)}(x-\xi(x))}f(x)\leq1$ ($\forall\,x>0$) to assure its optimality. By integration by parts, we can rewrite (\ref{represntation of the optimal function1}) as,
\begin{eqnarray}\label{fgamma2}
f(x)\hspace{-0.3cm}&=&\hspace{-0.3cm}
\int_{x}^{\infty}\mathrm{d}\left(-\frac{\exp\{-\frac{\gamma_{2}}{1-\gamma_{2}}\int_{x}^{y}\frac{W^{(q)'}(u-\xi(u))}{W^{(q)}(u-\xi(u))}\mathrm{d}u\}}
{\exp\{\int_{x}^{y}\frac{W^{(q)'}(u-\xi(u))}{W^{(q)}(u-\xi(u))}\mathrm{d}u\}\frac{W^{(q)'}(y-\xi(y))}{W^{(q)}(y-\xi(y))}}\right)\nonumber\\
\hspace{-0.3cm}&&\hspace{-0.3cm}
-\frac{\left[\exp\{\int_{x}^{y}\frac{W^{(q)'}(u-\xi(u))}{W^{(q)}(u-\xi(u))}\mathrm{d}u\}\frac{W^{(q)'}(y-\xi(y))}{W^{(q)}(y-\xi(y))}\right]'\exp\{-\frac{\gamma_{2}}{1-\gamma_{2}}\int_{x}^{y}\frac{W^{(q)'}(u-\xi(u))}{W^{(q)}(u-\xi(u))}\mathrm{d}u\}}
{\left[\exp\{\int_{x}^{y}\frac{W^{(q)'}(u-\xi(u))}{W^{(q)}(u-\xi(u))}\mathrm{d}u\}\frac{W^{(q)'}(y-\xi(y))}{W^{(q)}(y-\xi(y))}\right]^{2}}\mathrm{d}y
\nonumber\\
\hspace{-0.3cm}&=&\hspace{-0.3cm}
\left.-\frac{\exp\{-\frac{\gamma_{2}}{1-\gamma_{2}}\int_{x}^{y}\frac{W^{(q)'}(u-\xi(u))}{W^{(q)}(u-\xi(u))}\mathrm{d}u\}}
{\exp\{\int_{x}^{y}\frac{W^{(q)'}(u-\xi(u))}{W^{(q)}(u-\xi(u))}\mathrm{d}u\}\frac{W^{(q)'}(y-\xi(y))}{W^{(q)}(y-\xi(y))}}\right|_{x}^{_{\infty}}
\nonumber\\
\hspace{-0.3cm}&&\hspace{-0.3cm}
-\int_{x}^{\infty}\frac{\left[\exp\{\int_{x}^{y}\frac{W^{(q)'}(u-\xi(u))}{W^{(q)}(u-\xi(u))}\mathrm{d}u\}\frac{W^{(q)'}(y-\xi(y))}{W^{(q)}(y-\xi(y))}\right]'\exp\{-\frac{\gamma_{2}}{1-\gamma_{2}}\int_{x}^{y}\frac{W^{(q)'}(u-\xi(u))}{W^{(q)}(u-\xi(u))}\mathrm{d}u\}}
{\left[\exp\{\int_{x}^{y}\frac{W^{(q)'}(u-\xi(u))}{W^{(q)}(u-\xi(u))}\mathrm{d}u\}\frac{W^{(q)'}(y-\xi(y))}{W^{(q)}(y-\xi(y))}\right]^{2}}\mathrm{d}y
\nonumber\\
\hspace{-0.3cm}&=&\hspace{-0.3cm}
-\int_{x}^{\infty}\exp\{-\frac{1}{1-\gamma_{2}}\int_{x}^{y}\frac{W^{(q)'}(u-\xi(u))}{W^{(q)}(u-\xi(u))}\mathrm{d}u\}
\left(1+\frac{\left[\frac{W^{(q)'}(y-\xi(y))}{W^{(q)}(y-\xi(y))}\right]'}
{\left[\frac{W^{(q)'}(y-\xi(y))}{W^{(q)}(y-\xi(y))}\right]^{2}}\right)\mathrm{d}y
\nonumber\\
\hspace{-0.3cm}&&\hspace{-0.3cm}+\frac{W^{(q)}(x-\xi(x))}{W^{(q)'}(x-\xi(x))}
\nonumber\\
\hspace{-0.3cm}&=&\hspace{-0.3cm}
\frac{W^{(q)}(x-\xi(x))}{W^{(q)'}(x-\xi(x))}-\int_{x}^{\infty}\exp\{-\frac{1}{1-\gamma_{2}}\int_{x}^{y}\frac{W^{(q)'}(u-\xi(u))}{W^{(q)}(u-\xi(u))}\mathrm{d}u\}
\nonumber\\
\hspace{-0.3cm}&&\hspace{-0.3cm}\times
\left(1+\left(\frac{W^{(q)''}(y-\xi(y))W^{(q)}(y-\xi(y))}
{\left[W^{(q)'}(y-\xi(y))\right]^{2}}-1\right)\left(1-\xi'(y)\right)\right)\mathrm{d}y
\nonumber\\
\hspace{-0.3cm}&=&\hspace{-0.3cm}
\frac{W^{(q)}(x-\xi(x))}{W^{(q)'}(x-\xi(x))}-G_2(x),\qquad x\geq0,
\end{eqnarray}
where
\begin{eqnarray}\label{condition1bia}
G_{2}(x)\hspace{-0.3cm}&=&\hspace{-0.3cm}\int_{x}^{\infty}\exp\{-\frac{1}{1-\gamma_{2}}\int_{x}^{y}\frac{W^{(q)'}(u-\xi(u))}{W^{(q)}(u-\xi(u))}\mathrm{d}u\}g(y)\mathrm{d}y,\qquad x\geq 0.
\end{eqnarray}
Here the function $g$ is defined as
\begin{eqnarray}\label{assumption1equ.}
g(x)=\xi'(x)+\left(1-\xi'(x)\right)\frac{W^{(q)''}(x-\xi(x))W^{(q)}(x-\xi(x))}
{\left(W^{(q)'}(x-\xi(x))\right)^{2}},\qquad x\geq 0.\nonumber
\end{eqnarray}
With the above alternative representation of $f(x)$, we conclude that $\frac{W^{(q)'}(x-\xi(x))}{W^{(q)}(x-\xi(x))}f(x)\leq1$ ($\forall\,x>0$) is equivalent to the following inequality,
\begin{eqnarray}\label{condition1}
G_2(x)\geq0,\qquad \forall\,x>0.
\end{eqnarray}

%   Some algebraic manipulations
%   Solving (4.1) with $f(\infty)\leq\frac{1}{\Phi(c)}<\infty$ gives rise to
%   \begin{eqnarray}
%   f(x)
%   =\frac{\gamma_{2}}{1-\gamma_{2}}(W^{(q)}(x))^{\frac{1}{1-\gamma_{2}}}\int_{x}^{\infty}(W^{(q)}(y))^{-\frac{1}{1-\gamma_{2}}}\mathrm{d}y,\hspace{0.1cm}x\geq0.
%   \end{eqnarray}

In order to characterize the optimal tax strategy and the optimal tax return function explicitly, we need the following Assumption \ref{assumption1}, under which the optimal tax return function and the optimal tax strategy are obtained.
As can be seen, the extension from ruin-based tax optimization problem to our general draw-down-based tax optimization problem brings in much more difficulties and complexity, especially in finding the optimal tax return function and the optimal tax strategy, while  characterizing them is very critical in addressing optimal control problems.
When our model and tax structure are specified to those in
existing literature, our results coincide with the corresponding results.

\begin{Assumption}
\label{assumption1} There exists an $x_{0}\in[0,\infty)$ such that the  function $g$
changes its sign at the point $x_{0}$. Here, we say that the function $g(x)$ changes its sign at the point $x_{0}$, if and only if $g(x_{1})g(x_{2})\leq0,\,\forall\,\,x_{1}\leq x_{0},\,x_{2}\geq x_{0}$.

\end{Assumption}

\begin{proposition}\label{verify solution1}
Suppose that Assumption \ref{assumption1} and condition (\ref{condition1}) hold true, which is quadrivalent to the combination of the following two Cases (i) and (ii).
\begin{itemize}
    \item[(i)] $x_{0}=0$, $g(x)\leq0$ for all $x\leq x_{0}$, $g(x)\geq0$ for all $x\geq x_{0}$. Or, $x_{0}=\infty$, $g(x)\geq0$ for all $x\leq x_{0}$, $g(x)\leq0$ for all $x\geq x_{0}$.
    \item[(ii)] $x_{0}\in(0,\infty)$, $g(x)\leq0$ for all $x\leq x_{0}$, $g(x)\geq0$ for all $x\geq x_{0}$, and the inequality (\ref{condition1}) holds true for $x=0$.
\end{itemize}
Let $f(x)$ be defined by (\ref{represntation of the optimal function1}), then $f(x)$ is indeed a once continuously differentiable solution to
(\ref{HJB equation}). In addition, the maximizing function is
$\gamma(x)\equiv\gamma_{2},\,\forall\,\,x\geq0$.
\end{proposition}

\begin{proof}
 It is already proved, see the arguments before Assumption \ref{assumption1}.
\end{proof}

\vspace{0.3cm}
We proceed to characterize the optimal solution for (\ref{HJB equation}). This time we guess that the solution $f(x)$ of equation (\ref{HJB equation}) satisfies $\frac{W^{(q)'}(x-\xi(x))}{W^{(q)}(x-\xi(x))}f(x)\geq1$ for all $x$. Then, for all $x\geq0$, the function $f$ satisfies
\begin{eqnarray}\label{4.8}
0
=\gamma_{1}-\frac{W^{(q)'}(x-\xi(x))}{W^{(q)}(x-\xi(x))}f(x)+(1-\gamma_{1})f'(x).
\end{eqnarray}
Solving (\ref{4.8}) by similar arguments as in solving (\ref{4.1}), we obtain  the following solution,
\begin{eqnarray}\label{represntation of the optimal function2}
f(x)
\hspace{-0.3cm}&=&\hspace{-0.3cm}\frac{\gamma_{1}}{1-\gamma_{1}}\int_{x}^{\infty}\exp\{-\frac{1}{1-\gamma_{1}}\int_{x}^{y}\frac{W^{(q)'}(u-\xi(u))}{W^{(q)}(u-\xi(u))}\mathrm{d}u\}\mathrm{d}y
\nonumber\\
\hspace{-0.3cm}&=&\hspace{-0.3cm}
-\int_{x}^{\infty}\exp\{-\frac{1}{1-\gamma_{1}}\int_{x}^{y}\frac{W^{(q)'}(u-\xi(u))}{W^{(q)}(u-\xi(u))}\mathrm{d}u\}
\left(1+\frac{\left[\frac{W^{(q)'}(y-\xi(y))}{W^{(q)}(y-\xi(y))}\right]'}
{\left[\frac{W^{(q)'}(y-\xi(y))}{W^{(q)}(y-\xi(y))}\right]^{2}}\right)\mathrm{d}y
\nonumber\\
\hspace{-0.3cm}&&\hspace{-0.3cm}+\frac{W^{(q)}(x-\xi(x))}{W^{(q)'}(x-\xi(x))}
\nonumber\\
\hspace{-0.3cm}&=&\hspace{-0.3cm}
\frac{W^{(q)}(x-\xi(x))}{W^{(q)'}(x-\xi(x))}-G_1(x),\qquad x\geq0,
\end{eqnarray}
where
\begin{eqnarray}\label{condition2bia}
G_{1}(x)\hspace{-0.3cm}&=&\hspace{-0.3cm}\int_{x}^{\infty}\exp\{-\frac{1}{1-\gamma_{1}}\int_{x}^{y}\frac{W^{(q)'}(u-\xi(u))}{W^{(q)}(u-\xi(u))}\mathrm{d}u\}g(y)\mathrm{d}y.
\end{eqnarray}
By (\ref{represntation of the optimal function2}) it is obvious that
$$\frac{W^{(q)'}(x-\xi(x))}{W^{(q)}(x-\xi(x))}f(x)\geq1,\qquad \forall\,\,x>0,$$
is equivalent to
\begin{eqnarray}\label{condition2}
G_1(x)\leq0,\qquad \forall\,\,x>0.
\end{eqnarray}

\vspace{0.2cm}
\begin{proposition}\label{verify solution2}
Suppose that Assumption \ref{assumption1} and condition (\ref{condition2}) hold true, which is quadrivalent to the combination of the following two cases:
\begin{itemize}
    \item[(iii)] $x_{0}=0$, $g(x)\geq0$ for all $x\leq x_{0}$, $g(x)\leq0$ for all $x\geq x_{0}$. Or, $x_{0}=\infty$, $g(x)\leq0$ for all $x\leq x_{0}$, $g(x)\geq0$ for all $x\geq x_{0}$.
    \item[(iv)] $x_{0}\in(0,\infty)$, $g(x)\geq0$ for all $x\leq x_{0}$, $g(x)\leq0$ for all $x\geq x_{0}$, and the inequality (\ref{condition2}) holds true for $x=0$.
\end{itemize} Let $f(x)$ be defined by (\ref{represntation of the optimal function2}), then $f(x)$ is indeed a once continuously differentiable solution to
(\ref{HJB equation}). In addition, the maximizing function is
$\gamma(x)\equiv\gamma_{1},\,\forall\,\,x\geq0$.
\end{proposition}
\begin{proof}
It is proven in the arguments between  Propositions \ref{verify solution1} and \ref{verify solution2}.
\end{proof}

\vspace{0.3cm}
Aside from cases (i)-(iv), we also consider the following two cases.
\begin{itemize}
  \item[$(v)$]  $x_{0}\in(0,\infty)$, $g(x)\leq0$ for all $x\leq x_{0}$, $g(x)\geq0$ for all $x\geq x_{0}$,
  and there exists $\bar{x}>0$ such that (\ref{condition1}) does not hold true anymore, i.e., $G_{2}(\bar{x})<0$.
 \item[$(vi)$] $x_{0}\in(0,\infty)$, $g(x)\geq0$ for all $x\leq x_{0}$, $g(x)\leq0$ for all $x\geq x_{0}$,
and there exists $\bar{\bar{x}}>0$ such that (\ref{condition2}) does not hold true anymore, i.e., $G_{1}(\bar{\bar{x}})>0$.
\end{itemize}

In case $(v)$, we must have $\bar{x}<x_{0}$. By the definition of $x_{0}$ we know that $G_{2}(x_{0})\geq0$.
Then from the intermediate value theorem for continuous function $G_{2}(x)$ we claim that, there must exist some $x\in(\bar{x},x_{0}]$ such that $G_{2}(x)=0$. Let,
\begin{eqnarray}\label{x1}
x_{1}=\inf\{x\in(\bar{x},x_{0}]\mid G_{2}(x)=0\}.
\end{eqnarray}
Then, $G_{2}(x_{1})=0$, $G_{2}(x)<0$ for all $0<x<x_{1}$, and $G_{2}(x)\geq0$ for all $x\in[x_{1},\infty)$.

Hence, this time we guess that the solution $f(x)$ of equation (\ref{HJB equation}) satisfies $\frac{W^{(q)'}(x-\xi(x))}{W^{(q)}(x-\xi(x))}f(x)\geq1$ for all $x\in(0,x_{1})$, and satisfies $\frac{W^{(q)'}(x-\xi(x))}{W^{(q)}(x-\xi(x))}f(x)\leq1$ for all $x\in[x_{1},\infty)$. That is to say, $f(x)$ satisfies the following system of differential equations,
\begin{eqnarray}\label{4.12}
0\hspace{-0.3cm}&=&\hspace{-0.3cm}\gamma_{1}-\frac{W^{(q)'}(x-\xi(x))}{W^{(q)}(x-\xi(x))}f(x)+(1-\gamma_{1})f'(x),\qquad \forall\,\,x\in(0,x_{1}),
\nonumber\\
0\hspace{-0.3cm}&=&\hspace{-0.3cm}\gamma_{2}-\frac{W^{(q)'}(x-\xi(x))}{W^{(q)}(x-\xi(x))}f(x)+(1-\gamma_{2})f'(x),\qquad \forall\,\,x\in[x_{1},\infty).
\end{eqnarray}
Solving the above system of equations yields,
\begin{eqnarray}
f(x)\hspace{-0.3cm}&=&\hspace{-0.3cm}
\left(C-\frac{\gamma_{1}}{1-\gamma_{1}}\int_{0}^{x}\exp\{-\frac{1}{1-\gamma_{1}}\int_{0}^{y}\frac{W^{(q)'}(u-\xi(u))}{W^{(q)}(u-\xi(u))}\mathrm{d}u\}\mathrm{d}y\right)
\nonumber\\
\hspace{-0.3cm}&&\hspace{-0.3cm}\times\exp\{\frac{1}{1-\gamma_{1}}\int_{0}^{x}\frac{W^{(q)'}(y-\xi(y))}{W^{(q)}(y-\xi(y))}\mathrm{d}y\},\,\,\forall\,\,x\in(0,x_{1}),
\nonumber\\
f(x)\hspace{-0.3cm}&=&\hspace{-0.3cm}\frac{\gamma_{2}}{1-\gamma_{2}}\int_{x}^{\infty}\exp\{-\frac{1}{1-\gamma_{2}}\int_{x}^{y}\frac{W^{(q)'}(u-\xi(u))}{W^{(q)}(u-\xi(u))}\mathrm{d}u\}\mathrm{d}y,
\qquad \forall\,\,x\in[x_{1},\infty).\nonumber
\end{eqnarray}
Using the continuity condition $f(x_{1}-)=f(x_{1}+)=f(x_{1})$ we can determine the constant $C$ as
\begin{eqnarray}
C\hspace{-0.3cm}&=&\hspace{-0.3cm}\frac{\gamma_{1}}{1-\gamma_{1}}\int_{0}^{x_{1}}\exp\{-\frac{1}{1-\gamma_{1}}\int_{0}^{y}\frac{W^{(q)'}(u-\xi(u))}{W^{(q)}(u-\xi(u))}\mathrm{d}u\}\mathrm{d}y
+\exp\{\frac{-1}{1-\gamma_{1}}\int_{0}^{x_{1}}\frac{W^{(q)'}(y-\xi(y))}{W^{(q)}(y-\xi(y))}\mathrm{d}y\}\nonumber\\
\hspace{-0.3cm}&&\hspace{-0.3cm}
\times\frac{\gamma_{2}}{1-\gamma_{2}}\int_{x_{1}}^{\infty}\exp\{\frac{-1}{1-\gamma_{2}}\int_{x_{1}}^{y}\frac{W^{(q)'}(u-\xi(u))}{W^{(q)}(u-\xi(u))}\mathrm{d}u\}\mathrm{d}y.\nonumber
\end{eqnarray}
Therefore we can write down the solution to the set of equation (\ref{4.12}) as follows,
\begin{eqnarray}\label{4.13}
f(x)\hspace{-0.3cm}&=&\hspace{-0.3cm}
\frac{\gamma_{1}}{1-\gamma_{1}}\int_{x}^{x_{1}}\exp\{\frac{-1}{1-\gamma_{1}}\int_{x}^{y}\frac{W^{(q)'}(u-\xi(u))}{W^{(q)}(u-\xi(u))}\mathrm{d}u\}\mathrm{d}y
+\exp\{\frac{-1}{1-\gamma_{1}}\int_{x}^{x_{1}}\frac{W^{(q)'}(y-\xi(y))}{W^{(q)}(y-\xi(y))}\mathrm{d}y\}\nonumber\\
\hspace{-0.3cm}&&\hspace{-0.3cm}
\times\frac{\gamma_{2}}{1-\gamma_{2}}\int_{x_{1}}^{\infty}\exp\{\frac{-1}{1-\gamma_{2}}\int_{x_{1}}^{y}\frac{W^{(q)'}(u-\xi(u))}{W^{(q)}(u-\xi(u))}\mathrm{d}u\}\mathrm{d}y,\qquad \forall\,\,x\in(0,x_{1}),
\nonumber\\
f(x)\hspace{-0.3cm}&=&\hspace{-0.3cm}\frac{\gamma_{2}}{1-\gamma_{2}}\int_{x}^{\infty}\exp\{-\frac{1}{1-\gamma_{2}}\int_{x}^{y}\frac{W^{(q)'}(u-\xi(u))}{W^{(q)}(u-\xi(u))}\mathrm{d}u\}\mathrm{d}y,
\qquad \forall\,\,x\in[x_{1},\infty).
\end{eqnarray}
We need only to further prove that the function given by (\ref{4.13}) satisfies $\frac{W^{(q)'}(x-\xi(x))}{W^{(q)}(x-\xi(x))}f(x)\geq1$ for all $x\in(0,x_{1})$ to guarantee itself a solution to the HJB equation (\ref{HJB equation}). By some algebraic manipulations we get, for $x\in(0,x_{1})$,
\begin{eqnarray}\label{4.14}
\hspace{-0.3cm}
f(x)\hspace{-0.3cm}&=&\hspace{-0.3cm}\int_{x}^{x_{1}}\mathrm{d}\left(-\frac{\exp\{-\frac{\gamma_{1}}{1-\gamma_{1}}\int_{x}^{y}\frac{W^{(q)'}(u-\xi(u))}{W^{(q)}(u-\xi(u))}\mathrm{d}u\}}
{\exp\{\int_{x}^{y}\frac{W^{(q)'}(u-\xi(u))}{W^{(q)}(u-\xi(u))}\mathrm{d}u\}\frac{W^{(q)'}(y-\xi(y))}{W^{(q)}(y-\xi(y))}}\right)\nonumber\\
\hspace{-0.3cm}&&\hspace{-0.3cm}
-\frac{\left[\exp\{\int_{x}^{y}\frac{W^{(q)'}(u-\xi(u))}{W^{(q)}(u-\xi(u))}\mathrm{d}u\}\frac{W^{(q)'}(y-\xi(y))}{W^{(q)}(y-\xi(y))}\right]'\exp\{-\frac{\gamma_{1}}{1-\gamma_{1}}\int_{x}^{y}\frac{W^{(q)'}(u-\xi(u))}{W^{(q)}(u-\xi(u))}\mathrm{d}u\}}
{\left[\exp\{\int_{x}^{y}\frac{W^{(q)'}(u-\xi(u))}{W^{(q)}(u-\xi(u))}\mathrm{d}u\}\frac{W^{(q)'}(y-\xi(y))}{W^{(q)}(y-\xi(y))}\right]^{2}}\mathrm{d}y
\nonumber\\
\hspace{-0.3cm}&&\hspace{-0.3cm}
+\exp\{\frac{-1}{1-\gamma_{1}}\int_{x}^{x_{1}}\frac{W^{(q)'}(y-\xi(y))}{W^{(q)}(y-\xi(y))}\mathrm{d}y\}
\nonumber\\
\hspace{-0.3cm}&&\hspace{-0.3cm}\times
\int_{x_{1}}^{\infty}\mathrm{d}\left(-\frac{\exp\{-\frac{\gamma_{2}}{1-\gamma_{2}}\int_{x}^{y}\frac{W^{(q)'}(u-\xi(u))}{W^{(q)}(u-\xi(u))}\mathrm{d}u\}}
{\exp\{\int_{x}^{y}\frac{W^{(q)'}(u-\xi(u))}{W^{(q)}(u-\xi(u))}\mathrm{d}u\}\frac{W^{(q)'}(y-\xi(y))}{W^{(q)}(y-\xi(y))}}\right)\nonumber\\
\hspace{-0.3cm}&&\hspace{-0.3cm}
-\frac{\left[\exp\{\int_{x}^{y}\frac{W^{(q)'}(u-\xi(u))}{W^{(q)}(u-\xi(u))}\mathrm{d}u\}\frac{W^{(q)'}(y-\xi(y))}{W^{(q)}(y-\xi(y))}\right]'\exp\{-\frac{\gamma_{2}}{1-\gamma_{2}}\int_{x}^{y}\frac{W^{(q)'}(u-\xi(u))}{W^{(q)}(u-\xi(u))}\mathrm{d}u\}}
{\left[\exp\{\int_{x}^{y}\frac{W^{(q)'}(u-\xi(u))}{W^{(q)}(u-\xi(u))}\mathrm{d}u\}\frac{W^{(q)'}(y-\xi(y))}{W^{(q)}(y-\xi(y))}\right]^{2}}\mathrm{d}y
\nonumber
\\
\hspace{-0.3cm}&=&\hspace{-0.3cm}
\frac{W^{(q)}(x-\xi(x))}{W^{(q)'}(x-\xi(x))}
-\frac{\exp\{-\frac{1}{1-\gamma_{1}}\int_{x}^{x_{1}}\frac{W^{(q)'}(u-\xi(u))}{W^{(q)}(u-\xi(u))}\mathrm{d}u\}}
{\frac{W^{(q)'}(x_{1}-\xi(x_{1}))}{W^{(q)}(x_{1}-\xi(x_{1}))}}
\nonumber\\
\hspace{-0.3cm}&&\hspace{-0.3cm}
-\int_{x}^{x_{1}}
\exp\{-\frac{1}{1-\gamma_{1}}\int_{x}^{y}\frac{W^{(q)'}(u-\xi(u))}{W^{(q)}(u-\xi(u))}\mathrm{d}u\}
\nonumber\\
\hspace{-0.3cm}&&\hspace{-0.3cm}\times
\left(\xi'(y)+\frac{W^{(q)''}(y-\xi(y))W^{(q)}(y-\xi(y))}
{\left[W^{(q)'}(y-\xi(y))\right]^{2}}\left(1-\xi'(y)\right)\right)\mathrm{d}y
\nonumber\\
\hspace{-0.3cm}&&\hspace{-0.3cm}
+\exp\{\frac{-1}{1-\gamma_{1}}\int_{x}^{x_{1}}\frac{W^{(q)'}(y-\xi(y))}{W^{(q)}(y-\xi(y))}\mathrm{d}y\}
\nonumber\\
\hspace{-0.3cm}&&\hspace{-0.3cm}\times
\left(\frac{W^{(q)}(x_{1}-\xi(x_{1}))}{W^{(q)'}(x_{1}-\xi(x_{1}))}
-\int_{x_{1}}^{\infty}
\exp\{-\frac{1}{1-\gamma_{2}}\int_{x}^{y}\frac{W^{(q)'}(u-\xi(u))}{W^{(q)}(u-\xi(u))}\mathrm{d}u\}\right.
\nonumber\\
\hspace{-0.3cm}&&\hspace{-0.3cm}\times
\left.\left(\xi'(y)+\frac{W^{(q)''}(y-\xi(y))W^{(q)}(y-\xi(y))}
{\left[W^{(q)'}(y-\xi(y))\right]^{2}}\left(1-\xi'(y)\right)\right)\mathrm{d}y\right)
\nonumber
\\
\hspace{-0.3cm}&=&\hspace{-0.3cm}
\frac{W^{(q)}(x-\xi(x))}{W^{(q)'}(x-\xi(x))}
-\int_{x}^{x_{1}}
\exp\{-\frac{1}{1-\gamma_{1}}\int_{x}^{y}\frac{W^{(q)'}(u-\xi(u))}{W^{(q)}(u-\xi(u))}\mathrm{d}u\}
\nonumber\\
\hspace{-0.3cm}&&\hspace{-0.3cm}\times
\left(\xi'(y)+\frac{W^{(q)''}(y-\xi(y))W^{(q)}(y-\xi(y))}
{\left[W^{(q)'}(y-\xi(y))\right]^{2}}\left(1-\xi'(y)\right)\right)\mathrm{d}y
\nonumber\\
\hspace{-0.3cm}&&\hspace{-0.3cm}
-\exp\{\frac{-1}{1-\gamma_{1}}\int_{x}^{x_{1}}\frac{W^{(q)'}(y-\xi(y))}{W^{(q)}(y-\xi(y))}\mathrm{d}y\}
\int_{x_{1}}^{\infty}
\exp\{-\frac{1}{1-\gamma_{2}}\int_{x}^{y}\frac{W^{(q)'}(u-\xi(u))}{W^{(q)}(u-\xi(u))}\mathrm{d}u\}
\nonumber\\
\hspace{-0.3cm}&&\hspace{-0.3cm}\times
\left(\xi'(y)+\frac{W^{(q)''}(y-\xi(y))W^{(q)}(y-\xi(y))}
{\left[W^{(q)'}(y-\xi(y))\right]^{2}}\left(1-\xi'(y)\right)\right)\mathrm{d}y
\nonumber
\\
\hspace{-0.3cm}&=&\hspace{-0.3cm}
\frac{W^{(q)}(x-\xi(x))}{W^{(q)'}(x-\xi(x))}
-\int_{x}^{x_{1}}
\exp\{-\frac{1}{1-\gamma_{1}}\int_{x}^{y}\frac{W^{(q)'}(u-\xi(u))}{W^{(q)}(u-\xi(u))}\mathrm{d}u\}g(y)\mathrm{d}y
\nonumber\\
\hspace{-0.3cm}&&\hspace{-0.3cm}
-\exp\{\frac{-1}{1-\gamma_{1}}\int_{x}^{x_{1}}\frac{W^{(q)'}(y-\xi(y))}{W^{(q)}(y-\xi(y))}\mathrm{d}y\}
G_{2}(x_{1})\nonumber
\nonumber
\\
\hspace{-0.3cm}&\geq&\hspace{-0.3cm}\frac{W^{(q)}(x-\xi(x))}{W^{(q)'}(x-\xi(x))},
\end{eqnarray}
since $g(y)\leq0$ for $y\in[x,x_{1}]\subseteq[0,x_{0}]$ and $G_{2}(x_{1})=0$. Inequality (\ref{4.14}) reveals that $\frac{W^{(q)'}(x-\xi(x))}{W^{(q)}(x-\xi(x))}f(x)\geq1$ for all $x\in(0,x_{1})$.

\begin{proposition}\label{verify solution3}
In case (v), the once continuously differentiable solution to Equation (\ref{HJB equation}) (optimal tax function) $f(x)$ is defined by (\ref{4.13}) with $x_{1}$ determined by (\ref{x1}) and $x_{0}$ given by Assumption \ref{assumption1}. In addition, the maximizing function is
$\gamma(x)=\gamma_{1},\,\forall\,\,x<x_{1}$, and $\gamma(x)=\gamma_{2},\,\forall\,\,x\geq x_{1}$.

In case (vi), the once continuously differentiable solution to Equation (\ref{HJB equation}) (optimal tax function) $f(x)$ is defined by,
\begin{eqnarray}\label{4.15}
f(x)\hspace{-0.3cm}&=&\hspace{-0.3cm}
\frac{\gamma_{2}}{1-\gamma_{2}}\int_{x}^{x_{2}}\exp\{\frac{-1}{1-\gamma_{2}}\int_{x}^{y}\frac{W^{(q)'}(u-\xi(u))}{W^{(q)}(u-\xi(u))}\mathrm{d}u\}\mathrm{d}y
+\exp\{\frac{-1}{1-\gamma_{2}}\int_{x}^{x_{2}}\frac{W^{(q)'}(y-\xi(y))}{W^{(q)}(y-\xi(y))}\mathrm{d}y\}\nonumber\\
\hspace{-0.3cm}&&\hspace{-0.3cm}
\times\frac{\gamma_{1}}{1-\gamma_{1}}\int_{x_{2}}^{\infty}\exp\{\frac{-1}{1-\gamma_{2}}\int_{x_{2}}^{y}\frac{W^{(q)'}(u-\xi(u))}{W^{(q)}(u-\xi(u))}\mathrm{d}u\}\mathrm{d}y,\qquad \forall\,\,x\in(0,x_{2}),
\nonumber\\
f(x)\hspace{-0.3cm}&=&\hspace{-0.3cm}\frac{\gamma_{1}}{1-\gamma_{1}}\int_{x}^{\infty}\exp\{-\frac{1}{1-\gamma_{1}}\int_{x}^{y}\frac{W^{(q)'}(u-\xi(u))}{W^{(q)}(u-\xi(u))}\mathrm{d}u\}\mathrm{d}y,
\qquad \forall\,\,x\in[x_{2},\infty).
\end{eqnarray}
with $x_{2}$ determined by,
\begin{eqnarray}\label{x2}
x_{2}=\inf\{x\in(\bar{\bar{x}},x_{0}]\mid G_{1}(x)=0\},
\end{eqnarray}
and $x_{0}$ given by Assumption \ref{assumption1}. In addition, the maximizing function is
$\gamma(x)=\gamma_{2},\,\forall\,\,x<x_{2}$, and $\gamma(x)=\gamma_{1},\,\forall\,\,x\geq x_{2}$.
\end{proposition}
\begin{proof}
 The proving arguments for Case $(v)$ are given right above this proposition. The proof of Case $(vi)$ is much similar.
 \end{proof}

\begin{remark}\label{2} Let $\xi(x)\equiv0$, then the Hamilton-Jacobi-Bellman equation (\ref{HJB equation}) coincides well with Equation (2.6) in Wang and Hu (2012).
Further, if it is assumed that each scale function is three times differentiable and its first derivative
is a strictly convex function (as was assumed in Wang and Hu (2012) and Albrecher et al. (2008b)), then
$$g(x)=\frac{W^{(q)''}(x)W^{(q)}(x)}
{\left(W^{(q)'}(x)\right)^{2}},\qquad x\geq 0,$$
will change its sign at most once. In addition, if
$g(x)$ changes its sign once at $x_{0}\in[0,\infty)$,
we must have $g(x)\leq0$ for $x\in[0,x_{0}]$, and $g(x)\geq0$ for $x\geq x_{0}$.

That is to say, only cases $(i)$, $(ii)$ and $(v)$ are possible scenarios. Combined with the last equality in (\ref{4.14}), it is obvious that (\ref{4.13}) coincides with (5.7) of Wang and Hu (2012), in case $(v)$.
Meanwhile, (\ref{represntation of the optimal function1}) coincides with (4.2) of Wang and Hu (2012), in cases $(i)$ and $(ii)$.
Indeed, when $\xi\equiv0$ we have
\begin{eqnarray}\label{}
G_{2}(x)\hspace{-0.3cm}&=&\hspace{-0.3cm}\int_{x}^{\infty}
\exp\{-\frac{1}{1-\gamma_{2}}\int_{x}^{y}\frac{W^{(q)'}(u)}{W^{(q)}(u)}\mathrm{d}u\}g(y)\mathrm{d}y
\nonumber\\
\hspace{-0.3cm}&=&\hspace{-0.3cm}
\int_{x}^{\infty}\left(\frac{W^{(q)}(x)}{W^{(q)}(y)}\right)^{\frac{1}{1-\gamma_{2}}}
\frac{W^{(q)''}(y)W^{(q)}(y)}
{\left(W^{(q)'}(y)\right)^{2}}\mathrm{d}y
\nonumber\\
\hspace{-0.3cm}&=&\hspace{-0.3cm}
\left(W^{(q)}(x)\right)^{\frac{1}{1-\gamma_{2}}}
\int_{x}^{\infty}\frac{W^{(q)''}(y)\left(W^{(q)}(y)\right)^{1-\frac{1}{1-\gamma_{2}}}}
{\left(W^{(q)'}(y)\right)^{2}}\mathrm{d}y
,\qquad x\geq 0.\nonumber
\end{eqnarray}
Hence, $x_{1}$ as defined by \eqref{x1} in case (v) coincides with $u_{0}$ defined in (5.15) of Wang and Hu (2012), which implies that (5.22) of Wang and Hu (2012) holds true with $\beta=\gamma_{2}$ and $c=q$
\begin{eqnarray}\label{v1.4.21}
\frac{W^{(q)}(x_{1})}{W^{(q)'}(x_{1})}=\frac{\gamma_{2}}{1-\gamma_{2}}
\int_{x_{1}}^{\infty}\left(\frac{W^{(q)}(x_{1})}{W^{(q)}(y)}\right)^{\frac{1}{1-\gamma_{2}}}\mathrm{d}y.
\end{eqnarray}
 Then for $x\in(0,x_{1})$, by \eqref{v1.4.21}, (\ref{4.13}) can be rewritten as
\begin{eqnarray}\label{}
f(x)\hspace{-0.3cm}&=&\hspace{-0.3cm}
\frac{\gamma_{1}}{1-\gamma_{1}}\int_{x}^{x_{1}}\exp\{\frac{-1}{1-\gamma_{1}}\int_{x}^{y}\frac{W^{(q)'}(u)}{W^{(q)}(u)}\mathrm{d}u\}\mathrm{d}y
+\exp\{\frac{-1}{1-\gamma_{1}}\int_{x}^{x_{1}}\frac{W^{(q)'}(y)}{W^{(q)}(y)}\mathrm{d}y\}\nonumber\\
\hspace{-0.3cm}&&\hspace{-0.3cm}
\times\frac{\gamma_{2}}{1-\gamma_{2}}\int_{x_{1}}^{\infty}\exp\{\frac{-1}{1-\gamma_{2}}\int_{x_{1}}^{y}\frac{W^{(q)'}(u)}{W^{(q)}(u)}\mathrm{d}u\}\mathrm{d}y,\qquad \nonumber\\
\hspace{-0.3cm}&=&\hspace{-0.3cm}
\frac{\gamma_{1}}{1-\gamma_{1}}\int_{x}^{x_{1}}
\left(\frac{W^{(q)}(x)}{W^{(q)}(y)}\right)^{\frac{1}{1-\gamma_{1}}}\mathrm{d}y
+\frac{\gamma_{2}}{1-\gamma_{2}}\left(\frac{W^{(q)}(x)}{W^{(q)}(x_{1})}\right)^{\frac{1}{1-\gamma_{1}}}
\int_{x_{1}}^{\infty}\left(\frac{W^{(q)}(x_{1})}{W^{(q)}(y)}\right)^{\frac{1}{1-\gamma_{2}}}\mathrm{d}y
\nonumber\\
\hspace{-0.3cm}&=&\hspace{-0.3cm}
\frac{\gamma_{1}}{1-\gamma_{1}}\int_{x}^{x_{1}}
\left(\frac{W^{(q)}(x)}{W^{(q)}(y)}\right)^{\frac{1}{1-\gamma_{1}}}\mathrm{d}y
+\left(\frac{W^{(q)}(x)}{W^{(q)}(x_{1})}\right)^{\frac{1}{1-\gamma_{1}}}\frac{W^{(q)}(x_{1})}{W^{(q)'}(x_{1})}
\nonumber\\
\hspace{-0.3cm}&=&\hspace{-0.3cm}
\left(W^{(q)}(x)\right)^{\frac{1}{1-\gamma_{1}}}\left(\frac{\gamma_{1}}{1-\gamma_{1}}\int_{x}^{x_{1}}
\left(W^{(q)}(y)\right)^{-\frac{1}{1-\gamma_{1}}}\mathrm{d}y
+\frac{\left(W^{(q)}(x_{1})\right)^{1-\frac{1}{1-\gamma_{1}}}}{W^{(q)'}(x_{1})}\right)
,
\end{eqnarray}
which coincides well with (5.7) in Wang and Hu (2012).
\end{remark}

\section{Examples}
We should not take it for granted that the optimal loss-carry-forward tax strategy can be found, for any general spectrally negative L\'evy process $X$ and any general draw-down function $\xi$. Actually, it is under Assumption \ref{assumption1} that the optimal tax strategy is found in Section 5.
In this section, we consider several examples of spectrally negative L\'{e}vy processes $X$ and show that Assumption \ref{assumption1} holds for  certain choices of draw-down function $\xi$.
For simplicity of notation, write $\bar{\xi}(x):=x-\xi(x)$, $W_\xi(z):=\frac{W^{(q)}(\bar{\xi}(z))} {W^{(q)'}(\bar{\xi}(z))}$ and hence $W_0(z):=\frac{W^{(q)}(z)} {W^{(q)'}(z)}$.

\begin{exa}\label{exa1}
In this example, let $X_{t}=\mu t+B_{t}$ for constant $\mu$ and Brownian motion $B$.
The $q$-scale function of $X$ is given by
\begin{eqnarray}\label{Ex222}
W^{(q)}(x)=\frac{1}{\sqrt{\mu^2+2q}}\left(\mathrm{e}^{\theta_1x}-\mathrm{e}^{\theta_2x}\right),\quad x\geq0,
\end{eqnarray}
with $\theta_1=-\mu+\sqrt{\mu^2+2q}$, $\theta_2=-(\mu+\sqrt{\mu^2+2q})$.
It can be checked that
\begin{eqnarray}\label{Wqrelation}
&&W^{(q)'}(x)=\frac{1}{\sqrt{\mu^2+2q}}\left(\theta_1\mathrm{e}^{\theta_1x}-\theta_2\mathrm{e}^{\theta_2x}\right)>0,\quad x\geq0,
\nonumber\\
&
&W^{(q)''}(x)=2\left(qW^{(q)}(x)
-\mu W^{(q)'}(x)\right),\quad x\geq0.\nonumber
\end{eqnarray}
The function $g$ can be rewritten as
\begin{align*}
g(x)&=\xi'(x)+(1-\xi'(x))\frac{W^{(q)''}(\bar{\xi}(x))W^{(q)}(\bar{\xi}(x))}{(W^{(q)'}(\bar{\xi}(x)))^2}\\
&=\xi'(x)+(1-\xi'(x))(2qW_{\xi}(x)^2-2\mu W_{\xi}(x)),
\end{align*}
from which one can deduce that
\begin{align*}
g'(x)=\xi''(x)f_1(W_{\xi}(x))+(\xi'(x)-1)W_{\xi}'(x)f_2(W_{\xi}(x)),
\end{align*}
with $f_1(x)=1-2qx^2+2\mu x$ and $f_2(x)=-4qx+2\mu$. One can find that $f_1(x)>(<)\,0$ on $[0,\frac{\mu+\sqrt{\mu^2+2q}}{2q})$ $((\frac{\mu+\sqrt{\mu^2+2q}}{2q},\infty))$, and $f_1(x)=0$ for $x=\frac{\mu+\sqrt{\mu^2+2q}}{2q}$; while $f_2(x)>(<)\,0$ on $[0,\frac{\mu}{2q})$ $((\frac{\mu}{2q},\infty))$, and $f_2(x)=0$ for $x=\frac{\mu}{2q}$.
There are two cases to be considered.
\begin{itemize}
\item[(1)]
$\mu>0$, and hence $0<\frac{\mu}{2q}<\frac{\mu+\sqrt{\mu^2+2q}}{2q}$.

First, choose $\xi$ in such a way that
\begin{itemize}
\item[$\circ$]
$\xi(0)=0, \lim\limits_{x\rightarrow\infty}\bar{\xi}(x)=\infty, \xi'(0)\leq0$;
\item[$\circ$]
$\xi'(x)<1,\,\,\forall\,\, x\in[0,\infty)$;
\item[$\circ$]
$\xi''(x)\leq0,\,\,\forall\,\, x\in[0,W_{\xi}^{-1}(\frac{\mu}{2q})]$ and \,$
\xi''(x)\geq0,\,\,\forall\,\, x\in(W_{\xi}^{-1}(\frac{\mu}{2q}),\infty)$.
\end{itemize}
One can observe the following facts
\begin{itemize}
\item[$\bullet$]
$\bar{\xi}(x)$ is strictly increasing over $[0,\infty)$ (because $\bar{\xi}'(x)=1-\xi'(x)>0,\,\,\forall\,\, x\in[0,\infty)$);
\item[$\bullet$]
$W_\xi(z)$ is strictly increasing (since $\frac{W^{(q)}(z)} {W^{(q)'}(z)}$ and $\bar{\xi}(z)$ are both strictly increasing) with supremum $W_\xi(\infty)=\frac{\mu+\sqrt{\mu^2+2q}}{2q}$;
\item[$\bullet$]
$\lim\limits_{x\rightarrow\infty}g(x)=\lim\limits_{x\rightarrow\infty}(\xi'(x)+1-\xi'(x))=1>0$,\,\,$
\lim\limits_{x\rightarrow0}g(x)=\xi'(0)\leq0$;
\item[$\bullet$]
$g'(x)=\xi''(x)f_1(W_{\xi}(x))+(\xi'(x)-1)W_{\xi}'(x)f_2(W_{\xi}(x))<0,\,\,\forall\,\, x\in[0,W_{\xi}^{-1}(\frac{\mu}{2q}))$;
\item[$\bullet$]
$\xi'(W_{\xi}^{-1}(\frac{\mu}{2q}))<0$, \,$g(W_{\xi}^{-1}(\frac{\mu}{2q}))
=\xi'(W_{\xi}^{-1}(\frac{\mu}{2q}))+(1-\xi'(W_{\xi}^{-1}(\frac{\mu}{2q})))(-\frac{\mu^2}{2q})<0$;
\item[$\bullet$]
$g'(x)=\xi''(x)f_1(W_{\xi}(x))+(\xi'(x)-1)W_{\xi}'(x)f_2(W_{\xi}(x))>0,\,\,\forall\,\, x\in(W_{\xi}^{-1}(\frac{\mu}{2q}),\infty)$.
\end{itemize}
Based on these observations, one can deduce that there should exist an $x_0\in(W_{\xi}^{-1}(\frac{\mu}{2q}),\infty)$ such that $g(x)\leq0$ for $x\leq x_0$, and $g(x)>0$ for $x>x_0$. Thus Assumption \ref{assumption1} holds true.

However, the above class of general draw-down function $\xi$ may seem to be restrictive because conditions are imposed on the second derivative of $\xi$, in the following we are devoted to construct a more general class of general draw-down functions. It is seen that
\begin{eqnarray}\label{equuivalence}
g(x)\geq0&\Leftrightarrow&\frac{1}{1-\xi'(x)}-1\geq-\frac{W^{(q)''}(\bar{\xi}(x))W^{(q)}(\bar{\xi}(x))}{\left(W^{(q)'}(\bar{\xi}(x))\right)^{2}}
\nonumber\\
&\Leftrightarrow&
\frac{1}{1-\xi'(\bar{\xi}^{-1}(z))}-1\geq-\frac{W^{(q)''}(z)W^{(q)}(z)}{\left(W^{(q)'}(z)\right)^{2}},\,\,z=\bar{\xi}(x)
\nonumber\\
&\Leftrightarrow&
\left(\bar{\xi}^{-1}\right)'(z)-1\geq-\frac{W^{(q)''}(z)W^{(q)}(z)}{\left(W^{(q)'}(z)\right)^{2}},\,\,z=\bar{\xi}(x)
\nonumber\\
&\Leftrightarrow&
\left(\bar{\xi}^{-1}\right)'(z)-1\geq
2\mu W_{0}(z)-2qW_{0}(z)^2,\,\,z=\bar{\xi}(x),
\end{eqnarray}
provided that $\xi'(x)<1$ for all $x\geq0$ \emph{(}hence, $\bar{\xi}^{-1}$ is well defined\emph{)}.
Recalling that $W^{(q)''}(z)$ takes positive values for large $z$ \emph{(}In fact, $W_{\xi}^{-1}(\frac{\mu}{q})$ is the unique zero of the second derivative $W^{(q)''}(z)$ with $W^{(q)''}(z)< (>)\,0$ over $[0, W_{\xi}^{-1}(\frac{\mu}{q})) \,((W_{\xi}^{-1}(\frac{\mu}{q}),\infty))$\emph{)}, we can
also choose $\xi$ satisfying
\begin{align*}
&\circ~0\leq\xi'(x)<1, \,\,\forall\,\,x\in[0,\infty);\\
&\circ~\bar{\xi}(0)=\underline{d}\in(0,\infty),\,\,\bar{\xi}(\infty)=\overline{d}\in(0,\infty], \,\,a\in\left[\underline{d},\overline{d}\right];\\
&\circ~\left(\bar{\xi}^{-1}\right)'(z)-1<2\mu W_{0}(z)-2qW_{0}(z)^2,\,\,\,\forall\,\,\,z\in\left[\underline{d},a\right);\\
&\circ~\left(\bar{\xi}^{-1}\right)'(z)-1\geq2\mu W_{0}(z)-2qW_{0}(z)^2,\,\,\,\forall\,\,\,z\in\left[a,\overline{d}\right),
\end{align*}
such that Assumption \ref{assumption1} holds true with $x_{0}=\bar{\xi}^{-1}(a)$. In particular, if $a=\overline{d}$ then $g(x)<0$ for all $x\in[0,\infty)$;
if $a=\underline{d}$ then $g(x)\geq0$ for all $x\in[0,\infty)$.
To see that such a construction is feasible, it would be quite useful to note that $\xi'(x)\in[0,1)$ \emph{(}or, equivalently, $\bar{\xi}'(x)\in(0,1]$\emph{)} over $x\in[0,\infty)$ can result in $\left(\bar{\xi}^{-1}\right)'(z)\in[1,\infty)$ over $z\in[0,\infty)$. By the way, the function in the right hand side of the final inequality of \eqref{equuivalence}, does not depend on $\xi$.

In particular, we can fix $d_{1}\in[0,\infty)$, $a\in\left[d_{1},\infty\right)$ and $M_{1}\in[1,\infty)$, and then choose $\xi$ satisfying
\begin{align*}
&\circ~\,\bar{\xi}^{-1}(d_{1})=0;\\
&\circ~\left(\bar{\xi}^{-1}\right)'(z)\in[1,M_{1}), \,\,\,\,z\in[d_{1},\infty);\\
&\circ~\left(\bar{\xi}^{-1}\right)'(z)-1<2\mu W_{0}(z)-2qW_{0}(z)^2,\,\,\,\forall\,\,\,z\in\left[d_{1},a\right);\\
&\circ~\left(\bar{\xi}^{-1}\right)'(z)-1\geq2\mu W_{0}(z)-2qW_{0}(z)^2,\,\,\,\forall\,\,\,z\in\left[a,\infty\right),
\end{align*}
to fulfill Assumption \ref{assumption1} with $x_{0}=\bar{\xi}^{-1}(a)$. Here, we should note that the function in the right hand side of the final inequality of \eqref{equuivalence} is bounded, and $\bar{\xi}(\infty)=\infty$ by construction.

\item[(2)] \,$\mu\leq0$.
First, consider choosing $\xi$ satisfying
\begin{align*}
&\circ~\xi'(x)<1, \,\,\forall\,\,x\in[0,\infty);\\
&\circ~\xi(0)=0,\lim\limits_{x\rightarrow\infty}\bar{\xi}(x)=\infty,\xi'(0)=0;\\
&\circ~\xi''(x)\geq0, \,\,\forall\,\,x\in[0,\infty),
\end{align*}
such that
$\bar\xi(x)$ is strictly increasing on $[0,\infty)$; $g(0)=0$;
and $g'(x)=\xi''(x)f_1(W_{\xi}(x))+(\xi'(x)-1)W_{\xi}'(x)f_2(W_{\xi}(x))>0$ \emph{(}and, hence, $g(x)>0$\emph{)} for all $x\in[0,\infty)$.
Hence Assumption \ref{assumption1} should hold true with $x_0=0$.

Due to the fact that
\begin{align*}
g(x)\geq0&\Leftrightarrow\xi'(x)+(1-\xi'(x))(2qW_{\xi}(x)^2-2\mu W_{\xi}(x))\geq0,
\end{align*}
one can also just choose $\xi$ such that $\xi'(x)\in[0,1)$ for $x\in[0,\infty)$, so that Assumption \ref{assumption1} holds true with $x_0=0$.
\end{itemize}

\end{exa}

\begin{exa}
Let $X_{t}=x+pt-\sum\limits_{i=1}^{N_{t}}e_{i}$ where $p>0$, $\{e_{i};i\geq1\}$ are i.i.d. exponential  random variables  with
mean ${1}/{\mu}$ and $\{N_{t}, t\geq0\}$ is an independent Poisson process with  intensity $\lambda$.
The $q$-scale function of $X$ is
\begin{eqnarray}\label{Exacra}
W^{(q)}(x)=p^{-1}\left(A_{+}e^{\theta_{+}x}
-A_{-}e^{\theta_{-}x}\right),\,x\geq0,\nonumber
\end{eqnarray}
where $A_{\pm}=\frac{\mu+\theta_{\pm}}{\theta_{+}-\theta_{-}}>0, \,\,\,\,\,\,
\theta_{\pm}=\frac{q+\lambda-\mu p\pm\sqrt{(q+\lambda-\mu p)^{2}+4pq\mu}}{2p}$.
Some algebraic manipulations yield
$$W^{(q)''}(x)=p^{-1}(A_{+}(\theta_{+})^{2}e^{\theta_{+}x}-A_{-}(\theta_{-})^{2}e^{\theta_{-}x})=-\theta_{-}\theta_{+}W^{(q)}(x)+(\theta_{+}+\theta_{-})W^{(q)'}(x).$$
Hence, $g(x)$ can be rewritten as
\begin{align*}
g(x)&=\xi'(x)+(1-\xi'(x))\frac{W^{(q)''}(\bar{\xi}(x))W^{(q)}(\bar{\xi}(x))}{(W^{(q)'}(\bar{\xi}(x)))^2}\\
&=\xi'(x)+(1-\xi'(x))(-\theta_{+}\theta_{-}W_{\xi}(x)^2+(\theta_{+}+\theta_{-}) W_{\xi}(x)).
\end{align*}
One can verify that
\begin{eqnarray}\label{ggeq0.1}
g(x)\geq0\hspace{-0.3cm}&\Leftrightarrow&
\hspace{-0.3cm}\xi'(x)+(1-\xi'(x))(-\theta_{+}\theta_{-}W_{\xi}(x)^2+(\theta_{+}+\theta_{-}) W_{\xi}(x))\geq0
\nonumber\\
\hspace{-0.3cm}&\Leftrightarrow&\hspace{-0.3cm} \xi'(x)\geq\frac{\theta_{+}\theta_{-}W_{\xi}(x)^2-(\theta_{+}+\theta_{-}) W_{\xi}(x)}{1-\big(-\theta_{+}\theta_{-}W_{\xi}(x)^2+(\theta_{+}+\theta_{-}) W_{\xi}(x)\big)}
\nonumber\\
\hspace{-0.3cm}&&\hspace{0.3cm}\,\,\,:=\frac{-f_{3}(W_{\xi}(x))}{1-f_{3}(W_{\xi}(x))}=1-\frac{1}{1-f_{3}(W_{\xi}(x))},
\end{eqnarray}
with $f_{3}(x)=-\theta_{+}\theta_{-}x^2+(\theta_{+}+\theta_{-})x$ which is non-positive and decreasing over $[0,0\vee \frac{\theta_{+}+\theta_{-}}{2\theta_{+}\theta_{-}})$, and increasing over $[0\vee \frac{\theta_{+}+\theta_{-}}{2\theta_{+}\theta_{-}},\infty)$. Because $W_\xi(z)$ is strictly increasing with upper bound $\frac{1}{\theta_{+}}$ \emph{(}can not be attained\emph{)} and lower bound $0$, hence $1-f_{3}(W_{\xi}(x))>0$ and hence $1-\frac{1}{1-f_{3}(W_{\xi}(x))}<1$ for all $x\in[0,\infty)$.
Hence, we can choose $\xi$ satisfying
$$\circ~\xi'(x)\geq1-\frac{1}{1-f_{3}\left(0\vee\frac{\theta_{+}+\theta_{-}}{2\theta_{+}\theta_{-}}\right)},\,\,\forall\,\,\,x\in[0,\infty),$$
in order to guarantee that $g(x)\geq0$ for all $x\in[0,\infty)$, in which case Assumption \ref{assumption1} holds true with $x_0=0$.

One can also follow a very similar argument as the first constructing method for case $\mu>0$ in Example \ref{exa1}, to construct the appropriate class of general draw-down functions fulfilling Assumption \ref{assumption1}, by imposing conditions on the first and second derivatives of $\xi$.

Because \eqref{ggeq0.1} can be equivalent to
\begin{eqnarray}\label{ggeq00.1}
g(x)\geq0
\hspace{-0.3cm}&\Leftrightarrow&\hspace{-0.3cm} \xi'\left(\bar{\xi}^{-1}(z)\right)\geq\frac{-f_{3}(W_{0}(z))}{1-f_{3}(W_{0}(z))},\,z=\bar{\xi}(x)
\nonumber\\
\hspace{-0.3cm}&\Leftrightarrow&\hspace{-0.3cm} 1-\left(\frac{1}{1-\xi'\left(\bar{\xi}^{-1}(z)\right)}\right)^{-1}\geq\frac{-f_{3}(W_{0}(z))}{1-f_{3}(W_{0}(z))},\,z=\bar{\xi}(x)
\nonumber\\
\hspace{-0.3cm}&\Leftrightarrow&\hspace{-0.3cm} 1-\left(\left(\bar{\xi}^{-1}\right)'(z)\right)^{-1}\geq\frac{-f_{3}(W_{0}(z))}{1-f_{3}(W_{0}(z))},\,z=\bar{\xi}(x),
\end{eqnarray}
provided that $\bar{\xi}^{-1}$ is well defined.
Thus, we can
also choose $\xi$ satisfying
\begin{align*}
&\circ~\xi'(x)<1, \,\,\forall\,\,x\in[0,\infty);\\
&\circ~\bar{\xi}(0)=\underline{d}\in(0,\infty),\,\,\bar{\xi}(\infty)=\overline{d}\in(0,\infty], \,\,a\in\left[\underline{d},\overline{d}\right];\\
&\circ~1-\left(\left(\bar{\xi}^{-1}\right)'(z)\right)^{-1}<\frac{-f_{3}(W_{0}(z))}{1-f_{3}(W_{0}(z))},\,\,\,\forall\,\,\,z\in\left[\underline{d},a\right);\\
&\circ~1-\left(\left(\bar{\xi}^{-1}\right)'(z)\right)^{-1}\geq\frac{-f_{3}(W_{0}(z))}{1-f_{3}(W_{0}(z))},\,\,\,\forall\,\,\,z\in\left[a,\overline{d}\right),
\end{align*}
such that Assumption \ref{assumption1} holds true with $x_{0}=\bar{\xi}^{-1}(a)$. To understand the feasibility of such a construction, it would be quite useful to note that $\xi'(x)\in(-\infty,1)$ \emph{(}or, equivalently, $\bar{\xi}'(x)\in(0,\infty)$\emph{)} over $x\in[0,\infty)$ can lead to $\left(\bar{\xi}^{-1}\right)'(z)\in(0,\infty)$ over $z\in[0,\infty)$. By the way, the function in the right hand side of the final inequality of \eqref{ggeq00.1}, does not depend on $\xi$.

In particular, we can fix $d_{2}\in[0,\infty)$, $a\in\left[d_{2},\infty\right)$ and $M_{2}\in(0,\infty)$, and then choose $\xi$ satisfying
\begin{align*}
&\circ~\,\bar{\xi}^{-1}(d_{2})=0;\\
&\circ~\left(\bar{\xi}^{-1}\right)'(z)\in(0,M_{2}), \,\,\,\,z\in[d_{2},\infty);\\
&\circ~1-\left(\left(\bar{\xi}^{-1}\right)'(z)\right)^{-1}<\frac{-f_{3}(W_{0}(z))}{1-f_{3}(W_{0}(z))},\,\,\,\forall\,\,\,z\in\left[d_{2},a\right);\\
&\circ~1-\left(\left(\bar{\xi}^{-1}\right)'(z)\right)^{-1}\geq\frac{-f_{3}(W_{0}(z))}{1-f_{3}(W_{0}(z))},\,\,\,\forall\,\,\,z\in\left[a,\infty\right),
\end{align*}
to fulfill Assumption \ref{assumption1} with $x_{0}=\bar{\xi}^{-1}(a)$. Here, we should note that the function in the right hand side of the final inequality of \eqref{ggeq00.1} is  bounded, and $\bar{\xi}(\infty)=\infty$ by construction.

\end{exa}

\begin{exa}
In this example we consider a general spectrally negative L\'evy process $X$. It still holds that
\begin{eqnarray}\label{equuivalence01}
g(x)\geq0&\Leftrightarrow&
\left(\bar{\xi}^{-1}\right)'(z)-1\geq-\frac{W^{(q)''}(z)W^{(q)}(z)}{\left(W^{(q)'}(z)\right)^{2}},\,\,z=\bar{\xi}(x),\nonumber
\end{eqnarray}
provided that $\xi'(x)<1$ for all $x\geq0$.
We can choose $\xi$ satisfying
\begin{align*}
&\circ~\xi'(x)\in(-\infty,1), \,\,\forall\,\,x\in[0,\infty);\\
&\circ~\bar{\xi}(0)=\underline{d}\in(0,\infty),\,\,\bar{\xi}(\infty)=\overline{d}\in(0,\infty], \,\,a\in\left[\underline{d},\overline{d}\right];\\
&\circ~\left(\bar{\xi}^{-1}\right)'(z)-1<(>)-\frac{W^{(q)''}(z)W^{(q)}(z)}{\left(W^{(q)'}(z)\right)^{2}},\,\,\,\forall\,\,\,z\in\left[\underline{d},a\right);\\
&\circ~\left(\bar{\xi}^{-1}\right)'(z)-1\geq(\leq)-\frac{W^{(q)''}(z)W^{(q)}(z)}{\left(W^{(q)'}(z)\right)^{2}},\,\,\,\forall\,\,\,z\in\left[a,\overline{d}\right),
\end{align*}
to guarantee that Assumption \ref{assumption1} holds true with $x_{0}=\bar{\xi}^{-1}(a)$.
It would also be useful to note that $\xi'(x)\in(-\infty,1)$ \emph{(}or, equivalently, $\bar{\xi}'(x)\in(0,\infty)$\emph{)} over $x\in[0,\infty)$ can result in $\left(\bar{\xi}^{-1}\right)'(z)\in(0,\infty)$ over $z\in[0,\infty)$. In essence, appropriate class of general draw-down functions satisfying Assumption \ref{assumption1} are constructed via imposing conditions of the inverse function $\bar{\xi}^{-1}$ instead of on $\xi$ or $\bar{\xi}$.

One can also follow the method in Wang and Zhou (2018) to construct the class of general draw-down functions fulfilling Assumption \ref{assumption1} with $x_{0}=\bar{\xi}^{-1}(a)$ and $a:=\inf\{z\geq 0: W{(q)''}(z)> 0 \}$, for  the spectrally negative L\'evy process whose L\'evy measure has a completely monotone density.
\end{exa}

    %%%%%%%%%%%%%%%%%%%%%%%%%%%%%%%%%%%%%%%%%%%%%%%%%%%%%%%%%%%%%%%%
  \section{Numerical analysis}
  %%%%%%%%%%%%%%%%%%%%%%%%%%%%%%%%%%%%%%%%%%%%%%%%%%%%%%%%%%%%%%%%

In this section we provide some numerical examples to illustrate the theoretical results obtained in the previous sections.
We consider a linear draw-down function
$$
\xi(x)=kx-d,\qquad k<1,\ d\geq 0.
$$
Under this assumption, we have, for $x<y$,
$$
\int_{x}^{y}\frac{W^{(q)'}(u-\xi(u))}{W^{(q)}(u-\xi(u))}\mathrm{d}u
=\frac{1}{1-k}\ln\left(
\frac{W^{(q)}((1-k)y+d) }{W^{(q)}((1-k)x+d)}
\right),
$$
which has frequently appeared in formulae in Section 4.

For the risk process, it is assumed to be a linear Brownian motion:
\begin{equation}\label{BM}
X(t)=\mu t+\sigma B(t),\qquad t\geq 0,
\end{equation}
where $\mu\in \mathbb{R}$, $\sigma>0$, and $\{B(t)\}$ is a standard Brownian motion.
Note that our main results are all expressed in terms of the $q$-scale function $W^{(q)}$. By Kyprianou (2006) we know that that the $q$-scale function for the linear Brownian motion \eqref{BM} is given by
\begin{equation}\label{W-BM}
W^{(q)}(x)=\frac{2}{\sigma^2\Xi} \mathrm{e}^{-\frac{\mu}{\sigma^2} x}\sinh(\Xi x),\qquad x\geq 0,
\end{equation}
where
$
\Xi=\frac{\sqrt{ \mu^2+2q\sigma^2 } }{ \sigma^2}.
$
Throughout this paper, we set  $\mu=0.03, \sigma=0.4, q=0.01$.

It follows from Proposition 5 and 6 that  Assumption 1 plays an important role in characterizing the optimal return function. It is seen that
$$
\lim_{x\rightarrow \infty} \frac{W^{(q)''}(x-\xi(x))W^{(q)}(x-\xi(x))}
{\left(W^{(q)'}(x-\xi(x))\right)^{2}}=1,
$$
which yields $\lim_{x\rightarrow \infty}g(x)=1$.  In Figure 1, we plot the behavior of $g(x)$ for $d=1$ and $k=-10, -5, -1, 0, 0.1$. It follows that the function $g$ indeed converges to $1$ as $x\rightarrow \infty$, and we also find that the function $g$ indeed changes its sign at some point $x_0$.

\begin{figure}[!htbp]\label{fig-7}
\centering
{\scalebox{0.8}[0.85]{\includegraphics*[60,339][434,625]{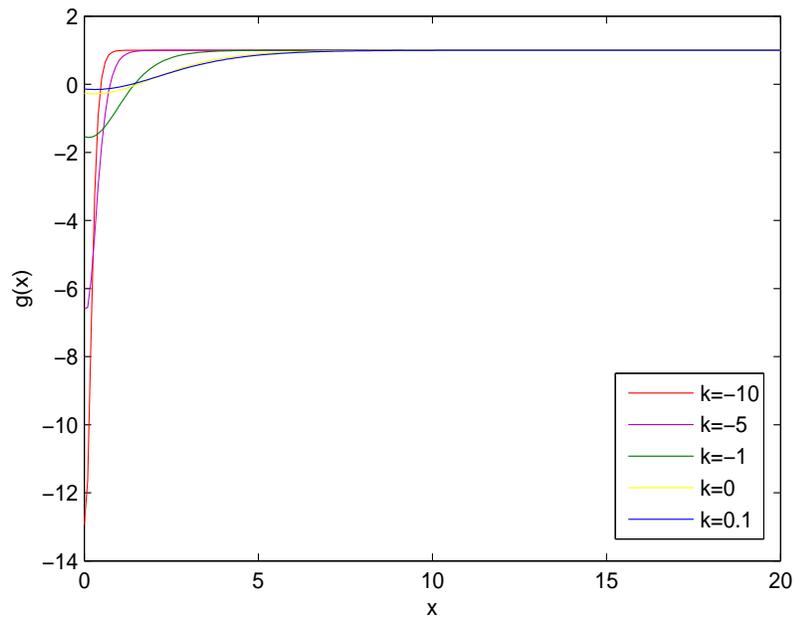}}}\hspace{.3in}
\caption{Behavior of the function $g$ for $d=1$ and $k=-10, -5, -1, 0, 0.1$.}
\end{figure}

Now for the draw-down parameters $(k, d)$, we shall consider two cases: $(0.1, 1)$ and $(-1, 1)$.  Furthermore, we set $(\gamma_1, \gamma_2)=(0.2, 0.6)$.
First, we consider the conditions in Proposition 5.  Besides Assumption 1, we also need to check inequality \eqref{condition1}.
It follows from Figure 1 that $g(y)>0$ for $y>x_0$.  Hence, to check inequality \eqref{condition1}, we only need to check the following condition
\begin{equation}\label{G2-condition}
G_2(x)\geq 0,\qquad \forall\ 0<x\leq x_0.
\end{equation}

\begin{figure}[!htbp]\label{fig-7}
\centering
{\scalebox{0.8}[0.85]{\includegraphics*[90,314][451,597]{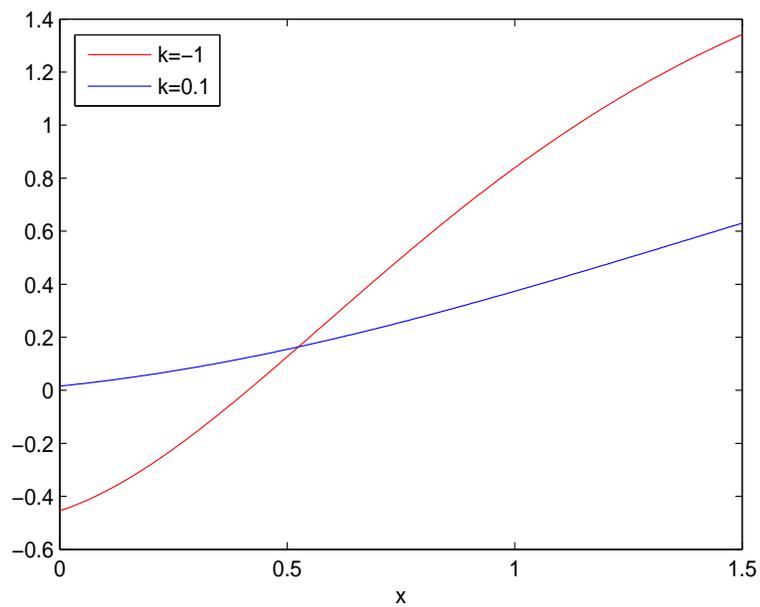}}}\hspace{.3in}
\caption{The function $G_2$ for $d=1$ and $k=-10, -5, -1, 0, 0.1$.}
\end{figure}

For $(k, d)=(0.1, 1)$, we can obtain $x_0=1.360$; for $(k,d)=(-1,1)$, we can obtain $x_0=1.443$.
 In Figure 2, we plot $G_2$ as a function of $x$ for $0\leq x \leq 1.5$. When $k=0.1$, we find that $G_2(x)$ is strictly positive for $0\leq x\leq x_0$. It follows from Proposition 5 that the maximizing strategy function is
$\gamma(x)\equiv\gamma_{2},\,\forall\,\,x\geq0$, and the optimal return function is given by
 formula \eqref{represntation of the optimal function1}. In Figure 3, we plot the optimal return value as a function of the initial surplus level $x$. We observe that $f(x)$ is an increasing function and converges to a constant value $2.8209$. The limit can be checked mathematically as follows. Using formula \eqref{fgamma2} we have
$$
 f(x)=\frac{W^{(q)}(x-\xi(x))}{W^{(q)'}(x-\xi(x))}-\int_{x}^{\infty}\exp\{-\frac{1}{1-\gamma_{2}}\int_{x}^{y}\frac{W^{(q)'}(u-\xi(u))}{W^{(q)}(u-\xi(u))}\mathrm{d}u\}
g(y)\mathrm{d}y,
$$
which, together with $g(y)\rightarrow 1$ as $y\rightarrow \infty$, yields for large $x$,
\begin{eqnarray}
f(x)
&\approx & \frac{W^{(q)}(x-\xi(x))}{W^{(q)'}(x-\xi(x))}-\frac{1-\gamma_2}{\gamma_2}f(x),\nonumber
\end{eqnarray}
where we have used \eqref{represntation of the optimal function1},  \eqref{lim-W}, \eqref{lim.for.det.C01} and the L'H\^{o}pital's rule to deduce
\begin{eqnarray}
&&
\lim\limits_{x\rightarrow\infty}\frac{\int_{x}^{\infty}\exp\{-\frac{1}{1-\gamma_{2}}\int_{x}^{y}\frac{W^{(q)'}(u-\xi(u))}{W^{(q)}(u-\xi(u))}\mathrm{d}u\}
g(y)\mathrm{d}y}{f(x)}
\nonumber\\
&
=&
\lim\limits_{x\rightarrow\infty}
\frac{-g(x)+\frac{1}{1-\gamma_{2}}\frac{W^{(q)'}(x-\xi(x))}{W^{(q)}(x-\xi(x))}
\int_{x}^{\infty}\exp\{-\frac{1}{1-\gamma_{2}}\int_{x}^{y}\frac{W^{(q)'}(u-\xi(u))}{W^{(q)}(u-\xi(u))}\mathrm{d}u\}
g(y)\mathrm{d}y}
{-\frac{\gamma_{2}}{1-\gamma_{2}}+\frac{\gamma_{2}}{1-\gamma_{2}}\frac{1}{1-\gamma_{2}}\frac{W^{(q)'}(x-\xi(x))}{W^{(q)}(x-\xi(x))}
\int_{x}^{\infty}\exp\{-\frac{1}{1-\gamma_{2}}\int_{x}^{y}\frac{W^{(q)'}(u-\xi(u))}{W^{(q)}(u-\xi(u))}\mathrm{d}u\}
\mathrm{d}y}
\nonumber\\
&
=&\frac{1-\gamma_{2}}{\gamma_{2}}.
\nonumber
\end{eqnarray}
Hence, for large $x$ we have
$$
f(x)\approx \gamma_2\frac{W^{(q)}(x-\xi(x))}{W^{(q)'}(x-\xi(x))}=
0.6\frac{W^{(q)}(0.9x+1)}{W^{(q)'}(0.9x+1)}\rightarrow \frac{0.6}{\Phi(0.01)}=2.8209,\ \ \text{as}\ x\rightarrow \infty,
$$
where the last step follows from formula \eqref{lim-W}.

As for $k=-1$, after observing Figure 2 we find that $G_2(x)$ is not always positive for $0\leq x\leq x_0$. Hence, the optimal return function is not characterized by Proposition 5. However, it is seen that the conditions in case (\emph{v}) are satisfied, then the optimal return function is given by formulae in \eqref{4.13}. In Figure 4 we plot the optimal return function for $0\leq x\leq 20$. In this case, we find that $f(x)$ converges to the same limit as the case when $k=0.1$, which is due to that the optimal return function has the same form when $x$ is large, and the corresponding limit $2.8209$ is independent of the parameter $k$.

\begin{figure}[!htbp]\label{fig-7}
\centering
{\scalebox{0.9}[0.85]{\includegraphics*[93,332][421,606]{f(gamma2=0.6,k=0.1).eps}}}\hspace{.3in}
\caption{The optimal return function when $d=1$, $k=0.1$, $\gamma_1=0.2$ and $\gamma_2=0.6$.}
\end{figure}

\begin{figure}[!htbp]\label{fig-7}
\centering
{\scalebox{0.8}[0.8]{\includegraphics*[47,352][415,628]{f(k=-1).eps}}}\hspace{.3in}
\caption{The optimal return function when $d=1$, $k=-1$, $\gamma_1=0.2$ and $\gamma_2=0.6$.}
\end{figure}

The patterns of the optimal return functions shown in Figure 3 and Figure 4 reveal that, the smaller the parameter $k$, the larger the optimal return function. This is because, for each admissible tax strategy $\gamma\in\Gamma$, the draw-down time $\tau^{\gamma}_{\xi}$ decreases almost surely as $k$ increases, bearing in mind Equation (\ref{definition of draw-down times}) and $\xi(x)=kx-d$ with $d=1$ set in Figure 3 and Figure 4.
%In addition, the set of admissible taxation strategies $\Gamma$ shrinks as $k$ increases: an admissible taxation strategy for larger $k$ is still admissible for smaller $k$, while an admissible tax strategy for smaller $k$ may no longer be admissible for larger $k$.
These reasonings together with the definition of the optimal return function given by (\ref{optimal function}) explain why smaller $k$ gives larger optimal return function value, i.e., $f(x)|_{k=k_{1}}\leq f(x)|_{k=k_{2}}$ for $x\in[0,\infty)$ and $k_{2}\leq k_{1}<1$.

    %%%%%%%%%%%%%%%%%%%%%%%%%%%%%%%%%%%%%%%%%%%%%%%%%%%%%%%%%%%%%%%%
  \section*{Appendix}
  %%%%%%%%%%%%%%%%%%%%%%%%%%%%%%%%%%%%%%%%%%%%%%%%%%%%%%%%%%%%%%%%

\subsection{Proof of Proposition \ref{exit problem}}

\begin{proof}
On the set $\{\tau_{a}^{+}<\tau^{\gamma}_{\xi}\}$, the taxed process $\{U^{\gamma}(t);t\geq0\}$ up-crosses level $a$ before the general draw-down time $\tau^{\gamma}_{\xi}$, i.e.,
\begin{eqnarray}
U^{\gamma}(t)\geq \xi\left(\bar{U}^{\gamma}(t)\right)=\xi\left(\sup\limits_{0\leq s\leq t}U^{\gamma}(s)\right) \mbox{ for all } t\in[0,\tau_{a}^{+}].\nonumber
\end{eqnarray}
Hence, we can rewrite the set $\{\tau_{a}^{+}<\tau^{\gamma}_{\xi}\}$ as follows,
\begin{eqnarray}\label{equivalent set relation 1}
\{\tau_{a}^{+}<\tau^{\gamma}_{\xi}\}=\{\bar{\varepsilon}_{t}\leq \overline{\gamma}_{x}(x+t)-\xi\left(\overline{\gamma}_{x}(x+t)\right),\qquad \forall t\in[0,\overline{\gamma}_{x}^{-1}(a)-x]\}.
\end{eqnarray}
Here we have used the fact that, for any $ s\in[0,L^{-1}(t)-L^{-1}(t-)]$,
\begin{eqnarray}
U^{\gamma}(L^{-1}(t-)+s)\hspace{-0.3cm}&=&\hspace{-0.3cm}
X(L^{-1}(t-)+s)-\int_{0}^{L^{-1}(t-)+s}\gamma\left(\bar{X}(u)\right)\mathrm{d}\bar{X}(u)
\nonumber\\
\hspace{-0.3cm}&
=&\hspace{-0.3cm}
X(L^{-1}(t-)+s)-\int_{0}^{L^{-1}(t-)}\gamma\left(\bar{X}(u)\right)\mathrm{d}\bar{X}(u)
,\nonumber
\end{eqnarray}
which implies that $U^{\gamma}\left(L^{-1}(t-)+s\right)\geq \xi\left(\sup\limits_{0\leq u\leq L^{-1}(t-)+s}U^{\gamma}(u)\right)=\xi\left(U^{\gamma}\left(L^{-1}(t)\right)\right)$ is equivalent to
\begin{eqnarray}
\varepsilon_{t}(s)=X(L^{-1}(t))-X(L^{-1}(t-)+s)\hspace{-0.3cm}&=&\hspace{-0.3cm}U^{\gamma}(L^{-1}(t))-U^{\gamma}(L^{-1}(t-)+s)\nonumber\\
\hspace{-0.3cm}&\leq&\hspace{-0.3cm} U^{\gamma}(L^{-1}(t))-\xi\left(U^{\gamma}\left(L^{-1}(t)\right)\right)=\overline{\gamma}_{x}(x+t)-\xi\left(\overline{\gamma}_{x}(x+t)\right)\nonumber
\end{eqnarray}
holds true for all $ s\in[0,L^{-1}(t)-L^{-1}(t-)]$.
By (\ref{equivalent set relation 1}) we have,
\begin{eqnarray} \label{vff2.8}
\mathbb{P}_{x}\left(\{\tau_{a}^{+}<\tau^{\gamma}_{\xi}\}\right)\hspace{-0.3cm}&=&\hspace{-0.3cm}\mathbb{P}_{x}\left(\{\bar{\varepsilon}_{t}\leq \overline{\gamma}_{x}(x+t)-\xi\left(\overline{\gamma}_{x}(x+t)\right), \forall t\in[0,\overline{\gamma}_{x}^{-1}(a)-x]\}\right)\nonumber\\
\hspace{-0.3cm}&=&\hspace{-0.3cm}\exp\{-\int_{0}^{\overline{\gamma}_{x}^{-1}(a)-x}n(\bar{\varepsilon}>\overline{\gamma}_{x}(x+t)-\xi\left(\overline{\gamma}_{x}(x+t)\right))\mathrm{d}t\}
\nonumber\\
\hspace{-0.3cm}&=&\hspace{-0.3cm}\exp\{-\int_{0}^{\overline{\gamma}_{x}^{-1}(a)-x}\frac{W'(\overline{\gamma}_{x}(x+t)-\xi\left(\overline{\gamma}_{x}(x+t)\right))}
{W(\overline{\gamma}_{x}(x+t)-\xi\left(\overline{\gamma}_{x}(x+t)\right))}\mathrm{d}t\}
\nonumber\\
\hspace{-0.3cm}&=&\hspace{-0.3cm}\exp\{-\int_{x}^{a}\frac{W'(y-\xi\left(y\right))}
{W(y-\xi\left(y\right))}\frac{1}{1-\gamma\left(\overline{\gamma}_{x}^{-1}(y)\right)}\mathrm{d}y\},
\end{eqnarray}
where in the last but one equality we used $n(\bar{\varepsilon}>x)=\frac{W'(x)}{W(x)}$ providing that $x$ is not a point of discontinuity in the derivative of $W$ (see, Kyprianou (2006)), and in the last step we have changed variables using $y=\overline{\gamma}_{x}(x+t)$, which yields $\mathrm{d}t=\frac{1}{1-\gamma(x+t)}\mathrm{d}y=\frac{1}{1-\gamma\left(\overline{\gamma}_{x}^{-1}(y)\right)}\mathrm{d}y$.

Note that, on the one hand we have
\begin{eqnarray}\label{probability1 change of measure}
\mathbb{P}_{x}^{\Phi(q)}\left(\tau_{a}^{+}<\tau^{\gamma}_{\xi}\right)
\hspace{-0.3cm}&=&\hspace{-0.3cm}\mathbb{E}_{x}\left(\mathrm{e}^{\Phi(q)(X(\tau_{a}^{+})-x)-q\tau_{a}^{+}}\mathbf{1}_{\{\tau_{a}^{+}<\tau^{\gamma}_{\xi}\}}\right)\nonumber\\
\hspace{-0.3cm}&=&\hspace{-0.3cm}\mathrm{e}^{\Phi(q)(\overline{\gamma}_{x}^{-1}(a)-x)}\mathbb{E}_{x}\left(\mathrm{e}^{-q\tau_{a}^{+}}\mathbf{1}_{\{\tau_{a}^{+}<\tau^{\gamma}_{\xi}\}}\right);
\end{eqnarray}
on the other hand we have
\begin{eqnarray}\label{probability2 change of measure}
\hspace{-0.3cm}&&\hspace{-0.3cm}\mathbb{P}_{x}^{\Phi(q)}\left(\{\tau_{a}^{+}<\tau^{\gamma}_{\xi}\}\right)\nonumber\\
\hspace{-0.3cm}&=&\hspace{-0.3cm}\exp\{-\int_{x}^{a}\frac{W_{\Phi(q)}'(y-\xi\left(y\right))}
{W_{\Phi(q)}(y-\xi\left(y\right))}\frac{1}{1-\gamma\left(\overline{\gamma}_{x}^{-1}(y)\right)}\mathrm{d}y\}\
\nonumber\\
\hspace{-0.3cm}&=&\hspace{-0.3cm}\exp\{-\int_{x}^{a}\left(\frac{W^{(q)'}(y-\xi\left(y\right))}
{W^{(q)}(y-\xi\left(y\right))}-\Phi(q)\right)\frac{1}{1-\gamma\left(\overline{\gamma}_{x}^{-1}(y)\right)}\mathrm{d}y\}
\nonumber\\
\hspace{-0.3cm}&=&\hspace{-0.3cm}\exp\{-\int_{x}^{a}\frac{W^{(q)'}(y-\xi\left(y\right))}
{W^{(q)}(y-\xi\left(y\right))}\frac{1}{1-\gamma\left(\overline{\gamma}_{x}^{-1}(y)\right)}\mathrm{d}y+\Phi(q)\int_{x}^{a} \frac{1}{1-\gamma\left(\overline{\gamma}_{x}^{-1}(y)\right)}\mathrm{d}y\}
\nonumber\\
\hspace{-0.3cm}&=&\hspace{-0.3cm}\exp\{-\int_{x}^{a}\frac{W^{(q)'}(y-\xi\left(y\right))}
{W^{(q)}(y-\xi\left(y\right))}\frac{1}{1-\gamma\left(\overline{\gamma}_{x}^{-1}(y)\right)}\mathrm{d}y+\Phi(q)\left(\overline{\gamma}_{x}^{-1}(a)-x\right)\},
\end{eqnarray}
where a similar argument as used in proving (\ref{vff2.8}) together the identity $W_{\Phi(q)}(x)=\mathrm{e}^{-\Phi(q)x}W^{(q)}(x)$ is used in the second equality, and the change of variable $z=\overline{\gamma}_{x}^{-1}(y)$ is used in the last step of the above display.

Combining (\ref{probability1 change of measure}) and (\ref{probability2 change of measure}), we arrive at
\begin{eqnarray}\label{}
\mathbb{E}_{x}\left(\mathrm{e}^{-q\tau_{a}^{+}}\mathbf{1}_{\{\tau_{a}^{+}<\tau^{\gamma}_{\xi}\}}\right)
\hspace{-0.3cm}&=&\hspace{-0.3cm}\mathrm{e}^{-\Phi(q)(\overline{\gamma}_{x}^{-1}(a)-x)}\mathbb{P}_{x}^{\Phi(q)}\left(\tau_{a}^{+}<\tau^{\gamma}_{\xi}\right)
\nonumber\\
\hspace{-0.3cm}&=&\hspace{-0.3cm}\exp\{-\int_{x}^{a}\frac{W^{(q)'}(y-\xi\left(y\right))}
{W^{(q)}(y-\xi\left(y\right))}\frac{1}{1-\gamma\left(\overline{\gamma}_{x}^{-1}(y)\right)}\mathrm{d}y\},\nonumber
\end{eqnarray}
which is our desired result \eqref{exit problem eqaution}.
\end{proof}

\subsection{Proof of Proposition \ref{expected discounted tax}}

\begin{proof}
Using (\ref{equivalent set}) it can be verified that,
\begin{eqnarray}\label{expected discounted tax identity1}
\hspace{-0.3cm}&&\hspace{-0.3cm}\mathbb{E}_{x}\left[\int_{0}^{\tau_{a}^{+}\wedge \tau^{\gamma}_{\xi}}\mathrm{e}^{-q t}\gamma\left(\bar{X}(t)\right)\mathrm{d}\left(\bar{X}(t)-x\right)\right]\nonumber\\
\hspace{-0.3cm}&=&\hspace{-0.3cm}\mathbb{E}_{x}\left[\int_{0}^{\infty}\mathbf{1}_{\{t<\tau_{a}^{+}\wedge \tau^{\gamma}_{\xi}\}}\mathrm{e}^{-q t}\gamma\left(\bar{X}(t)\right)\mathrm{d}\left(\bar{X}(t)-x\right)\right]
\nonumber\\
\hspace{-0.3cm}&=&\hspace{-0.3cm}\mathbb{E}_{x}\left[\int_{0}^{\infty}\mathbf{1}_{\left\{t<L^{-1}\left(L\left(\tau_{a}^{+}\wedge \tau^{\gamma}_{\xi}\right)-\right)\right\}}\mathrm{e}^{-q t}\gamma\left(\bar{X}(t)\right)\mathrm{d}\left(\bar{X}(t)-x\right)\right]
\nonumber\\
\hspace{-0.3cm}&=&\hspace{-0.3cm}\mathbb{E}_{x}\left[\int_{0}^{\infty}\mathbf{1}_{\left\{L(t)<L\left(\tau_{a}^{+}\wedge \tau^{\gamma}_{\xi}\right)\right\}}\mathrm{e}^{-q t}\gamma\left(L(t)+x\right)\mathrm{d}L(t)
\right]
\nonumber\\
\hspace{-0.3cm}&=&\hspace{-0.3cm}\mathbb{E}_{x}\left[\int_{0}^{\infty}\mathbf{1}_{\left\{y<L\left(\tau_{a}^{+}\wedge \tau^{\gamma}_{\xi}\right)\right\}}\mathrm{e}^{-q L^{-1}(y)}\gamma\left(y+x\right)\mathrm{d}y
\right]
\nonumber\\
\hspace{-0.3cm}&=&\hspace{-0.3cm}\int_{0}^{\infty}\mathbb{E}_{x}\left(\mathrm{e}^{-q L^{-1}(y)}\mathbf{1}_{\left\{y<L\left(\tau_{a}^{+}\wedge \tau^{\gamma}_{\xi}\right)\right\}}\right)\gamma\left(y+x\right)\mathrm{d}y
\nonumber\\
\hspace{-0.3cm}&=&\hspace{-0.3cm}\int_{0}^{\infty}\mathbb{E}_{x}\left(\mathrm{e}^{-q L^{-1}(y)}\mathbf{1}_{\left\{y<L\left(\tau_{a}^{+}\right)\wedge L\left(\tau^{\gamma}_{\xi}\right)\right\}}\right)\gamma\left(y+x\right)\mathrm{d}y
\nonumber\\
\hspace{-0.3cm}&=&\hspace{-0.3cm}\int_{0}^{\infty}\mathbb{E}_{x}\left(\mathrm{e}^{-q L^{-1}(y)}\mathbf{1}_{\left\{y<L\left(\tau_{a}^{+}\right)=\overline{\gamma}_{x}^{-1}(a)-x, y<L\left(\tau^{\gamma}_{\xi}\right)\right\}}\right)\gamma\left(y+x\right)\mathrm{d}y
\nonumber\\
\hspace{-0.3cm}&=&\hspace{-0.3cm}\int_{0}^{\overline{\gamma}_{x}^{-1}(a)-x}\mathbb{E}_{x}\left(\mathrm{e}^{-q L^{-1}(y)}\mathbf{1}_{\left\{y<L\left(\tau^{\gamma}_{\xi}\right)\right\}}\right)\gamma\left(y+x\right)\mathrm{d}y.
\end{eqnarray}
In addition, the set $\left\{y<L\left(\tau^{\gamma}_{\xi}\right)\right\}$ can be re-expressed as follows,
$$\left\{y<L\left(\tau^{\gamma}_{\xi}\right)\right\}=\left\{\bar{\varepsilon}_{t}\leq \overline{\gamma}_{x}(x+t)-\xi\left(\overline{\gamma}_{x}(x+t)\right),\forall \, t\in[0,y]\right\},$$
which together with (\ref{expected discounted tax identity1}) implies,
\begin{eqnarray}
\hspace{-0.3cm}&&\hspace{-0.3cm}\mathbb{E}_{x}\left[\int_{0}^{\tau_{a}^{+}\wedge \tau^{\gamma}_{\xi}}\mathrm{e}^{-q t}\gamma\left(\bar{X}(t)\right)\mathrm{d}\left(\bar{X}(t)-x\right)\right]\nonumber\\
\hspace{-0.3cm}&=&\hspace{-0.3cm}\int_{0}^{\overline{\gamma}_{x}^{-1}(a)-x}\mathbb{E}_{x}\left(\mathrm{e}^{-q L^{-1}(y)}\mathbf{1}_{\left\{\bar{\varepsilon}_{t}\leq \overline{\gamma}_{x}(x+t)-\xi\left(\overline{\gamma}_{x}(x+t)\right),\forall \, t\in[0,y]\right\}}\right)\gamma\left(y+x\right)\mathrm{d}y
\nonumber\\
\hspace{-0.3cm}&=&\hspace{-0.3cm}\int_{0}^{\overline{\gamma}_{x}^{-1}(a)-x}\mathrm{e}^{-\Phi(q)\left(X\left(L^{-1}(y)\right)-x\right)}\mathbb{E}_{x}\left(\mathrm{e}^{\Phi(q)\left(X\left(L^{-1}(y)\right)-x\right)-q L^{-1}(y)}\mathbf{1}_{\left\{\bar{\varepsilon}_{t}\leq \overline{\gamma}_{x}(x+t)-\xi\left(\overline{\gamma}_{x}(x+t)\right),\forall \, t\in[0,y]\right\}}\right)\gamma\left(y+x\right)\mathrm{d}y
\nonumber\\
\hspace{-0.3cm}&=&\hspace{-0.3cm}\int_{0}^{\overline{\gamma}_{x}^{-1}(a)-x}\mathrm{e}^{-\Phi(q)y}\,\mathbb{E}_{x}^{\Phi(q)}\left(\mathbf{1}_{\left\{\bar{\varepsilon}_{t}\leq \overline{\gamma}_{x}(x+t)-\xi\left(\overline{\gamma}_{x}(x+t)\right),\forall \, t\in[0,y]\right\}}\right)\gamma\left(y+x\right)\mathrm{d}y
\nonumber\\
\hspace{-0.3cm}&=&\hspace{-0.3cm}\int_{0}^{\overline{\gamma}_{x}^{-1}(a)-x}\mathrm{e}^{-\Phi(q)y}\,
\exp\{-\int_{0}^{y}n_{\Phi(q)}(\bar{\varepsilon}>\overline{\gamma}_{x}(x+t)-\xi\left(\overline{\gamma}_{x}(x+t)\right))\mathrm{d}t\}
\gamma\left(y+x\right)\mathrm{d}y
\nonumber\\
\hspace{-0.3cm}&=&\hspace{-0.3cm}\int_{0}^{\overline{\gamma}_{x}^{-1}(a)-x}\mathrm{e}^{-\Phi(q)y}\,
\exp\{-\int_{0}^{y}\frac{W_{\Phi(q)}'\left(\overline{\gamma}_{x}(x+t)-\xi\left(\overline{\gamma}_{x}(x+t)\right)\right)}{W_{\Phi(q)}\left(\overline{\gamma}_{x}(x+t)-\xi\left(\overline{\gamma}_{x}(x+t)\right)\right)}\mathrm{d}t\}
\gamma\left(y+x\right)\mathrm{d}y
\nonumber\\
\hspace{-0.3cm}&=&\hspace{-0.3cm}\int_{0}^{\overline{\gamma}_{x}^{-1}(a)-x}\mathrm{e}^{-\Phi(q)y}\,
\exp\{-\int_{0}^{y}\left(\frac{W^{(q)'}\left(\overline{\gamma}_{x}(x+t)-\xi\left(\overline{\gamma}_{x}(x+t)\right)\right)}{W^{(q)}\left(\overline{\gamma}_{x}(x+t)-\xi\left(\overline{\gamma}_{x}(x+t)\right)\right)}
-\Phi(q)\right)\mathrm{d}t\}
\gamma\left(y+x\right)\mathrm{d}y
\nonumber\\
\hspace{-0.3cm}&=&\hspace{-0.3cm}\int_{0}^{\overline{\gamma}_{x}^{-1}(a)-x}
\exp\{-\int_{0}^{y}\frac{W^{(q)'}\left(\overline{\gamma}_{x}(x+t)-\xi\left(\overline{\gamma}_{x}(x+t)\right)\right)}{W^{(q)}\left(\overline{\gamma}_{x}(x+t)-\xi\left(\overline{\gamma}_{x}(x+t)\right)\right)}\mathrm{d}t\}
\gamma\left(y+x\right)\mathrm{d}y,\nonumber
\end{eqnarray}
where $n_{\Phi(q)}$ is the excursion measure under the probability measure $\mathbb{P}_{x}^{\Phi(q)}$.
\end{proof}

\subsection{Proof of Proposition \ref{HJB}}

\begin{proof}
For any $\epsilon>0$ and $h>0$, by the definition of supremum we can always find a strategy $\gamma_{_{\epsilon}}\in \Gamma$ yielding the first inequality in (\ref{2.16}), where the stopping time $\hat{\tau}_{x+h}^{+}$ is defined via (2.1) with $\gamma$ replaced by $\gamma_{_{\epsilon}}$.
\begin{eqnarray}\label{2.16}
\hspace{-0.3cm}&&\hspace{-0.3cm}f(x)-\epsilon
=\sup\limits_{\gamma\in\Gamma}\mathbb{E}_{x}\int_{0}^{\tau^{\gamma}_{\xi}}\mathrm{e}^{-qt}\gamma\left(\bar{X}(t)\right)\mathrm{d}\bar{X}(t)-\epsilon
\nonumber\\
\hspace{-0.3cm}&\leq&\hspace{-0.3cm}
\mathbb{E}_{x}\int_{0}^{\tau^{\gamma_{_{\epsilon}}}_{\xi}}\mathrm{e}^{-qt}\gamma_{_{\epsilon}}\left(\bar{X}(t)\right)\mathrm{d}\bar{X}(t)\nonumber\\
\hspace{-0.3cm}&=&\hspace{-0.3cm}\mathbb{E}_{x}\left(\int_{0}^{\tau^{\gamma_{_{\epsilon}}}_{\xi}\wedge\hat{\tau}_{x+h}^{+}}\mathrm{e}^{-qt}\gamma_{_{\epsilon}}\left(\bar{X}(t)\right)\mathrm{d}\bar{X}(t)
+\int_{\tau^{\gamma_{_{\epsilon}}}_{\xi}\wedge\hat{\tau}_{x+h}^{+}}^{\tau^{\gamma_{_{\epsilon}}}_{\xi}}\mathrm{e}^{-qt}\gamma_{_{\epsilon}}\left(\bar{X}(t)\right)\mathrm{d}\bar{X}(t)\right)
\nonumber\\
\hspace{-0.3cm}&=&\hspace{-0.3cm}\mathbb{E}_{x}\left(\int_{0}^{\tau^{\gamma_{_{\epsilon}}}_{\xi}\wedge\hat{\tau}_{x+h}^{+}}\mathrm{e}^{-qt}\gamma_{_{\epsilon}}\left(\bar{X}(t)\right)\mathrm{d}\bar{X}(t)\right)
+\mathbb{E}_{x}\left(\int_{\hat{\tau}_{x+h}^{+}}^{\tau^{\gamma_{_{\epsilon}}}_{\xi}}\mathrm{e}^{-qt}\gamma_{_{\epsilon}}\left(\bar{X}(t)\right)\mathrm{d}\bar{X}(t)
\mathbf{1}_{\{\hat{\tau}_{x+h}^{+}
<\tau^{\gamma_{_{\epsilon}}}_{\xi}\}}\right)
\nonumber\\
\hspace{-0.3cm}&=&\hspace{-0.3cm}\mathbb{E}_{x}\left(\int_{0}^{\tau^{\gamma_{_{\epsilon}}}_{\xi}\wedge\hat{\tau}_{x+h}^{+}}\mathrm{e}^{-qt}\gamma_{_{\epsilon}}\left(\bar{X}(t)\right)\mathrm{d}\bar{X}(t)\right)
+\mathbb{E}_{x}\left(\mathbb{E}_{x}\left(\left.\int_{\hat{\tau}_{x+h}^{+}}^{\tau^{\gamma_{_{\epsilon}}}_{\xi}}\mathrm{e}^{-qt}\gamma_{_{\epsilon}}\left(\bar{X}(t)\right)\mathrm{d}\bar{X}(t)
\right|\mathcal{F}_{\hat{\tau}_{x+h}^{+}}\right)\mathbf{1}_{\{\hat{\tau}_{x+h}^{+}<\tau^{\gamma_{_{\epsilon}}}_{\xi}\}}\right)
\nonumber\\
\hspace{-0.3cm}&\leq&\hspace{-0.3cm}\mathbb{E}_{x}\left(\int_{0}^{\tau^{\gamma_{_{\epsilon}}}_{\xi}\wedge\hat{\tau}_{x+h}^{+}}\mathrm{e}^{-qt}\gamma_{_{\epsilon}}\left(\bar{X}(t)\right)\mathrm{d}\bar{X}(t)\right)
+\mathbb{E}_{x}\left(\mathrm{e}^{-q\hat{\tau}_{x+h}^{+}}\mathbf{1}_{\{\hat{\tau}_{x+h}^{+}<\tau^{\gamma_{_{\epsilon}}}_{\xi}\}}\right)f(x+h)
\nonumber\\
\hspace{-0.3cm}&\leq&\hspace{-0.3cm}\sup\limits_{\gamma\in\Gamma}\mathbb{E}_{x}\left(\int_{0}^{\tau^{\gamma}_{\xi}\wedge\tau_{x+h}^{+}}\mathrm{e}^{-qt}\gamma\left(\bar{X}(t)\right)\mathrm{d}\bar{X}(t)
+\mathrm{e}^{-q\tau_{x+h}^{+}}\mathbf{1}_{\{\tau_{x+h}^{+}<\tau^{\gamma}_{\xi}\}}f(x+h)\right).
\end{eqnarray}

On the other hand, given any $\gamma\in \Gamma$, define a new strategy $\tilde{\gamma}\in \Gamma$ as follows: during the time interval $[0,\tau_{x+\epsilon}^{+}]$ the strategy $\gamma$ is adopted, which is then switched to an $\epsilon$-optimal strategy associated with initial reserve $x+h$. Because any given strategy must be suboptimal, we get
\begin{eqnarray}\label{2.17}
f(x)\hspace{-0.3cm}&\geq&\hspace{-0.3cm}\mathbb{E}_{x}\int_{0}^{\tau^{\tilde{\gamma}}_{\xi}}\mathrm{e}^{-qt}\tilde{\gamma}\left(\bar{X}(t)\right)\mathrm{d}\bar{X}(t)\nonumber\\
\hspace{-0.3cm}&=&\hspace{-0.3cm}\mathbb{E}_{x}\left(\int_{0}^{\tau^{\tilde{\gamma}}_{\xi}\wedge\tau_{x+h}^{+}}\mathrm{e}^{-qt}\tilde{\gamma}\left(\bar{X}(t)\right)\mathrm{d}\bar{X}(t)
+\int_{\tau^{\tilde{\gamma}}_{\xi}\wedge\tau_{x+h}^{+}}^{\tau^{\tilde{\gamma}}_{\xi}}\mathrm{e}^{-qt}\tilde{\gamma}\left(\bar{X}(t)\right)\mathrm{d}\bar{X}(t)\right)
\nonumber\\
\hspace{-0.3cm}&=&\hspace{-0.3cm}\mathbb{E}_{x}\left(\int_{0}^{\tau^{\gamma}_{\xi}\wedge\tau_{x+h}^{+}}\mathrm{e}^{-qt}\gamma\left(\bar{X}(t)\right)\mathrm{d}\bar{X}(t)\right)
+\mathbb{E}_{x}\left(\int_{\tau_{x+h}^{+}}^{\tau^{\tilde{\gamma}}_{\xi}}\mathrm{e}^{-qt}\tilde{\gamma}\left(\bar{X}(t)\right)\mathrm{d}\bar{X}(t)\mathbf{1}_{
\{\tau_{x+h}^{+}
<\tau^{\tilde{\gamma}}_{\xi}\}}\right)
\nonumber\\
\hspace{-0.3cm}&=&\hspace{-0.3cm}\mathbb{E}_{x}\left(\int_{0}^{\tau^{\gamma}_{\xi}\wedge\tau_{x+h}^{+}}\mathrm{e}^{-qt}\gamma\left(\bar{X}(t)\right)\mathrm{d}\bar{X}(t)\right)
+\mathbb{E}_{x}\left(\mathbb{E}_{x}\left(\int_{\tau_{x+h}^{+}}^{\tau^{\tilde{\gamma}}_{\xi}}\mathrm{e}^{-qt}\tilde{\gamma}\left(\bar{X}(t)\right)\mathrm{d}\bar{X}(t)
\big|\mathcal{F}_{\tau_{x+h}^{+}}\right)\mathbf{1}_{\{\tau_{x+h}^{+}<\tau^{\gamma}_{\xi}\}}\right)
\nonumber\\
\hspace{-0.3cm}&\geq&\hspace{-0.3cm}\mathbb{E}_{x}\left(\int_{0}^{\tau^{\gamma}_{\xi}\wedge\tau_{x+h}^{+}}\mathrm{e}^{-qt}\gamma\left(\bar{X}(t)\right)\mathrm{d}\bar{X}(t)\right)
+\mathbb{E}_{x}\left(\mathrm{e}^{-q\tau_{x+h}^{+}}\mathbf{1}_{\{\tau_{x+h}^{+}<\tau^{\gamma}_{\xi}\}}\right)(f(x+h)-\epsilon)
\nonumber\\
\hspace{-0.3cm}&\geq&\hspace{-0.3cm}\mathbb{E}_{x}\left(\int_{0}^{\tau^{\gamma}_{\xi}\wedge\tau_{x+h}^{+}}\mathrm{e}^{-qt}\gamma\left(\bar{X}(t)\right)\mathrm{d}\bar{X}(t)
+\mathrm{e}^{-q\tau_{x+h}^{+}}\mathbf{1}_{\{\tau_{x+h}^{+}<\tau^{\gamma}_{\xi}\}}f(x+h)\right)-\epsilon,
\end{eqnarray}
where we have used the fact that $\tau^{\tilde{\gamma}}_{\xi}\wedge\tau_{x+h}^{+}=\tau^{\gamma}_{\xi}\wedge\tau_{x+h}^{+}$ in the second step of the above display due to the definition of the strategy  $\tilde{\gamma}$.
Taking supremum over $\gamma\in\Gamma$ in (\ref{2.17}), we should have,
\begin{eqnarray}\label{2.18}
f(x)\geq\sup\limits_{\gamma\in\Gamma}\mathbb{E}_{x}\left(\int_{0}^{\tau^{\gamma}_{\xi}\wedge\tau_{x+h}^{+}}\mathrm{e}^{-qt}\gamma\left(\bar{X}(t)\right)\mathrm{d}\bar{X}(t)
+\mathrm{e}^{-q\tau_{x+h}^{+}}\mathbf{1}_{\{\tau_{x+h}^{+}<\tau^{\gamma}_{\xi}\}}f(x+h)\right)-\epsilon.
\end{eqnarray}
Putting together (\ref{2.16}) and (\ref{2.18}), by the arbitrariness of $\epsilon>0$ we obtain the following dynamic programming principle,
\begin{eqnarray}
f(x)=\sup\limits_{\gamma\in\Gamma}\mathbb{E}_{x}\left(\int_{0}^{\tau^{\gamma}_{\xi}\wedge\tau_{x+h}^{+}}\mathrm{e}^{-qt}\gamma\left(\bar{X}(t)\right)\mathrm{d}\bar{X}(t)
+\mathrm{e}^{-q\tau_{x+h}^{+}}\mathbf{1}_{\{\tau_{x+h}^{+}<\tau^{\gamma}_{\xi}\}}f(x+h)\right).
\end{eqnarray}
Fix an arbitrary constant $\gamma_{0}\in[\gamma_{1},\gamma_{2}]$, and choose $\gamma\in\Gamma$ such that tax is paid at the fixed rate $\gamma_{0}$ during the time interval $[0,\tau_{x+h}^{+}]$.
In this case,
\begin{eqnarray}
f(x)\hspace{-0.3cm}&\geq&\hspace{-0.3cm}\mathbb{E}_{x}\left(\int_{0}^{\tau^{\gamma}_{\xi}\wedge\tau_{x+h}^{+}}\mathrm{e}^{-qt}\gamma_{0} \mathrm{d}\bar{X}(t)
+\mathrm{e}^{-q\tau_{x+h}^{+}}\mathbf{1}_{\{\tau_{x+h}^{+}<\tau^{\gamma}_{\xi}\}}f(x+h)\right)
\nonumber\\
\hspace{-0.3cm}&\geq&\hspace{-0.3cm}\mathbb{E}_{x}\left(\int_{0}^{\tau_{x+h}^{+}}\mathrm{e}^{-qt}\gamma_{0} \mathrm{d}\bar{X}(t)\mathbf{1}_{\{\tau_{x+h}^{+}<\tau^{\gamma}_{\xi}\}}
+\mathrm{e}^{-q\tau_{x+h}^{+}}\mathbf{1}_{\{\tau_{x+h}^{+}<\tau^{\gamma}_{\xi}\}}f(x+h)\right)
\nonumber\\\hspace{-0.3cm}&=&\hspace{-0.3cm}\mathbb{E}_{x}\left(\left(\gamma_{0} \mathrm{e}^{-q\tau_{x+h}^{+}}\left(\bar{X}\left(\tau_{x+h}^{+}\right)-x\right)+q\gamma_{0} \int_{0}^{\tau_{x+h}^{+}}\mathrm{e}^{-qt}\left(\bar{X}(t)-x\right)\mathrm{d}t\right)
\mathbf{1}_{\{\tau_{x+h}^{+}<\tau^{\gamma}_{\xi}\}}\right)
\nonumber\\\hspace{-0.3cm}&&\hspace{-0.3cm}+\mathbb{E}_{x}\left(\mathrm{e}^{-q\tau_{x+h}^{+}}\mathbf{1}_{\{\tau_{x+h}^{+}<\tau^{\gamma}_{\xi}\}}\right)f(x+h)
\nonumber\\\hspace{-0.3cm}&\geq&\hspace{-0.3cm}\mathbb{E}_{x}\left(\gamma_{0} \mathrm{e}^{-q\tau_{x+h}^{+}}\left(\bar{X}\left(\tau_{x+h}^{+}\right)-x\right)
\mathbf{1}_{\{\tau_{x+h}^{+}<\tau^{\gamma}_{\xi}\}}\right)
+\mathbb{E}_{x}\left(\mathrm{e}^{-q\tau_{x+h}^{+}}\mathbf{1}_{\{\tau_{x+h}^{+}<\tau^{\gamma}_{\xi}\}}\right)f(x+h).\nonumber
\end{eqnarray}
It can be verified from (\ref{v12.3}) that
\begin{eqnarray}
x+h=U^{\gamma}\left(\tau_{x+h}^{+}\right)\hspace{-0.3cm}&=&\hspace{-0.3cm}x+\int_{0}^{\tau_{x+h}^{+}}\left(1-\gamma\left(\bar{X}(t)\right)\right)\mathrm{d}\bar{X}(t)
\nonumber\\
\hspace{-0.3cm}&=&\hspace{-0.3cm}
x+\int_{0}^{\tau_{x+h}^{+}}\left(1-\gamma_{0}\right)\mathrm{d}\bar{X}(t)
=x+(1-\gamma_{0})\left(\bar{X}(\tau_{x+h}^{+})-x\right),
\end{eqnarray}
which implies that $\bar{X}\left(\tau_{x+h}^{+}\right)=x+\frac{h}{1-\gamma_{0}}$. Using this result we can further push forward our inequality as follows,
\begin{eqnarray}
f(x)\hspace{-0.3cm}&\geq&\hspace{-0.3cm}\frac{\gamma_{0}h}{1-\gamma_{0}}\mathbb{E}_{x}\left(\mathrm{e}^{-q\tau_{x+h}^{+}}
\mathbf{1}_{\{\tau_{x+h}^{+}<\tau^{\gamma}_{\xi}\}}\right)
+\mathbb{E}_{x}\left(\mathrm{e}^{-q\tau_{x+h}^{+}}\mathbf{1}_{\{\tau_{x+h}^{+}<\tau^{\gamma}_{\xi}\}}\right)f(x+h)
\nonumber\\
\hspace{-0.3cm}&=&\hspace{-0.3cm}
\exp\{-\int_{x}^{x+h}\frac{W^{(q)'}(y-\xi\left(y\right))}
{W^{(q)}(y-\xi\left(y\right))}\frac{1}{1-\gamma_{0}}\mathrm{d}y\}\left(f(x+h)+\frac{\gamma_{0}h}{1-\gamma_{0}}\right)
\nonumber\\
\hspace{-0.3cm}&=&\hspace{-0.3cm}\left(1-\frac{1}{1-\gamma_{0}}\frac{W^{(q)'}(x-\xi(x))}{W^{(q)}(x-\xi(x))}h+o(h)\right)\left(\frac{\gamma_{0}h}{1-\gamma_{0}}
+f(x)+f'(x)h+o(h)\right),
\end{eqnarray}
where we used (\ref{exit problem eqaution}) in Proposition \ref{exit problem} with $\gamma\equiv \gamma_{0}$.

Subtracting both sides of (3.7) with $f(x)$ and collecting the terms of order $h$ yield
\begin{eqnarray}
0\hspace{-0.3cm}&\geq&\hspace{-0.3cm}\frac{\gamma_{0}}{1-\gamma_{0}}-\frac{1}{1-\gamma_{0}} \frac{W^{(q)'}(x-\xi(x))}{W^{(q)}(x-\xi(x))}f(x)+f'(x).
\end{eqnarray}
The arbitrariness of $\gamma_{0}$ leads to
\begin{eqnarray}\label{2.23}
0\hspace{-0.3cm}&\geq&\hspace{-0.3cm}\sup\limits_{\gamma\in[\gamma_{1},\gamma_{2}]}[\frac{\gamma}{1-\gamma}-\frac{1}{1-\gamma} \frac{W^{(q)'}(x-\xi(x))}{W^{(q)}(x-\xi(x))}f(x)+f'(x)].
\end{eqnarray}
For $h^{2}>0$, by the definition of supremum there should exist a strategy $\tilde{\gamma}$ such that
\begin{eqnarray}\label{2.24}
f(x)
\hspace{-0.3cm}&\leq&\hspace{-0.3cm}\mathbb{E}_{x}\left(\int_{0}^{\tau^{\tilde{\gamma}}_{\xi}\wedge\tilde{\tau}_{x+h}^{+}}\mathrm{e}^{-qt}\tilde{\gamma}\left(\bar{X}(t)\right)\mathrm{d}\bar{X}(t)
+\mathrm{e}^{-q\tilde{\tau}_{x+h}^{+}}\mathbf{1}_{\{\tilde{\tau}_{x+h}^{+}<\tau^{\tilde{\gamma}}_{\xi}\}}f(x+h)\right)+h^{2}
\nonumber\\
\hspace{-0.3cm}&\leq&\hspace{-0.3cm}\mathbb{E}_{x}\left(\int_{0}^{\tilde{\tau}_{x+h}^{+}}\mathrm{e}^{-qt}\tilde{\gamma}\left(\bar{X}(t)\right)\mathrm{d}\bar{X}(t)
\mathbf{1}_{\{\tilde{\tau}_{x+h}^{+}<\tau^{\tilde{\gamma}}_{\xi}\}}+
\int_{0}^{\tau^{\tilde{\gamma}}_{\xi}}\mathrm{e}^{-qt}\tilde{\gamma}\left(\bar{X}(t)\right)\mathrm{d}\bar{X}(t)\mathbf{1}_{\{\tilde{\tau}_{x+h}^{+}>\tau^{\tilde{\gamma}}_{\xi}\}}\right)\nonumber\\
\hspace{-0.3cm}&&\hspace{-0.3cm}
+\mathbb{E}_{x}\left(\mathrm{e}^{-q\tilde{\tau}_{x+h}^{+}}\mathbf{1}_{\{\tilde{\tau}_{x+h}^{+}<\tau^{\tilde{\gamma}}_{\xi}\}}\right)f(x+h)
+h^{2}.
\end{eqnarray}
Here, $\tilde{\tau}_{x+h}^{+}$ is defined by (2.1) with $\gamma$ replaced by $\tilde{\gamma}$.

The first term on the right hand side of (\ref{2.24}) can be rewritten as
\begin{eqnarray}\label{2.25}
\hspace{-1.3cm}\mathbb{E}_{x}\left(\int_{0}^{\tilde{\tau}_{x+h}^{+}}\mathrm{e}^{-qt}\tilde{\gamma}\left(\bar{X}(t)\right)\mathrm{d}\bar{X}(t)
\mathbf{1}_{\{\tilde{\tau}_{x+h}^{+}<\tau^{\tilde{\gamma}}_{\xi}\}}\right)
\hspace{-0.3cm}&=&\hspace{-0.3cm}\mathbb{E}_{x}\left(\mathrm{e}^{-q\tilde{\tau}_{x+h}^{+}}\int_{0}^{\tilde{\tau}_{x+h}^{+}}\tilde{\gamma}\left(\bar{X}(t)\right)\mathrm{d}\bar{X}(t)
\mathbf{1}_{\{\tilde{\tau}_{x+h}^{+}<\tau^{\tilde{\gamma}}_{\xi}\}}\right)\nonumber\\\hspace{-1.3cm}&&\hspace{-1.3cm}
+\mathbb{E}_{x}\left(q\int_{0}^{\tilde{\tau}_{x+h}^{+}}\mathrm{e}^{-qt}\left(\int_{0}^{t}\tilde{\gamma}\left(\bar{X}(r)\right)\mathrm{d}\bar{X}(r)\right)\mathrm{d}t
\mathbf{1}_{\{\tilde{\tau}_{x+h}^{+}<\tau^{\tilde{\gamma}}_{\xi}\}}\right).
\end{eqnarray}
The cumulative (non-discounted) tax until the
stoping time $\tilde{\tau}_{x+h}^{+}\,(h>0)$ can be re-expressed as
\begin{eqnarray}\label{2.26}
\int_{0}^{\tilde{\tau}_{x+h}^{+}}\tilde{\gamma}\left(\bar{X}(s)\right)\mathrm{d}\bar{X}(s)
\hspace{-0.3cm}&=&\hspace{-0.3cm}\bar{X}\left(\tilde{\tau}_{x+h}^{+}\right)-\left(x+\int_{0}^{\tilde{\tau}_{x+h}^{+}}\left(1-\tilde{\gamma}\left(\bar{X}(s)\right)\right)\mathrm{d}\bar{X}(s)\right)\nonumber\\
\hspace{-0.3cm}
&=&\hspace{-0.3cm}\left(\bar{\tilde{\gamma}}\right)^{-1}(x+h)-U^{\tilde{\gamma}}\left(\tilde{\tau}_{x+h}^{+}\right)=\left(\bar{\tilde{\gamma}}\right)^{-1}(x+h)-x-h
\nonumber\\
\hspace{-0.3cm}
&=&\hspace{-0.3cm}\frac{\tilde{\gamma}(x)}{1-\tilde{\gamma}(x)}h+o(h),
\end{eqnarray}
where we have used the fact that,
\begin{eqnarray}
\lim\limits_{h\downarrow0}\frac{\left(\bar{\tilde{\gamma}}\right)^{-1}(x+h)-x-h}{h}=\left(\left(\bar{\tilde{\gamma}}\right)^{-1}(x)\right)'-1
=\frac{\tilde{\gamma}\left(\left(\bar{\tilde{\gamma}}\right)^{-1}(x)\right)}{1-\tilde{\gamma}\left(\left(\bar{\tilde{\gamma}}\right)^{-1}(x)\right)}
=\frac{\tilde{\gamma}(x)}{1-\tilde{\gamma}(x)},\nonumber
\end{eqnarray}
since $\bar{\tilde{\gamma}}(z)$ is a strictly increasing and continuous function of $z$. By (\ref{2.26}), the first quantity on the right hand side of (\ref{2.25}) can be re-written as,
\begin{eqnarray}\label{2.27}
&&\mathbb{E}_{x}\left(\mathrm{e}^{-q\tilde{\tau}_{x+h}^{+}}\int_{0}^{\tilde{\tau}_{x+h}^{+}}\tilde{\gamma}\left(\bar{X}(t)\right)\mathrm{d}\bar{X}(t)
\mathbf{1}_{\{\tilde{\tau}_{x+h}^{+}<\tau^{\tilde{\gamma}}_{\xi}\}}\right)\nonumber\\
\hspace{-0.3cm}&=&\hspace{-0.3cm}\left(\frac{\tilde{\gamma}(x)}{1-\tilde{\gamma}(x)}h
+o(h)\right)
\mbox{exp}\{-\int_{x}^{x+h}\frac{W^{(q)'}(y-\xi(y))}{W^{(q)}(y-\xi(y))(1-\tilde{\gamma}((\bar{\tilde{\gamma}})^{-1}(y)))}\mathrm{d}y\}
\nonumber\\
\hspace{-0.3cm}&=&\hspace{-0.3cm}\left(\frac{\tilde{\gamma}(x)}{1-\tilde{\gamma}(x)}h
+o(h)\right)\left(1-\frac{W^{(q)'}(x-\xi(x))}{W^{(q)}(x-\xi(x))(1-\tilde{\gamma}(x))}h
+o(h)\right)\nonumber\\
\hspace{-0.3cm}&=&\hspace{-0.3cm}\frac{\tilde{\gamma}(x)}{1-\tilde{\gamma}(x)}h
+o(h)
\end{eqnarray}
and
\begin{eqnarray}\label{2.28}
&&\mathbb{E}_{x}\left(q\int_{0}^{\tilde{\tau}_{x+h}^{+}}\mathrm{e}^{-qt}\left(\int_{0}^{t}\tilde{\gamma}\left(\bar{X}(r)\right)\mathrm{d}\bar{X}(r)\right)\mathrm{d}t
\mathbf{1}_{\{\tilde{\tau}_{x+h}^{+}<\tau^{\tilde{\gamma}}_{\xi}\}}\right)\nonumber\\
\hspace{-0.3cm}&\leq&\hspace{-0.3cm}\left(\left(\bar{\tilde{\gamma}}\right)^{-1}(x+h)-x-h\right)\hspace{0.05cm}
\mathbb{E}_{x}\left((1-\mathrm{e}^{-q\tilde{\tau}_{x+h}^{+}})
\mathbf{1}_{\{\tilde{\tau}_{x+h}^{+}<\tau^{\tilde{\gamma}}_{\xi}\}}\right)
\nonumber\\
\hspace{-0.3cm}&=&\hspace{-0.3cm}\left(\frac{\tilde{\gamma}(x)}{1-\tilde{\gamma}(x)}h
+o(h)\right)\hspace{0.05cm}
\left(\mbox{exp}\{-\int_{x}^{x+h}\frac{W^{(0)'}(y-\xi(y))}{W^{(0)}(y-\xi(y))(1-\tilde{\gamma}((\bar{\tilde{\gamma}})^{-1}(y)))}\mathrm{d}y\}\right.
\nonumber\\
\hspace{-0.3cm}&&\hspace{-0.3cm}-\left.\mbox{exp}\{-\int_{x}^{x+h}\frac{W^{(q)'}(y-\xi(y))}{W^{(q)}(y-\xi(y))(1-\tilde{\gamma}((\bar{\tilde{\gamma}})^{-1}(y)))}\mathrm{d}y\}\right)
\nonumber\\
\hspace{-0.3cm}&=&\hspace{-0.3cm}o(h).
\end{eqnarray}
The second term on the right hand side of (\ref{2.24}) can be calculated as follows,
\begin{eqnarray}\label{2.29}
&&\mathbb{E}_{x}\left(\int_{0}^{\tau^{\tilde{\gamma}}_{\xi}}\mathrm{e}^{-qt}\tilde{\gamma}\left(\bar{X}(t)\right)\mathrm{d}\bar{X}(t)\mathbf{1}_{\{\tilde{\tau}_{x+h}^{+}>\tau^{\tilde{\gamma}}_{\xi}\}}\right)
\nonumber\\
\hspace{-0.3cm}&\leq&\hspace{-0.3cm}
\mathbb{E}_{x}\left(\int_{0}^{\tilde{\tau}_{x+h}^{+}}\mathrm{e}^{-qt}\tilde{\gamma}\left(\bar{X}(t)\right)\mathrm{d}\bar{X}(t)\mathbf{1}_{\{\tilde{\tau}_{x+h}^{+}
>\tau^{\tilde{\gamma}}_{\xi}\}}\right)
\nonumber\\
\hspace{-0.3cm}&=&\hspace{-0.3cm}
\mathbb{E}_{x}\left(\mathrm{e}^{-q\tilde{\tau}_{x+h}^{+}}\int_{0}^{\tilde{\tau}_{x+h}^{+}}\tilde{\gamma}\left(\bar{X}(t)\right)\mathrm{d}\bar{X}(t)
\mathbf{1}_{\{\tilde{\tau}_{x+h}^{+}>\tau^{\tilde{\gamma}}_{\xi}\}}+q\int_{0}^{\tilde{\tau}_{x+h}^{+}}
\mathrm{e}^{-qt}\left(\int_{0}^{t}\tilde{\gamma}\left(\bar{X}(r)\right)\mathrm{d}\bar{X}(t)\right)\mathrm{d}t
\mathbf{1}_{\{\tilde{\tau}_{x+h}^{+}>\tau^{\tilde{\gamma}}_{\xi}\}}\right)
\nonumber\\
\hspace{-0.3cm}&\leq&\hspace{-0.3cm}
\left(\left(\bar{\tilde{\gamma}}\right)^{-1}(x+h)-x-h\right)
\mathbb{E}_{x}\left(\mathrm{e}^{-q\tilde{\tau}_{x+h}^{+}}
\mathbf{1}_{\{\tilde{\tau}_{x+h}^{+}>\tau^{\tilde{\gamma}}_{\xi}\}}+q\int_{0}^{\tilde{\tau}_{x+h}^{+}}\mathrm{e}^{-qt}\mathrm{d}t
\mathbf{1}_{\{\tilde{\tau}_{x+h}^{+}>\tau^{\tilde{\gamma}}_{\xi}\}}\right)
\nonumber\\
\hspace{-0.3cm}&=&\hspace{-0.3cm}
\left(\left(\bar{\tilde{\gamma}}\right)^{-1}(x+h)-x-h\right)
\left(1-\mbox{exp}\{-\int_{x}^{x+h}\frac{W^{(0)'}(y-\xi(y))}{W^{(0)}(y-\xi(y))(1-\tilde{\gamma}((\bar{\tilde{\gamma}})^{-1}(y)))}\mathrm{d}y\}\right)
\nonumber\\\hspace{-0.3cm}&=&\hspace{-0.3cm}o(h).
\end{eqnarray}
Hence, (\ref{2.24}), (\ref{2.25}), (\ref{2.27}), (\ref{2.28}) together with (\ref{2.29}) give rise to
\begin{eqnarray}\label{2.30}
f(x)
\hspace{-0.3cm}&\leq&\hspace{-0.3cm}\frac{\tilde{\gamma}(x)}{1-\tilde{\gamma}(x)}h
+f(x)+f'(x)h-\frac{W^{(q)'}(x-\xi(x))}{W^{(q)}(x-\xi(x))(1-\tilde{\gamma}(x))}f(x)h
+o(h).
\end{eqnarray}
Subtracting both sides of (\ref{2.30}) by $f(x)$ and then collecting the terms of order $h$ we have
\begin{eqnarray}\label{2.31}
0
\hspace{-0.3cm}&\leq&\hspace{-0.3cm}\frac{\tilde{\gamma}(x)}{1-\tilde{\gamma}(x)}
+f'(x)-\frac{W^{(q)'}(x-\xi(x))}{W^{(q)}(x-\xi(x))(1-\tilde{\gamma}(x))}f(x)
\nonumber\\
\hspace{-0.3cm}&\leq&\hspace{-0.3cm}\sup\limits_{\gamma\in\left[\gamma_{1},\gamma_{2}\right]}\left(\frac{\gamma}{1-\gamma}-\frac{1}{1-\gamma} \frac{W^{(q)'}(x-\xi(x))}{W^{(q)}(x-\xi(x))}f(x)+f'(x)\right).
\end{eqnarray}
Finally, combing (\ref{2.23}) and (\ref{2.31}) leads to (\ref{HJB equation}).
\end{proof}

\subsection{Proof of a technical lemma}

\begin{lemma}\label{lem.1}
For $z\in[x,\infty)$, put
$$\xi_{1}(y):=
\int_{x}^{z}\gamma(w)\mathrm{d}w
+\xi(y-\int_{x}^{z}\gamma(w)\mathrm{d}w),\quad y\in[z,\infty),$$
then we have
\begin{eqnarray}\label{r11.1}
\hspace{-0.3cm}&&\hspace{-0.3cm}
\mathbb{E}_{x}\left(\left.\mathrm{e}^{-q \left(\tau_{\overline{\gamma}_{x}(z)+a}^{+}
-\tau_{\overline{\gamma}_{x}(z)}^{+}\right)}
\mathbf{1}_{\{\tau_{\overline{\gamma}_{x}(z)+a}^{+}<\tau^{\gamma}_{\xi}\}}\right|\mathcal{F}_{\tau_{\overline{\gamma}_{x}(z)}^{+}}\right)
\nonumber\\
\hspace{-0.3cm}&=&\hspace{-0.3cm}
\mathbf{1}_{\{\tau_{\overline{\gamma}_{x}(z)}^{+}<\tau^{\gamma}_{\xi}\}}
\exp\left(-\int_{z}^{z+a}\frac{W^{(q)'}(y-\xi_{1}\left(y\right))}
{W^{(q)}(y-\xi_{1}\left(y\right))}\frac{1}{1-\gamma\left(\overline{\gamma}_{z}^{-1}(y)\right)}
\mathrm{d}y\right),\quad z\in[x,\infty),\,a\in[0,\infty).
\end{eqnarray}
\end{lemma}

\begin{proof}
It is found that
\begin{eqnarray}
U_{\gamma}(t)
\hspace{-0.3cm}&=&\hspace{-0.3cm}
X(t)-\bar{X}(t)
+z+\int_{z}^{\bar{X}(t)}\left(1-\gamma(w)\right)\mathrm{d}w
-\int_{x}^{z}\gamma(w)\mathrm{d}w
\nonumber\\
\hspace{-0.3cm}&=&\hspace{-0.3cm}
X(t)-\bar{X}(t)
+\overline{\gamma}_{z}\left(\bar{X}(t)\right)
-\int_{x}^{z}\gamma(w)\mathrm{d}w,\quad t\in[\tau_{\overline{\gamma}_{x}(z)}^{+},\infty),\,z\in[x,\infty),\nonumber
\end{eqnarray}
with
$\overline{\gamma}_{x}(z)=z
-\int_{x}^{z}\gamma(w)\mathrm{d}w,\,\,z\in[x,\infty)$.
Hence, we have
\begin{eqnarray}
\label{.1}
\tau_{\overline{\gamma}_{x}(z)+a}^{+}
\hspace{-0.3cm}&=&\hspace{-0.3cm}
\inf\{t\geq\tau_{\overline{\gamma}_{x}(z)}^{+}; U_{\gamma}(t)
>
\overline{\gamma}_{x}(z)+a\}
\nonumber\\
\hspace{-0.3cm}&=&\hspace{-0.3cm}
\inf\{t\geq\tau_{\overline{\gamma}_{x}(z)}^{+}; X(t)-\bar{X}(t)
+\overline{\gamma}_{z}\left(\bar{X}(t)\right)
>z+a\}
,
\end{eqnarray}
and, on $\{\tau_{\overline{\gamma}_{x}(z)}^{+}<\tau^{\gamma}_{\xi}\}$
\begin{eqnarray}
\label{.2}
\tau^{\gamma}_{\xi}
\hspace{-0.3cm}&=&\hspace{-0.3cm}
\inf\{t\geq\tau_{\overline{\gamma}_{x}(z)}^{+}; X(t)-\bar{X}(t)
+\overline{\gamma}_{z}\left(\bar{X}(t)\right)
<
\int_{x}^{z}\gamma(w)\mathrm{d}w
+\xi(\overline{\gamma}_{z}(\bar{X}(t))-\int_{x}^{z}\gamma(w)\mathrm{d}w)\}.
\end{eqnarray}
Recalling that
$$\tau_{\overline{\gamma}_{x}(z)}^{+}=\inf\{t\geq0; X(t)>z\}$$ and
$$U_{\gamma}(t)=X(t)-\bar{X}(t)
+\overline{\gamma}_{x}\left(\bar{X}(t)\right),$$
one can find from \eqref{.1} and \eqref{.2}
\begin{eqnarray}
\tau_{\overline{\gamma}_{x}(z)+a}^{+}
-\tau_{\overline{\gamma}_{x}(z)}^{+}=\tau_{z+a}^{+}\small{\circ}\theta_{\tau_{\overline{\gamma}_{x}(z)}^{+}},
\quad \mathbb{P}_{x}-a.s.,\nonumber
\end{eqnarray}
and
\begin{eqnarray}
\tau^{\gamma}_{\xi}-\tau_{\overline{\gamma}_{x}(z)}^{+}
=\tau_{\xi_{1}}^{\gamma}\small{\circ}\theta_{\tau_{\overline{\gamma}_{x}(z)}^{+}},\quad \mathbb{P}_{x}-a.s. \mbox{ on }\{\tau_{\overline{\gamma}_{x}(z)}^{+}<\tau^{\gamma}_{\xi}\},\nonumber
\end{eqnarray}
which combined with the strong Markov property yield
\begin{eqnarray}
\hspace{-0.3cm}&&\hspace{-0.3cm}
\mathbb{E}_{x}\left(\left.\mathrm{e}^{-q \left(\tau_{\overline{\gamma}_{x}(z)+a}^{+}
-\tau_{\overline{\gamma}_{x}(z)}^{+}\right)}
\mathbf{1}_{\{\tau_{\overline{\gamma}_{x}(z)+a}^{+}<\tau^{\gamma}_{\xi}\}}\right|\mathcal{F}_{\tau_{\overline{\gamma}_{x}(z)}^{+}}\right)
\nonumber\\
\hspace{-0.3cm}&=&\hspace{-0.3cm}
\mathbf{1}_{\{\tau_{\overline{\gamma}_{x}(z)}^{+}<\tau^{\gamma}_{\xi}\}}
\mathbb{E}_{z}\left(\mathrm{e}^{-q \tau_{z+a}^{+}}
\mathbf{1}_{\{\tau_{z+a}^{+}<\tau_{\xi_{1}}^{\gamma}\}}\right),
\nonumber
\end{eqnarray}
which together with \eqref{exit problem eqaution} gives \eqref{r11.1}.
\end{proof}

\subsection{Proof of Proposition \ref{verification pro}}

\begin{proof}
Since $\gamma(z)\in[\gamma_{1},\gamma_{2}]$ for all $z\in[x,\infty)$ and $\gamma\in\Gamma$, it can be checked from (\ref{HJB equation}) that for any $\gamma\in\Gamma$ and $z\geq x$,
\begin{eqnarray}
\gamma(z)\leq -\left(\left(1-\gamma(z)\right)f'\left(\overline{\gamma}_{x}(z)\right)-\frac{W^{(q)'}\left(\overline{\gamma}_{x}(z)-\xi\left(\overline{\gamma}_{x}(z)\right)\right)}
{W^{(q)}\left(\overline{\gamma}_{x}(z)-\xi\left(\overline{\gamma}_{x}(z)\right)\right)}f\left(\overline{\gamma}_{x}(z)\right)\right),\nonumber
\end{eqnarray}
with the equality holds true when $\gamma\equiv\gamma^{*}$. Replacing $z$ with $\bar{X}(t)$ in the above inequality we should have, for all $s\geq0$,
\begin{eqnarray}
\label{3.18}
\hspace{-0.2cm}\gamma\hspace{-0.05cm}\left(\bar{X}(s)\right)\hspace{-0.05cm}\leq \hspace{-0.05cm} -\hspace{-0.05cm}\left(\hspace{-0.05cm}\left(1-\gamma\hspace{-0.05cm}\left(\bar{X}(s)\right)\right)f'\hspace{-0.1cm}\left(\overline{\gamma}_{x}\hspace{-0.05cm}\left(\bar{X}(s)\right)\right)
\hspace{-0.05cm}-\hspace{-0.05cm}\frac{W^{(q)'}\left(\overline{\gamma}_{x}\hspace{-0.05cm}\left(\bar{X}(s)\right)-\xi\left(\overline{\gamma}_{x}\hspace{-0.05cm}\left(\bar{X}(s)\right)\right)\right)}
{W^{(q)}\left(\overline{\gamma}_{x}\hspace{-0.05cm}\left(\bar{X}(s)\right)-\xi\left(\overline{\gamma}_{x}\hspace{-0.05cm}\left(\bar{X}(s)\right)\right)\right)}f\hspace{-0.1cm}
\left(\overline{\gamma}_{x}\hspace{-0.05cm}\left(\bar{X}(s)\right)\right)\hspace{-0.05cm}\right),
\end{eqnarray}
with the equality holds true when $\gamma\equiv\gamma^{*}$.

Define $\{\eta(s);s\geq0\}$ as follows,
\begin{eqnarray}\label{3.19}
\eta(s)\hspace{-0.3cm}&=&\hspace{-0.3cm}\lim\limits_{h\downarrow0}\frac{1}{h}\,\mathbb{E}_{x}\left(\left(
\mathrm{e}^{-q\tau_{\overline{\gamma}_{x}\left(\bar{X}(s)\right)+\left(1-\gamma\left(\bar{X}(s)\right)\right)h}^{+}}\mathbf{1}_{\{
\tau_{\overline{\gamma}_{x}\left(\bar{X}(s)\right)+\left(1-\gamma\left(\bar{X}(s)\right)\right)h}^{+}<\tau^{\gamma}_{\xi}\}}
f\left(\overline{\gamma}_{x}\left(\bar{X}(s)\right)+\left(1-\gamma\left(\bar{X}(s)\right)\right)h\right)\right.\right.
\nonumber\\
&&\hspace{1.8cm}\left.\left.\left.-\mathrm{e}^{-q\tau_{_{\overline{\gamma}_{x}\left(\bar{X}(s)\right)}}^{+}}\mathbf{1}_{\{\tau_{_{\overline{\gamma}_{x}\left(\bar{X}(s)\right)}}^{+}
<\tau^{\gamma}_{\xi}\}}
f\left(\overline{\gamma}_{x}\left(\bar{X}(s)\right)\right)
\right)\right|\mathcal{F}_{\tau^{+}_{\overline{\gamma}_{x}\left(\bar{X}(s)\right)}}\right).
\end{eqnarray}
Then, by mimicking the arguments in the proof of the martingale property of the process (4.3) in Gerber and Shiu (2006) or the process (3.5) in Wang and Hu (2012), we can verify that the following compensated process,
\begin{eqnarray}
Z(t)=\mathrm{e}^{-q\tau_{_{\overline{\gamma}_{x}\left(\bar{X}(t)\right)}}^{+}}
\mathbf{1}_{\{\tau_{_{\overline{\gamma}_{x}\left(\bar{X}(t)\right)}}^{+}
<\tau^{\gamma}_{\xi}\}}
f\left(\overline{\gamma}_{x}\left(\bar{X}(t)\right)\right)
-\int_{0}^{\tau_{_{\overline{\gamma}_{x}\left(\bar{X}(t)\right)}}^{+}}\eta(s)\mathrm{d}\bar{X}(s),\qquad t\geq0,
\end{eqnarray}
is a martingale with respect to the filtration $\{\mathcal{F}_{\tau_{\overline{\gamma}_{x}\left(\bar{X}(t)\right)}^{+}};t\geq0\}$.
In addition, the right hand side of (\ref{3.19}) can be rewritten as,
\begin{eqnarray}\label{3.25}
\eta(s)
\hspace{-0.3cm}&=&\hspace{-0.3cm}\lim\limits_{h\downarrow0}\frac{1}{h}\left(\mathbb{E}_{x}\left(
\mathrm{e}^{-q\tau_{_{\overline{\gamma}_{x}\left(\bar{X}(s)\right)+\left(1-\gamma\left(\bar{X}(s)\right)\right)h}}^{+}}\mathbf{1}_{\{
\tau_{_{\overline{\gamma}_{x}\left(\bar{X}(s)\right)+\left(1-\gamma\left(\bar{X}(s)\right)\right)h}}^{+}<\tau^{\gamma}_{\xi}\}}\right.\right.
\nonumber\\
&&
\left.\left.\times
f\left(\overline{\gamma}_{x}\left(\bar{X}(s)\right)+\left(1-\gamma\left(\bar{X}(s)\right)\right)h\right)\right|\mathcal{F}_{\tau_{\overline{\gamma}_{x}\left(\bar{X}(s)\right)}^{+}}\right)
\left.-\mathrm{e}^{-q\tau_{\overline{\gamma}_{x}\left(\bar{X}(s)\right)}^{+}}\mathbf{1}_{\{\tau_{\overline{\gamma}_{x}\left(\bar{X}(s)\right)}^{+}
<\tau^{\gamma}_{\xi}\}}
f\left(\overline{\gamma}_{x}\left(\bar{X}(s)\right)\right)\right)
\nonumber\\
\hspace{-0.3cm}&=&\hspace{-0.3cm}\lim\limits_{h\downarrow0}\frac{1}{h}\left(\mathbb{E}_{x}
\left(\left.\mathrm{e}^{-q\tau_{_{\overline{\gamma}_{x}\left(\bar{X}(s)\right)+\left(1-\gamma\left(\bar{X}(s)\right)\right)h}}^{+}}\mathbf{1}_{\{
\tau_{_{\overline{\gamma}_{x}\left(\bar{X}(s)\right)+\left(1-\gamma\left(\bar{X}(s)\right)\right)h}}^{+}<\tau^{\gamma}_{\xi}\}}
\right|\mathcal{F}_{\tau_{\overline{\gamma}_{x}\left(\bar{X}(s)\right)}^{+}}\right)
\right.
\nonumber\\
&&\times
\left.\left(f\left(\overline{\gamma}_{x}\left(\bar{X}(s)\right)\right)+f'\left(\overline{\gamma}_{x}\left(\bar{X}(s)\right)\right)\left(1-\gamma\left(\bar{X}(s)\right)\right)h+o(h)\right)
\right.
\nonumber\\
&&
\left.
-\mathrm{e}^{-q\tau_{\overline{\gamma}_{x}\left(\bar{X}(s)\right)}^{+}}\mathbf{1}_{\{\tau_{\overline{\gamma}_{x}\left(\bar{X}(s)\right)}^{+}
<\tau^{\gamma}_{\xi}\}}
f\left(\overline{\gamma}_{x}\left(\bar{X}(s)\right)\right)
\right)
\nonumber\\
\hspace{-0.3cm}&=&\hspace{-0.3cm}\lim\limits_{h\downarrow0}\frac{1}{h}\left(\left(1-\frac{(W^{(q)})'\left(\overline{\gamma}_{x}\left(\bar{X}(s)\right)-\xi\left(\overline{\gamma}_{x}\left(\bar{X}(s)\right)\right)\right)}
{W^{(q)}\left(\overline{\gamma}_{x}\left(\bar{X}(s)\right)-\xi\left(\overline{\gamma}_{x}\left(\bar{X}(s)\right)\right)\right)}h+o(h)\right)\right.
\nonumber\\
&&
\left.
\left(f\left(\overline{\gamma}_{x}\left(\bar{X}(s)\right)\right)+f'\left(\overline{\gamma}_{x}\left(\bar{X}(s)\right)\right)\left(1-\gamma\left(\bar{X}(s)\right)\right)h+o(h)\right)\right.
\nonumber\\
&&
\left.-f\left(\overline{\gamma}_{x}\left(\bar{X}(s)\right)\right)
\right) \mathrm{e}^{-q\tau_{\overline{\gamma}_{x}\left(\bar{X}(s)\right)}^{+}}\mathbf{1}_{\{\tau_{\overline{\gamma}_{x}\left(\bar{X}(s)\right)}^{+}
<\tau^{\gamma}_{\xi}\}}\nonumber\\
\hspace{-0.3cm}&=&\hspace{-0.3cm}\left(\left(1-\gamma\left(\bar{X}(s)\right)\right)
f'\left(\overline{\gamma}_{x}\left(\bar{X}(s)\right)\right)-\frac{(W^{(q)})'\left(\overline{\gamma}_{x}\left(\bar{X}(s)\right)-\xi\left(\overline{\gamma}_{x}\left(\bar{X}(s)\right)\right)\right)}
{W^{(q)}\left(\overline{\gamma}_{x}\left(\bar{X}(s)\right)-\xi\left(\overline{\gamma}_{x}
\left(\bar{X}(s)\right)\right)\right)}
f\left(\overline{\gamma}_{x}\left(\bar{X}(s)\right)\right)\right)
\nonumber\\
&&
\times \mathrm{e}^{-q\tau_{\overline{\gamma}_{x}\left(\bar{X}(s)\right)}^{+}}\mathbf{1}_{\{\tau_{\overline{\gamma}_{x}\left(\bar{X}(s)\right)}^{+}
<\tau^{\gamma}_{\xi}\}},
\end{eqnarray}
where in the third equality of (\ref{3.25}) we have used Lemma \ref{lem.1} and the fact
$$z-\xi_{1}(z)=\overline{\gamma}_{x}\left(z\right)-\xi\left(\overline{\gamma}_{x}
\left(z\right)\right).$$
Combining (\ref{3.18}) and (\ref{3.25}) yields,% $\gamma\left(\bar{X}(s)\right)\leq-\eta(s)$, which further implies,
\begin{eqnarray}\label{less}
\gamma\left(\bar{X}(s)\right)\mathrm{e}^{-q\tau_{\overline{\gamma}_{x}\left(\bar{X}(s)\right)}^{+}}\mathbf{1}_{\{\tau_{\overline{\gamma}_{x}\left(\bar{X}(s)\right)}^{+}
<\tau^{\gamma}_{\xi}\}}\leq-\eta(s),
%\mathrm{e}^{-q\tau_{\overline{\gamma}_{x}\left(\bar{X}(s)\right)}^{+}}\mathbf{1}_{\{\tau_{\overline{\gamma}_{x}\left(\bar{X}(s)\right)}^{+}
%<\tau^{\gamma}_{\xi}\}}
\end{eqnarray}
with the equality holds true when $\gamma=\gamma^{*}$. It is also seen from (\ref{less}) that $\{-\eta(s);s\geq0\}$ is a nonnegative valued process.
Noting that $\{Z(t);t\geq0\}$ is a martingale with respect to $\{\mathcal{F}_{\tau_{\overline{\gamma}_{x}\left(\bar{X}(t)\right)}^{+}};t\geq0\}$, we have
\begin{eqnarray}\label{3.22}
\mathbb{E}_{x}\left[\mathrm{e}^{-q\tau_{_{\overline{\gamma}_{x}\left(\bar{X}(t)\right)}}^{+}}
\mathbf{1}_{\{\tau_{_{\overline{\gamma}_{x}\left(\bar{X}(t)\right)}}^{+}
<\tau^{\gamma}_{\xi}\}}
f\left(\overline{\gamma}_{x}\left(\bar{X}(t)\right)\right)
-\int_{0}^{\tau_{\overline{\gamma}_{x}\left(\bar{X}(t)\right)}^{+}}\eta(s)\mathrm{d}\bar{X}(s)\right]=f(x).
\end{eqnarray}
Now notice that $\limsup\limits_{t\rightarrow\infty}X(t)=\infty$ almost surely (a.s.) implies $\lim\limits_{t\rightarrow\infty}\bar{X}(t)=\infty$ a.s., which implies further $\lim\limits_{t\rightarrow\infty}\tau_{\overline{\gamma}_{x}\left(\bar{X}(t)\right)}^{+}=\infty$ a.s..
Combining the facts that $f(x)$ is bounded over $[0,\infty)$ (see, Remark \ref{rem1}) and $\lim\limits_{t\rightarrow\infty}\tau_{\overline{\gamma}_{x}\left(\bar{X}(t)\right)}^{+}=\infty$ a.s., using the bounded convergence theorem we get
\begin{eqnarray}
\lim\limits_{t\rightarrow\infty}\mathbb{E}_{x}\left(\mathrm{e}^{-q\tau_{_{\overline{\gamma}_{x}\left(\bar{X}(t)\right)}}^{+}}
\mathbf{1}_{\{\tau_{_{\overline{\gamma}_{x}\left(\bar{X}(t)\right)}}^{+}
<\tau^{\gamma}_{\xi}\}}
f\left(\overline{\gamma}_{x}\left(\bar{X}(t)\right)\right)\right)=0.
\end{eqnarray}
Letting $t\rightarrow\infty$ in (\ref{3.22}) and using the monotone convergence theorem we arrive at
\begin{eqnarray}\label{3.23}
\mathbb{E}_{x}\left[
-\int_{0}^{\infty}\eta(s)\mathrm{d}\bar{X}(s)\right]= f(x).
\end{eqnarray}

Since (\ref{less}) becomes equality when $\gamma=\gamma^{*}$, by (\ref{3.23}) we obtain
\begin{eqnarray}\label{less1}
\mathbb{E}_{x}\left[
\int_{0}^{\infty}\mathrm{e}^{-q\tau_{\bar{\gamma^{*}}\left(\bar{X}(s)\right)}^{+}}
\mathbf{1}_{\{\tau_{\bar{\gamma^{*}}\left(\bar{X}(s)\right)}^{+}
<\tau^{\gamma^{*}}_{\xi}\}}
\gamma^{*}
\left(\Bar{X}(s)\right)\mathrm{d}\bar{X}(s)\right]= f(x).
\end{eqnarray}
%  Now notice that $\limsup\limits_{t\rightarrow\infty}X(t)=\infty$ almost surely (a.s.) implies $\lim\limits_{t\rightarrow\infty}\bar{X}(t)=\infty$ a.s., which again implies
%  $\lim\limits_{t\rightarrow\infty}\tau_{\overline{\gamma}_{x}\left(\bar{X}(t)\right)}^{+}=\infty$ a.s..
%  Letting $t\rightarrow\infty$ in (\ref{less1}) and using the monotone convergence theorem we arrive at,
%  \begin{eqnarray}
%  \mathbb{E}_{x}\left[
%  \int_{0}^{\infty}\mathrm{e}^{-q\tau_{\overline{\gamma}_{x}\left(\bar{X}(s)\right)}^{+}}
%  \mathbf{1}_{\{\tau_{\overline{\gamma}_{x}\left(\bar{X}(s)\right)}^{+}
%  <\tau^{\gamma}_{\xi}\}}
%  \gamma
%  \left(\Bar{X}(s)\right)\mathrm{d}\bar{X}(s)\right]\leq f(x).
%  \end{eqnarray}
Because tax would be paid when and only when $\bar{X}(s)$ genuinely increase (intuitively, when $"\mathrm{d}\bar{X}(s)>0"$ or $s\in\{\tau_{\bar{\gamma^{*}}\left(\bar{X}(s)\right)}^{+};s\geq0\}$), we have
\begin{eqnarray}
\mathbb{E}_{x}\left[
\int_{0}^{\infty}\mathrm{e}^{-qs}
\mathbf{1}_{\{s
<\tau^{\gamma^{*}}_{\xi}\}}
\gamma^{*}
\left(\Bar{X}(s)\right)\mathrm{d}\bar{X}(s)\right]\hspace{-0.05cm}=\hspace{-0.05cm}\mathbb{E}_{x}\left[
\int_{0}^{\infty}\mathrm{e}^{-q\tau_{\bar{\gamma^{*}}\left(\bar{X}(s)\right)}^{+}}
\mathbf{1}_{\{\tau_{\bar{\gamma^{*}}\left(\bar{X}(s)\right)}^{+}
<\tau^{\gamma^{*}}_{\xi}\}}
\gamma^{*}
\left(\Bar{X}(s)\right)\mathrm{d}\bar{X}(s)\right]\hspace{-0.05cm}=\hspace{-0.05cm} f(x),\nonumber
\end{eqnarray}
that is $ f_{\gamma^{*}}(x)= f(x)$.
%  Combining the facts that $f(x)$ is bounded over $[0,\infty)$ (see, Proposition \ref{expected discounted tax}) and
%  $\lim\limits_{t\rightarrow\infty}\tau_{\overline{\gamma}_{x}\left(\bar{X}(t)\right)}^{+}=\infty$ a.s., using the bounded convergence theorem we get
%  \begin{eqnarray}
%  \lim\limits_{t\rightarrow\infty}\mathbb{E}_{x}\left(\mathrm{e}^{-q\tau_{_{\overline{\gamma}_{x}\left(\bar{X}(t)\right)}}^{+}}
%  \mathbf{1}_{\{\tau_{_{\overline{\gamma}_{x}\left(\bar{X}(t)\right)}}^{+}
%  <\tau^{\gamma}_{\xi}\}}
%  f\left(\overline{\gamma}_{x}\left(\bar{X}(t)\right)\right)\right)=0.
%  \end{eqnarray}
%  Thus, considering $\gamma^{\ast}$ we can follow the same proof as above in which all inequalities are replaced by
%  equalities, and we thereby obtain
%  \begin{eqnarray}
%  \mathbb{E}_{x}\left(
%  \int_{0}^{\tau^{\gamma^{\ast}}_{\xi}}\mathrm{e}^{-qs}\gamma^{\ast}\left(\bar{X}(s)\right)\mathrm{d}\bar{X}(s)\right)=f(x).
%  \end{eqnarray}
%  This implies that $f(x)=f_{\gamma^{*}}(x)\leq\sup\limits_{\gamma\in\Gamma}f_{\gamma}(x)$. Therefor $f(x)=\sup\limits_{\gamma\in\Gamma}f_{\gamma}(x)$. The proof of Proposition \ref{verification pro} is completed.

While, for arbitrary $\gamma\in\Gamma$, by (\ref{less}) and similar arguments one can obtain
\begin{eqnarray}
\mathbb{E}_{x}\left[
\int_{0}^{\infty}\mathrm{e}^{-qs}
\mathbf{1}_{\{s
<\tau^{\gamma}_{\xi}\}}
\gamma
\left(\bar{X}(s)\right)\mathrm{d}\bar{X}(s)\right]\leq f(x),\nonumber
\end{eqnarray}
that is $ f_{\gamma}(x)\leq f(x)$ for all $\gamma\in\Gamma$.
\end{proof}

\section*{Acknowledgements}
The authors are grateful to the anonymous referees for their very helpful comments.
%The first %author, Wenyuan Wang, acknowledges the financial support from the National Natural Science %Foundation of China (No.11601197), the Program for New Century Excellent Talents in Fujian Province %University (No.Z0210103) and Fundamental Research Funds for the Central Universities %(No.20720170096).
%The second author, Zhimin Zhang, acknowledges the financial support
%from the National Natural Science Foundation of China (11471058, 11101451), MOE (Ministry of %Education in China) Project of Humanities
%and Social Sciences (16YJC910005) and Fundamental Research Funds for the Central Universities %(106112017CDJZRPY0201).

\vspace{-.3in}

\end{document}